\documentclass[a4paper,11pt]{article}
\input{Preamble.tex}

\begin{document}
\hypersetup{pageanchor=false}
%\title{Improved Parameterized Algorithm and Approximation Scheme for Treewidth}
\title{An Improved Parameterized Algorithm for Treewidth\thanks{The research leading to these results has received funding from the Research Council of Norway via the project
BWCA (grant no. 314528).}}

\author{
Tuukka Korhonen\thanks{Department of Informatics, University of Bergen, Norway. \texttt{tuukka.korhonen@uib.no}}
\and
Daniel Lokshtanov\thanks{Department of Computer Science, University of California Santa Barbara, USA. \texttt{daniello@ucsb.edu}}
}

\maketitle

\thispagestyle{empty}

\begin{abstract}
We give an algorithm 
that takes as input an $n$-vertex graph $G$ and an integer $k$,
runs in time $2^{\OO(k^2)} n^{\OO(1)}$,
and outputs a tree decomposition of $G$ of width at most $k$, if such a decomposition exists. 
This resolves the long-standing open problem of whether there is a $2^{o(k^3)} n^{\OO(1)}$ time algorithm for treewidth.
In particular, our algorithm is the first improvement on the dependency on $k$ 
in algorithms for treewidth 
since the $2^{\OO(k^3)} n^{\OO(1)}$ time algorithm given by Bodlaender and Kloks [ICALP 1991] and Lagergren and Arnborg [ICALP 1991].

We also give an algorithm that given an $n$-vertex graph $G$, an integer $k$, and a rational $\varepsilon \in (0,1)$, in time $k^{\OO(k/\varepsilon)} n^{\OO(1)}$ either outputs a tree decomposition of $G$ of width at most $(1+\varepsilon)k$ or determines that the treewidth of $G$ is larger than $k$.
Prior to our work, no approximation algorithms for treewidth with approximation ratio less than $2$, other than the exact algorithms, were known. Both of our algorithms work in polynomial space.
\end{abstract}

\newpage
\hypersetup{pageanchor=true}
\pagestyle{plain}
\pagenumbering{arabic}

%!tex root=main.tex
\section{Introduction}
A {\em tree decomposition} of a graph $G$ is a pair $(T, \bag)$ where $T$ is a tree and $\bag$
% : V(T) \rightarrow 2^{V(G)}$ 
is a function assigning to each node $t$ of $T$ a set $\bag(t)$ (called a {\em bag}) of vertices of $G$. 
%The vertex sets $\{\bag(u) ~:~ u \in V(T)\}$ are called the {\em bags} of the tree decomposition, and 
The function $\bag$ must satisfy the tree decomposition axioms: 
{\em (i)} for every edge $uv$ of $G$ at least one bag $\bag(t)$ contains both $u$ and $v$, and
{\em (ii)} for every vertex $v$ of $G$, the set $\{t \in V(T) \mid v \in \bag(t)\}$ induces a {\em non-empty} and {\em connected} subtree of $T$.
The {\em width} of a tree decomposition $(T, \bag)$ is the size of a largest bag minus one, and the {\em treewidth} of a graph $G$ is the minimum width of a tree decomposition of $G$.
The treewidth of a graph $G$ measures, in some sense, how far away $G$ is from being a tree. The treewidth of $G$ is at most $1$ if and only if every connected component of $G$ is a tree%(this result is probably the reason behind the ``minus one'' in the definition of width of a tree decomposition)
, while the treewidth of a complete graph on $n$ vertices is equal to $n-1$~\cite{Diestel}.

Treewidth and tree decompositions play a central role in graph theory and graph algorithms, and the concept has been independently rediscovered several times~\cite{BerteleF72,Halin:1976it,RobertsonS3} under different names in different contexts. 
It is a key tool in the celebrated Graph Minors project of Robertson and Seymour~\cite{RobertsonS3,RobertsonS-GMXIII,RobertsonS04}.
Many problems that are intractable on general graphs are solvable in linear time if a tree decomposition of the input graph $G$ of constant width is provided as an input (see e.g.~\cite{Bodlaender93} and references within). Indeed, the classic Courcelle's Theorem~\cite{Courcelle90} states that such an algorithm exists for every problem expressible in Monadic Second Order Logic (see also~\cite{BoriePT92}).

Therefore it should not come as a surprise that a significant amount of attention has been devoted to designing algorithms to determine, given as input a graph $G$ and an integer $k$, whether the treewidth of $G$ is at most $k$ (and to produce a tree decomposition of width at most $k$ in the ``yes'' case). 
%compute the treewidth of an input graph $G$ (and the corresponding tree decomposition). 
This problem is known to be \textsf{NP}-complete~\cite{ArnborgCP87},
%The associated decision problem is \textsf{NP}-complete, 
however, in many settings tree decompositions are only relevant if the treewidth of the input graph is sufficiently small, directing research towards algorithms with running times of the form $f(k) \cdot n^{g(k)}$ or $f(k) \cdot n^{O(1)}$.
Algorithms with running time of the first type are called {\em slicewise polynomial}, since they run in polynomial time when $k$ is considered a constant. Algorithms of the second type are called {\em fixed-parameter tractable (FPT)} as they run in polynomial time if $k$ is considered a constant, and furthermore the exponent of the polynomial remains the same for different values of $k$. We refer to the textbooks~\cite{cygan2015parameterized, DowneyFbook13, FlumGrohebook, Niedermeierbook06} for an introduction to parameterized algorithms. 

The first slicewise polynomial algorithm for treewidth was given by Arnborg, Corneil and Proskurowski~\cite{ArnborgCP87}, with running time $O(n^{k+2})$.
%\todo{i'm ignoring poly k factors here, although its likely n choose k so it should be fine.}
%
Subsequently, Robertson and Seymour~\cite{RobertsonS-GMXIII}, gave a non-constructive (see Bodlaender~\cite{DBLP:journals/dam/Bodlaender94} for a discussion of the non-constructive nature of~\cite{RobertsonS-GMXIII}) $f(k)n^2$ time algorithm for treewidth, 
and Bodlaender~\cite{DBLP:journals/dam/Bodlaender94}, building on work of Fellows and Langston~\cite{FellowsL89} made this algorithm constructive. The function $f$ in the running time both of the algorithm of Robertson and Seymour~\cite{RobertsonS-GMXIII} and of Bodlaender~\cite{DBLP:journals/dam/Bodlaender94} is unspecified and was not even known to be computable at the time of publication. 

The algorithm of Robertson and Seymour~\cite{RobertsonS-GMXIII} follows a ``two-step'' approach. In the first step they compute a tree decomposition of $G$ of width at most $4k+3$ in time $O(3^{3k}n^2)$, or conclude that the treewidth of $G$ is more than $k$.
In the second step they do dynamic programming over the tree decomposition found in the first step.
The second step is the only non-constructive part of their algorithm, and runs in time $f(k)n$ where the function $f$ is unspecified.

% in time $f(k)n$ to determine whether the treewidth of $G$ is at most $k$, this is the non-constructive part of their algorithm. 
%Lagergren~\cite{} gave a $\OO(\log^3 n)$ time parallell algorithm for 

Matou\v{s}ek and Thomas~\cite{DBLP:journals/jal/MatousekT91}, Lagergren~\cite{DBLP:journals/jal/Lagergren96}, and Reed~\cite{DBLP:conf/stoc/Reed92} gave improved algorithms for the first step.
The algorithms of Matou\v{s}ek and Thomas and Lagergren run in time $k^{\OO(k)} n \log^2 n$, and the algorithm of Reed runs in time $k^{\OO(k)} n \log n$.
%, with running times $k^{\OO(k)} n \log^2 n$, $k^{\OO(k)} n \log^2 n$, and $k^{\OO(k)} n \log n$, respectively.
%Matou\v{s}ek and Thomas~\cite{DBLP:journals/jal/MatousekT91} gave a $k^{\OO(k^2)} n \log^2 n$ time algorithm, Lagergren~\cite{DBLP:journals/jal/Lagergren96} gave $f(k) n \log^2 n$ time algorithms for the first step, and
% $\OO(\log^3 n)$ time parallell algorithm with $O(n)$ processors,   gave a randomized $\OO(n \log^2 n)$ time algorithm, while 
%Reed~\cite{DBLP:conf/stoc/Reed92} gave a $k^{\OO(k)} n \log n$ time algorithm.
All three algorithms either conclude that the treewidth of $G$ is more than $k$, or produce a tree decomposition of width at most $O(k)$.
%The algorithm of Matou\v{s}ek and Thomas~\cite{DBLP:journals/jal/MatousekT91} is randomized and 
The algorithm of Lagergren~\cite{DBLP:journals/jal/Lagergren96} is given as a parallel algorithm with $k^{\OO(k)} \log^3 n$ running time on $\OO(k^2 n)$ processors.
%
%The algorithm of Lagergren~\cite{} is actually a parallel algorithm
%with the algorithm of Reed~\cite{}) 

For the second step, constructive $2^{\OO(k^3)} n$ time dynamic programming algorithms were discovered in 1991 independently by Lagergren and Arnborg~\cite{DBLP:conf/icalp/LagergrenA91}, and Bodlaender and Kloks~\cite{DBLP:conf/icalp/BodlaenderK91,BodlaenderK96}. None of \cite{DBLP:conf/icalp/LagergrenA91,DBLP:conf/icalp/BodlaenderK91,BodlaenderK96} explicitly mention the dependence on $k$, but the $2^{\OO(k^3)}$ bound directly follows from the analysis in~\cite{BodlaenderK96} and is mentioned in~\cite{Bodlaender96}. Combined with the algorithm for the first step by Lagergren~\cite{DBLP:journals/jal/Lagergren96}, this led to a  $2^{\OO(k^3)}n\log^2 n$  time algorithm for treewidth. In 1993, Bodlaender showed that the first phase of the algorithms can be replaced by an ingenious recursion scheme, and designed a linear $2^{\OO(k^3)} n$ time algorithm for treewidth~\cite{DBLP:conf/stoc/Bodlaender93,Bodlaender96}. Much more recently, Elberfeld, Jakoby, and Tantau~\cite{DBLP:conf/focs/ElberfeldJT10} gave an algorithm for treewidth that uses {\em space} $f(k)\log n$ and time $n^{f(k)}$.

Downey and Fellows asked in their monograph from 1999 whether the dependence on $k$ in Bodlaender's algorithm could be improved from $2^{\OO(k^3)}$ to $2^{\OO(k)}$~\cite[Chapter~6.3]{DowneyF99}.
Later, in 2006, Telle~\cite[Problem~2.7.1]{bodlaender2006open} asked the less ambitious question of whether there is any fixed-parameter algorithm for treewidth whose running time as a function of $k$ is better than~$2^{\OO(k^3)}$.
The problem of obtaining a $2^{o(k^3)} n^{\OO(1)}$ time algorithm was also asked by Bodlaender, Drange, Dregi, Fomin, Lokshtanov, and Pilipczuk~\cite{BodlanderDDFLP13} and called a ``long-standing open problem'' by Bodlaender, Jaffke, and Telle~\cite{bodlaenderrevisited}.
In this paper, we resolve this problem.

\begin{theorem}
\label{the:mainexact}
There is an algorithm that takes as input an $n$-vertex graph $G$ and an integer $k$, and in time $2^{\OO(k^2)} n^4$ either outputs a tree decomposition of $G$ of width at most $k$ or concludes that the treewidth of $G$ is larger than $k$.
Moreover, the algorithm works in space polynomial in $n$.
\end{theorem}

An interesting feature of our algorithm is that it runs in polynomial space, and in particular that it is not based on dynamic programming. All previously known parameterized algorithms for computing treewidth exactly~\cite{ArnborgCP87,DBLP:conf/icalp/LagergrenA91, RobertsonS-GMXIII, BodlaenderK96,Bodlaender96} are based on dynamic programming and use space exponential in $k$.
The running time dependence on $n$ of the algorithm of \Cref{the:mainexact} is significantly worse than that of Bodlaender~\cite{Bodlaender96}.  The dependence on $n$ of our algorithm can probably be improved, nevertheless we believe that  an algorithm with running time  $2^{\OO(k^2)} n^{2}$ or better should require new and interesting ideas. 

Our second contribution is a new parameterized approximation algorithm for treewidth. 

%In the current article we consistently chose ease of presentation over any attempt at optimizing the dependence on $n$. We strongly believe that the dependence on $n$ of the current algorithm can be improved to $n^4$, while an algorithm with running time  $2^{\OO(k^2)} n^{2}$ or better probably should require new and interesting ideas. 

\begin{theorem}\label{the:mainapx}
There is an algorithm that takes as input an $n$-vertex graph $G$, an integer $k$, and a rational $\varepsilon \in (0,1)$, and in time $k^{\OO(k/\varepsilon)} n^{4}$ either outputs a tree decomposition of $G$ of width at most $(1+\varepsilon)k$ or concludes that the treewidth of $G$ is larger than $k$.
Moreover, the algorithm works in space polynomial in $n$.
\end{theorem}

There is a rich history of approximation algorithms for treewidth. In terms of {\em polynomial time} approximation algorithms, the best known approximation algorithm~\cite{FeigeHL08} by Feige, Hajiaghayi and Lee has approximation factor $\OO(\sqrt{\log{k}})$, improving upon a $\OO(\log n)$-approximation algorithm~\cite{BodlaenderGHK95} and a $\OO(\log k)$-approximation algorithm~\cite{Amir01}. On the other hand, Wu, Austrin, Pitassi and Liu~\cite{WuAPL14} showed that assuming the Small Set Expansion Conjecture (and \textsf{P $\neq$ NP}), there is no constant factor approximation algorithm for treewidth. 

Treewidth is one of the unusual cases where the first FPT-approximation algorithm (an approximation algorithm with running time $f(k)n^{O(1)})$ pre-dates the first polynomial time approximation algorithm.  The first such algorithm, a $4$-approximation algorithm running in time  $\OO(3^{3k}n^2)$,  is the ``first step'' of the $f(k)n^2$ time non-constructive algorithm by Robertson and Seymour for exactly computing treewidth~\cite{RobertsonS-GMXIII}. Subsequent research, summarized in \Cref{tab:history} attained different trade-offs between the running time dependence on $n$, the running time dependence $f(k)$  on $k$, and the approximation factor. The algorithm of \Cref{the:mainapx} is the first FPT-approximation algorithm for treewidth with approximation ratio below $2$ and running time $2^{o(k^2)}n^{O(1)}$ (or even  $2^{o(k^3)}n^{O(1)}$, discounting \Cref{the:mainexact}).
Note that by setting $\varepsilon = \frac{1}{k+1}$, the algorithm of \Cref{the:mainapx} gives an exact algorithm with only a slightly slower ($k^{\OO(k^2)} n^4$) running time than the algorithm of \Cref{the:mainexact}, in particular, being sufficient for resolving the open problem of obtaining a $2^{o(k^3)} n^{\OO(1)}$ time algorithm for treewidth.
This is worth noting since the algorithm of \Cref{the:mainapx} is in fact considerably simpler than the algorithm of \Cref{the:mainexact}.

% (except for the exact algorithms, which have running time exponential in $n$ or $\OO(k^3)$, $\OO(k^2)$ if we include the algorithm of Theorem~\ref{the:mainexact} .

\begin{table*}[t]
\centering
\begin{tabular}{|c | c | c | c|}
\hline
Reference & Appx. $\alpha(k)$ & $f(k)$ & $g(n)$\\
\hline
Arnborg, Corneil, and Proskurowski~\cite{ArnborgCP87} & exact & $\OO(1)$ & $n^{k+2}$\\
Robertson and Seymour~\cite{RobertsonS-GMXIII} & $4k + 3$ & $\OO(3^{3k})$ & $n^2$\\
Matou\v{s}ek and Thomas~\cite{DBLP:journals/jal/MatousekT91} & $6k+5$ & $k^{\OO(k)}$ & $n \log^2 n$\\
Lagergren~\cite{DBLP:journals/jal/Lagergren96} & $8k+7$ & $k^{\OO(k)}$ & $n \log^2 n$\\
Reed~\cite{DBLP:conf/stoc/Reed92} & $8k + \OO(1)$ & $k^{\OO(k)}$ & $n \log n$\\
Bodlaender~\cite{Bodlaender96} & exact & $2^{\OO(k^3)}$ & n\\
Amir~\cite{DBLP:journals/algorithmica/Amir10} & $4.5k$ & $\OO(2^{3k} k^{3/2})$ & $n^2$\\
Amir~\cite{DBLP:journals/algorithmica/Amir10} & $(3 + 2/3)k$ & $\OO(2^{3.6982k} k^3)$ & $n^2$\\
Amir~\cite{DBLP:journals/algorithmica/Amir10} & $\OO(k \log k)$ & $\OO(k \log k)$ & $n^4$\\
Feige, Hajiaghayi, and Lee~\cite{FeigeHL08} & $\OO(k \sqrt{\log k})$ & $\OO(1)$ & $n^{\OO(1)}$\\
Fomin, Todinca, and Villanger~\cite{DBLP:journals/siamcomp/FominTV15} & exact & $\OO(1)$ & $1.7347^n$\\
Fomin et al.~\cite{DBLP:journals/talg/FominLSPW18} & $\OO(k^2)$ & $\OO(k^7)$ & $n \log n$\\
Bodlaender et al.~\cite{BodlanderDDFLP13} & $3k +4$ & $2^{\OO(k)}$ & $n \log n$\\
Bodlaender et al.~\cite{BodlanderDDFLP13} & $5k+4$ & $2^{\OO(k)}$ & $n$\\
Korhonen~\cite{Korhonen21} & $2k+1$ & $2^{\OO(k)}$ & $n$\\
Belbasi and F\"{u}rer~\cite{BelbasiF22}  & $5k+4$ & $2^{7.61k}$ & $n\log n$\\
Belbasi and F\"{u}rer~\cite{BelbasiF21} & $5k+4$ & $2^{6.755k}$ & $n\log n$\\
This paper & exact & $2^{\OO(k^2)}$ & $n^4$\\
This paper & $(1 + \varepsilon)k$ & $k^{\OO(k/\varepsilon)}$ & $n^4$\\
\hline
\end{tabular}
\caption{Overview of treewidth algorithms with running time $f(k) \cdot g(n)$, each either outputting a tree decomposition of width at most $\alpha(k)$ or determining that the treewidth of the input graph is larger than $k$. Most of the rows are based on a similar tables in~\cite{BodlanderDDFLP13} and~\cite{Korhonen21}.}
\label{tab:history}
\end{table*}

\paragraph{Methods.}
\smallskip
Both the exact algorithm of \Cref{the:mainexact} and the approximation algorithm of \Cref{the:mainapx} are based on a generalization of the local improvement method introduced by Korhonen~\cite{Korhonen21}, which in turn was inspired by a proof of Bellenbaum and Diestel~\cite{bellenbaum2002two}. 
%\smallskip
In each local improvement step we are given a tree decomposition $(T,\bag)$ of $G$ of width more than $k$, and the goal is to either
conclude that the treewidth of $G$ is more than $k$, 
or to find a ``better'' tree decomposition of $G$. Here {\em better} means that either the width of the output tree decomposition is strictly smaller than that of  $(T,\bag)$, or that the width of the output decomposition is the same as the width of  $(T,\bag)$, but there are fewer bags of maximum size. 
%smaller width than $(T,\bag)$, or with the same width as $(T,\bag)$ but with fewer bags of maximum size. 
%\smallskip
%

We show that the local improvement step is in fact {\em equivalent} to solving the following problem, which we call {\sc Subset Treewidth}: given as input a graph $G$ and a set $W$ of vertices, 
%a tree decomposition $(T,\bag)$ of $G$, and a maximum size bag $W$ of $(T,\bag)$,
conclude that the treewidth of $G$ is at least $|W|-1$, 
or find a tree decomposition $(T', \bag')$ such that $W$ is contained in the union of the non-leaf bags of $(T', \bag')$ and all non-leaf bags have size at most $|W|-1$ (the formal definition of this problem in~\Cref{sec:overview} is worded differently, but can easily be seen to be equivalent).
Observe here that if the treewidth of $G$ is strictly less than $|W|-1$ then every tree decomposition  $(T', \bag')$ of $G$ of width at most $|W|-2$ is a valid output for {\sc Subset Treewidth} (after possibly adding empty dummy leaf bags).

The first key insight behind our algorithms is that if $W$ is a maximum size bag of the tree decomposition $(T,\bag)$, and an algorithm for {\sc Subset Treewidth} on input $(G,W)$ outputs the decomposition $(T', \bag')$, then a tree decomposition better than $(T,\bag)$ (in the sense above) can be computed from $(T,\bag)$ and $(T', \bag')$ in polynomial time. 
The proof of this statement is given in~\Cref{sec:redu} and is a non-trivial generalization of corresponding improvement arguments by Bellenbaum and Diestel~\cite{bellenbaum2002two} and Korhonen~\cite{Korhonen21}.
Indeed, in retrospect, the $2$-approximation algorithm of Korhonen~\cite{Korhonen21} can be thought of as using this approach
%solving {\sc Subset Treewidth} 
with the additional assumption that $|W| \geq 2k+3$, where $k$ is the treewidth of $G$, and in this case a solution to {\sc Subset Treewidth}  whose non-leaf bags form the star $K_{1,3}$ exists.
%a tree decomposition which is not lean (see~\cite{bellenbaum2002two}), is precisely a tree decomposition which has a bag $W$ for which the {\sc Subset Treewidth} problem admits a solution $(T', \bag')$ where $T'$ has precisely two internal vertices. The factor $2$ approximation algorithm of Korhonen~\cite{Korhonen21} can be thought of as solving {\sc Subset Treewidth} with the additional assumption that $|W| \geq 2k+3$, where $k$ is the treewidth of $G$. 
The exact algorithm of \Cref{the:mainexact} is based on solving  {\sc Subset Treewidth} without any additional assumptions, while the approximation algorithm of~\Cref{the:mainapx} is based on solving {\sc Subset Treewidth} with the additional assumption that $|W| \ge k(1+\varepsilon) + 2$.

The second key insight is that the {\sc Subset Treewidth} problem is more approachable than the treewidth problem, because the problem formulation allows us to focus on one small set $W$ and ``discard'' all parts of the graph (by placing them into leaves of $(T', \bag')$) that are not relevant for providing connectivity between vertices of $W$. 
Both the algorithm of \Cref{the:mainexact} and of \Cref{the:mainapx} are based on branching on important separators (see e.g. \cite[Chapter 8]{cygan2015parameterized}),
a carefully chosen measure to quantify the progress made by the algorithms, 
and a ``safe separation'' reduction rule for the  {\sc Subset Treewidth} problem.
This rule states that if the algorithm has identified two vertex sets $B_1$ and $B_2$ that can be chosen as bags of $(T', \bag')$, 
and $S$ is a minimum size $(B_1,B_2)$-separator, then it is safe to also make $S$ a bag of $(T', \bag')$ and recurse on the connected components of $G \setminus S$.
A generalization of this reduction rule was formulated for the treewidth problem by Bodlaender and Koster~\cite[Lemma 11]{BodlaenderK06}. However it is not clear how to utilize this reduction rule to directly obtain efficient algorithms for treewidth. On the other hand, for {\sc Subset Treewidth}, this reduction rule is the main engine of our algorithms.

\paragraph{Organization.}
\smallskip
The rest of the paper is organized as follows.
In \Cref{sec:overview} we formally define the \stw problem, give statements of intermediate theorems on how \Cref{the:mainexact,the:mainapx} follow from algorithms for \stw, and then present an overview of the proofs.
In \Cref{sec:preli} we present notation and preliminary results.
In \Cref{sec:toolbox} we give results about important separators and a ``pulling lemma'' for tree decompositions, which will be used for our algorithms.
In \Cref{sec:redu} we show that algorithms for \stw imply algorithms for treewidth.
Then, in \Cref{sec:algopstw} we give the algorithm for \stw that implies \Cref{the:mainapx}, and in \Cref{sec:fasttw} we give the algorithm that implies \Cref{the:mainexact}.
The algorithms of \Cref{sec:algopstw,sec:fasttw} are presented in this order because the algorithm of \Cref{sec:fasttw} builds upon the algorithm of \Cref{sec:algopstw} and is more involved.
Finally, we conclude in \Cref{sec:conclusion}.

\section{Overview}\label{sec:overview}
%\todo[inline]{TODO the whole section}
In this section we state the main intermediate theorems leading into \Cref{the:mainexact,the:mainapx} and overview the proofs of them.
The proofs of \Cref{the:mainexact} and of \Cref{the:mainapx} neatly split in two parts. The first part is common to the proofs of  \Cref{the:mainexact} and \Cref{the:mainapx}, while the second part requires separate proofs.
The first and common part is the overall scheme of the algorithms, namely that we proceed by ``local improvement''. Each local improvement step is reduced to another problem, which we call {\sc Subset Treewidth}.
In the second part we give two different algorithms for the {\sc Subset Treewidth} problem, one exact, leading to a proof of \Cref{the:mainexact}, and one approximate, leading to a proof of \Cref{the:mainapx}. We start by discussing the first part.

\subsection{Reduction to Subset Treewidth}
\label{subsec:overredustw}
Suppose that we are given as input the graph $G$ and integer $\tau$, and the task is to either return that the treewidth of $G$ is more than $\tau$, or find a ``good enough'' tree decomposition of $G$. For an exact algorithm this simply means a tree decomposition of width at most $\tau$, for a $(1+\varepsilon)$-approximation algorithm this means a tree decomposition of width at most $\tau(1+\varepsilon)$.
Assume now that we are also given as input a tree decomposition $(T, \bag)$ of $G$ of width at most $\OO(\tau)$. Initially such a tree decomposition can be obtained by an approximation algorithm, such as the $4$-approximation algorithm of Robertson and Seymour~\cite{RobertsonS-GMXIII} with running time $\OO(3^{3\tau}n^2)$.
If the tree decomposition $(T, \bag)$  is already good enough, then we can output it and halt. Otherwise, a largest bag $W$ of  $(T, \bag)$ is too large. We would like to make $(T, \bag)$  better by getting rid of this bag $W$ that is too large. More formally we want to find a tree decomposition $(T'', \bag'')$ of $G$ of width at most $|W|-1$ and with strictly fewer bags of size $|W|$ than $(T, \bag)$ has. On the surface this does not really look any easier than trying to find a tree decomposition of width at most $|W|-2$. Somewhat miraculously it turns out that it is in fact easier, because this problem is equivalent to the {\sc Subset Treewidth} problem, which we will define shortly. To define the {\sc Subset Treewidth} problem we first need to introduce some notation.

Let $G$ be a graph and $X \subseteq V(G)$.
The graph $\torso_G(X)$ has vertices $V(\torso_G(X)) = X$ and has $uv \in E(\torso_G(X))$ if $u,v \in X$ and there is a path from $u$ to $v$ whose all internal vertices (if any) are in $V(G) \setminus X$.
In particular, note that $E(\torso_G(X)) \supseteq E(G[X])$.
An equivalent definition of $\torso_G(X)$ is that it is the graph obtained from $G[X]$ by making $N_G(C)$ a clique for every connected component $C$ of $G \setminus X$.
A \emph{torso tree decomposition} in a graph $G$ is a pair $(X, (T, \bag))$, where $X \subseteq V(G)$ and $(T, \bag)$ is a tree decomposition of $\torso_G(X)$. The {\em width} of the torso tree decomposition $(X, (T, \bag))$ is simply the width of $(T, \bag)$. For a set $W \subseteq V(G)$, we say that $(X, (T, \bag))$ {\em covers} $W$ if $W \subseteq X$. We are now ready to define the {\sc Subset Treewidth} problem.

\smallskip
\defparproblem{\textsc{Subset Treewidth}}{Graph $G$, integer $k$, and a set of vertices $W$ of size $|W|=k+2$.}{$k$}{Return a torso tree decomposition of width at most $k$ in $G$ that covers $W$ or conclude that the treewidth of $G$ is at least $k+1$.}

Note that at least one of the two cases in the definition of {\sc Subset Treewidth} must apply. In particular, if $G$ has a tree decomposition $(T', \bag')$ of width at most $k$ then $(V(G), (T', \bag'))$ is a torso tree decomposition of width at most $k$ in $G$ that covers $W$.
The two cases need not be mutually exclusive: there exists graphs $G$ with treewidth at least $k+1$ and sets $W$ of size $k+2$ that nevertheless can be covered by a torso tree decomposition of width $k$.
%uch that  $\torso_G(W)$\todo{confusing?} has a tree decomposition of width at most $k$.
In such a case an algorithm for \textsc{Subset treewidth} may output either one of the two options. 
%and yet there exists a a torso tree decomposition of width at most $k$ in $G$ that covers $W$
%Further note that the two cases need not be mutually exclusive; if $G$ has two connected components, one of them is a large clique, the other is a path on three vertices, and $W$ is the vertex set of the three vertex path then the treewidth of $G$ is at least $2$, but $\torso_G(W)$ has a tree decomposition of width $1$.

The {\sc Subset Treewidth} problem directly reduces to treewidth: using a hypothetical treewidth algorithm we can determine whether the treewidth of $G$ is at most $k$. If no, then report that the treewidth of $G$ is at least $k+1$. Otherwise output $(V(G), (T', \bag'))$ where $(T', \bag')$ is the width $k$ tree decomposition returned by the treewidth algorithm. 
%If yes, output $(V(G), (T, \bag))$. Otherwise report that the treewidth of $G$ is at least $k+1$. 
Our algorithms for treewidth are based on the result that we can reduce in the other direction as well. We encapsulate this insight in the following lemma.

\begin{lemma}\label{lem:improveTDSimplified}
Let $(T, \bag)$ be a tree decomposition of $G$ and $W$ be a largest bag of $(T, \bag)$.
If there exists a torso tree decomposition $(X, (T', \bag'))$ in $G$ that covers $W$ and has width at most $|W|-2$,
then there exists a tree decomposition $(T'', \bag'')$ of $G$ of width at most $|W|-1$ with strictly fewer bags of size $|W|$. Moreover, given $G$,  $(T, \bag)$ and $(X, (T', \bag'))$ we can compute $(T'', \bag'')$ in polynomial time. 
\end{lemma}

\Cref{lem:improveTDSimplified} is more carefully stated and proved as \Cref{lem:main_impr_lem} in \Cref{sec:redu}. Before giving a proof sketch of \Cref{lem:improveTDSimplified} in \Cref{sec:improveSketch}, we show how \Cref{lem:improveTDSimplified}, together with an algorithm (or approximation algorithm) for {\sc Subset Treewidth} yields an algorithm (or approximation algorithm) for treewidth. 

Indeed, starting with a tree decomposition $(T, \bag)$ of width $\OO(\tau)$ but more than $\tau$, we can call an algorithm for {\sc Subset Treewidth} on a largest bag $W$ of $(T, \bag)$, and either conclude that the treewidth of $G$ is more than $\tau$ or obtain a torso tree decomposition $(X, (T', \bag'))$ that covers $W$ and has width at most $|W|-2$. \Cref{lem:improveTDSimplified} now yields a tree decomposition  $(T'', \bag'')$ with no larger width and strictly fewer bags of size $|W|$. We now repeat the process with $(T'', \bag'')$ as the new $(T, \bag)$. After at most $\OO(\tau n)$ iterations we will either have obtained a tree decomposition of $G$ of width at most $\tau$ or concluded that the treewidth of $G$ is more than $\tau$. We now state this as a theorem. 
%Next we state a general theorem on how algorithms for subset treewidth imply algorithms for treewidth.
For the running times, we use $m = |V(G)|+|E(G)|$ to denote the size of the graph and we assume that the function $T(k)$ is increasing.

\begin{restatable}{theorem}{eximpltheorem}
\label{the:stweximpl}
Given an algorithm for \stw with running time $T(k) \cdot m^c$, an algorithm for treewidth with running time $T(\OO(k)) \cdot \OO((nk)^{c+1}) + k^{\OO(1)} n^4 + 2^{\OO(k)} n^2$ can be constructed.
Moreover, if the algorithm for \stw works in polynomial space, then the algorithm for treewidth works in polynomial space.
\end{restatable}

\Cref{the:stweximpl} is proved in \Cref{sec:redu}.
We remark that the additive $2^{\OO(k)} n^2$ term comes from starting by applying the factor $4$-approximation algorithm of Robertson and Seymour~\cite{RobertsonS-GMXIII}. This additive term could be avoided by replacing this approximation algorithm by Bodlaender's recursive compression technique~\cite{Bodlaender96}, at the expense of a $k^{\OO(1)}$ multiplicative factor in the running time. 
In light of \Cref{the:stweximpl} it is natural to focus on parameterized algorithms for {\sc Subset Treewidth}, which is precisely our line of attack. In \Cref{sec:fasttw} we give a $2^{\OO(k^2)} nm$ time polynomial space algorithm for {\sc Subset Treewidth}.

\begin{theorem}
\label{the:stwexalg}
There is a $2^{\OO(k^2)} nm$ time polynomial space algorithm for \stw.
\end{theorem}

A proof sketch for \Cref{the:stwexalg} is given in \Cref{subsec:branchoverview}. Putting \Cref{the:stweximpl,the:stwexalg} together implies \Cref{the:mainexact}.
The argument proving \Cref{the:stweximpl} (assuming \Cref{lem:improveTDSimplified}) also works for approximation algorithms. In particular the same argument shows that in order to obtain a $(1+\varepsilon)$-approximation algorithm for treewidth, it is sufficient to design an algorithm for {\sc Subset Treewidth} that is only required to work correctly on instances where $|W| \geq (1+\varepsilon)\tau + 2$ (and $\tau$ is the treewidth of $G$).
Towards designing such an algorithm we define an intermediate problem, called \textsc{Partitioned Subset Treewidth}. 
%This problem also shows up as an intermediate problem in the algorithm of \Cref{the:stwexalg}, but there is an interesting insight 

\defparsproblem{\textsc{Partitioned Subset Treewidth}}{Graph $G$, integer $k$, set of vertices $W$ of size $|W|=k+2$, and $t$ cliques $W_1,\ldots,W_t$ of $G$ such that $\bigcup_{i=1}^t W_t = W$.}{$k$, $t$}{Return a torso tree decomposition of width at most $k$ in $G$ that covers $W$ or conclude that the treewidth of $G$ is at least $k+1$.}

We remark that despite being called {\sc Partitioned Subset Treewidth}, the $t$ cliques $W_1,\ldots,W_t$ are not required to form a partition of $W$, but are allowed to overlap.
The {\sc Partitioned Subset Treewidth} problem arises naturally when designing a recursive branching algorithm for  {\sc Subset Treewidth}. Every instance of {\sc Subset Treewidth} is an instance of {\sc Partitioned Subset Treewidth} with $t=|W|=k+2$ and $W_i = \{w_i\}$ (where $W = \{w_1, w_2, \ldots, w_{|W|}\}$). However,  {\sc Partitioned Subset Treewidth} appears substantially easier when $t$ is much smaller than $k$. The following theorem, proved in \Cref{sec:algopstw}, formalizes this intuition.

\begin{theorem}
\label{the:stwpalg}
There is a $k^{\OO(kt)} nm$ time polynomial space algorithm for \pstw.
\end{theorem}

We give a proof sketch of \Cref{the:stwpalg} in \Cref{subsec:branchoverview}. In light of  \Cref{the:stwpalg} it is natural to ask whether it is possible to reduce {\sc Subset Treewidth} to {\sc Partitioned Subset Treewidth}  with $t$ much smaller than $k$. While we do not know of a way to do this for exact algorithms, we obtain such a reduction for the variant of {\sc Subset Treewidth} which is sufficient for $(1+\varepsilon)$-approximating treewidth. 
In particular, it is possible to show, using standard methods, that for every graph $G$, vertex set $W$ and positive integer $t$ there exists a partition of $W$ into $t$ sets $W_1, \ldots, W_t$ such that making all of $W_1, \ldots, W_t$ into cliques increases the treewidth of $G$ by at most $3\lceil|W|/t\rceil$ (we give a proof of essentially this fact with parameter values relevant to our applications in \Cref{lem:wpartapx}).
Thus, for instances of {\sc Subset Treewidth} where $\tau(1+\varepsilon) +2 \leq |W| = O(\tau)$ (and $\tau$ is the treewidth of $G$), setting $t = \OO(1/\varepsilon)$ shows that there exists a partition of $W$ into at most $t = \OO(1/\varepsilon)$ sets such that making all of $W_1, \ldots, W_t$ into cliques increases the treewidth of $G$ to at most $\tau(1 + \varepsilon)$. The proof of \Cref{the:stweximpl} (assuming \Cref{lem:improveTDSimplified}) coupled with this partitioning argument yields a reduction from $(1+\varepsilon)$-approximating treewidth to solving {\sc Partitioned Subset Treewidth} exactly with $t = \OO(1/\varepsilon)$.

\begin{restatable}{theorem}{apximpltheorem}
\label{the:stwapximpl}
Given an algorithm for \pstw with running time $T(k, t) \cdot m^c$, we can construct an $(1+\varepsilon)$-approximation algorithm for treewidth with running time $T(\OO(k), \OO(1/\varepsilon)) \cdot \OO((nk)^{c+1}) \cdot (1+1/\varepsilon)^{\OO(k)} + k^{\OO(1)} n^4 + 2^{\OO(k)} n^2$.
Moreover, if the algorithm for \pstw works in polynomial space, then the algorithm for treewidth works in polynomial space.
\end{restatable}

Putting \Cref{the:stwapximpl,the:stwpalg} together implies \Cref{the:mainapx}. We now give a proof sketch of the main engine behind the proofs of \Cref{the:stweximpl,the:stwapximpl}, namely \Cref{lem:improveTDSimplified}.

%\subsection{From Subset Treewidth to Partitioned Subset Treewidth}

\subsection{Proof sketch of \Cref{lem:improveTDSimplified}}\label{sec:improveSketch}
The proof of \Cref{lem:improveTDSimplified} proceeds as follows.
Given $G$, $(T, \bag)$, $W$ and $(X, (T', \bag'))$ as in the premise of \Cref{lem:improveTDSimplified} we will construct a tree decomposition $(T'', \bag'')$ of $G$ in polynomial time. 
We will always succeed in making $(T'', \bag'')$, but  $(T'', \bag'')$ might not be ``better'' than  $(T, \bag)$, in the sense that its width might be more than $|W|-1$, or it may have at least as many bags of size $|W|$ as $(T, \bag)$.
We will show that in this case we can ``improve'' $(X, (T', \bag'))$ instead! In particular we will find a torso tree decomposition $(X^\star, (T^\star, \bag^\star))$ in $G$ that covers $W$, with $|X^\star| < |X|$ and no larger width.
Since $G$, $(T, \bag)$, $W$ together with $(X^\star, (T^\star, \bag^\star))$ again satisfy the premise of \Cref{lem:improveTDSimplified} we can repeat the process with 
$(X^\star, (T^\star, \bag^\star))$ as the new $(X, (T', \bag'))$. Since $|X|$ cannot keep decreasing forever, eventually the tree decomposition  $(T'', \bag'')$ will satisfy the conclusion of the Lemma. 
Thus it remains to sketch {\em (i)} how we construct $(T'', \bag'')$ and {\em (ii)} how to improve $(X, (T', \bag'))$ when $(T'', \bag'')$ is not better than $(T, \bag)$. We start by describing the construction of $(T'', \bag'')$.

\paragraph{Constructing $(T'', \bag'')$.}
%We root $(T, \bag)$ at the vertex $r$ of $T$ corresponding to $W$, and we want to replace the root $r$ by the torso tree decomposition $(X, (T', \bag'))$. This is (at least somewhat) sensible because $W \subseteq X$. 
We root $(T, \bag)$ at the node $r$ of $T$ corresponding to $W$.
For each connected component $C$ of $G \setminus X$ we make a tree decomposition $(T_C, \bag_C)$ of $G[N[C]]$ as follows.
The decomposition tree $T_C$ is simply a copy of $T$. 
For each node $t$ of $T$ let $t_C$ be the copy of $t$ in $T_C$. Thus $r_C$ is the copy of the root $r$ of $T$ in $T_C$.
For every node $t_C$ of $T_C$ we set $\bag_C(t_C) = \bag(t) \cap N[C]$ {\em plus all vertices of $N(C)$ that appear in at least one bag below $t$ in $T$}.
In other words $(T_C, \bag_C)$ is simply the restriction of the tree decomposition $(T, \bag)$ to the vertex set $N[C]$, but additionally for every vertex $v$ of $N(C)$ we add $v$ to all bags on the path from $r_C$ to the subtree of $T_C$ where this vertex already occurs.
Observe that $C$ is disjoint from $X$, which contains $W$, and therefore $\bag_C(r_C) = N(C)$.

Now the tree decomposition $(T'', \bag'')$ consists of a copy of
$(T', \bag')$
together with the tree decomposition $(T_C, \bag_C)$ of $G[N[C]]$ for every connected component $C$ of $G \setminus X$. For each component $C$ of $G \setminus X$ we have that $N(C)$ is a clique in $\torso_G(X)$, and that therefore (see e.g.~\cite[Chapter 6]{cygan2015parameterized}) at least one bag of $(T', \bag')$ contains $N(C)$. We add an edge from $r_C$ to this bag. See \Cref{fig:construction} for a visualization of this construction.

It is quite easy to verify that $(T'', \bag'')$ is indeed a tree decomposition of $G$, and that it can be constructed in polynomial time.  However it is not at all obvious that it should be better than $(T, \bag)$ - it could even be worse, because we added vertices to the bags $\bag_C(t_C)$ that were not there in the corresponding bag $\bag(t)$ of  $(T, \bag)$. Thus, all that remains is to show how to improve $(X, (T', \bag'))$ when $(T'', \bag'')$ is not better than $(T, \bag)$. Working towards this goal we first state the main tool that we will use to improve $(X, (T', \bag'))$.

\begin{figure}[htb]
\begin{center}
\includegraphics[width=0.9\textwidth]{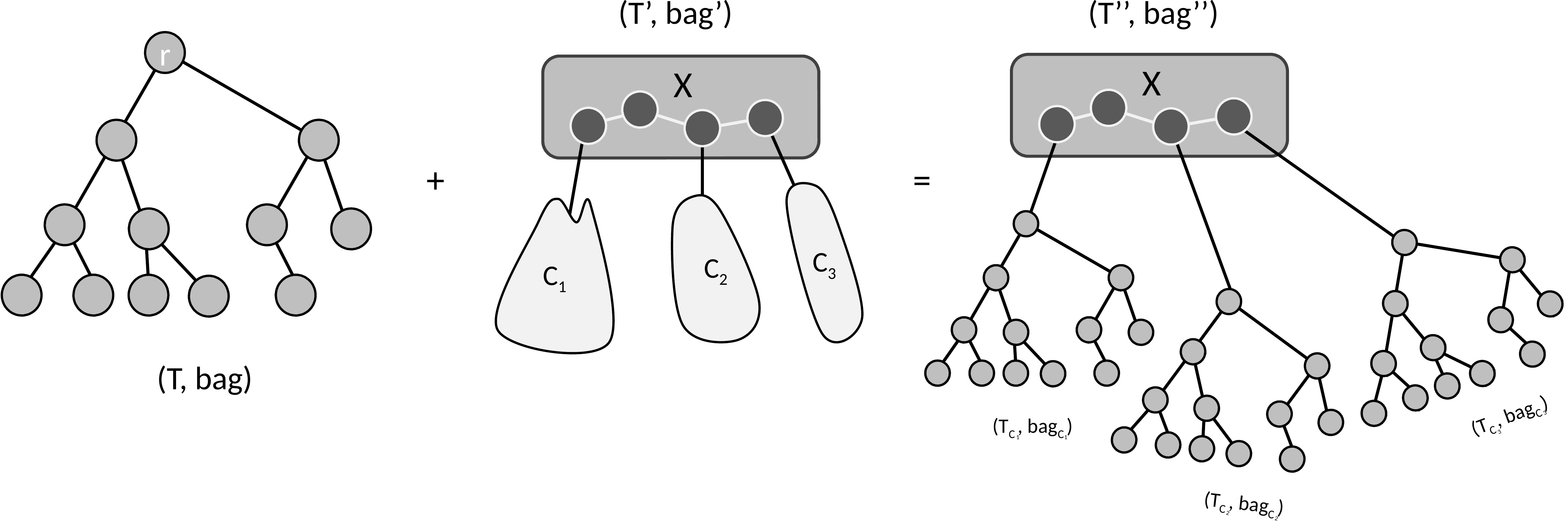}
\caption{Construction of $(T'', \bag'')$.}\label{fig:construction}
\end{center}
\end{figure}

%\todo{here a figure would have been great}

\paragraph{A pulling lemma.}
The following lemma is very useful to improve $(X, (T', \bag'))$, and also in many other arguments in this paper.  We call it the ``pulling lemma'' because in the proof the separator $S$ will be ``pulled'' along disjoint paths into a bag of the (torso) tree decomposition. To state the pulling lemma we need to define the notions of separations and linkedness. A separation in a graph $G$ is a partition of $V(G)$ into three parts $(A,S,B)$ such that no edge of $G$ has one endpoint in $A$ and the other in $B$. We call $S$ the {\em separator} of the separation $(A,S,B)$ and $|S|$ is the \emph{order} of the separation.
A vertex set $A$ is {\em linked} into a vertex set $B$ in $G$ if there exist $|A|$ vertex disjoint paths with one endpoint in $A$ and one in $B$ (paths on a single vertex that start and end in $A \cap B$ count).

\begin{lemma}[Pulling lemma, simplified variant]
\label{lem:pullSimple}
Let $G$ be a graph, $k$ be an integer, $(X, (T,\bag))$ be a torso tree decomposition in $G$ of width $k$,
$Z \subseteq X$ be a vertex set in $G$ that is a subset of at least one bag of $(T,\bag)$,
and $(A, S, B)$ be a separation of $G$ so that $Z \subseteq S \cup B$ and $S$ is linked into $Z$.
There exists a torso tree decomposition $((X \cap A) \cup S, (T', \bag'))$ of width at most $k$.
Moreover, $S$ is a bag of $(T', \bag')$.
Furthermore, when $G$, $(X, (T,\bag))$, $(A,S,B)$, and $Z$ are given as inputs, the torso tree decomposition $((X \cap A) \cup S, (T', \bag'))$ can be constructed in polynomial time.
\end{lemma}

The formal version of \Cref{lem:pullSimple} is stated and proved as \Cref{lem:pull} in \Cref{subsec:pull}. Analogous lemmas with similar proofs have been used in the context of tree decompositions, for example by~\cite{bellenbaum2002two,BodlaenderK06}, and so we do not give a sketch of \Cref{lem:pullSimple} in this overview. 

\paragraph{Improving $X$ when $(T'', \bag'')$ is not better than $(T, \bag)$.}
We now return our focus to the setting where we started with $(T, \bag)$, a maximum size bag $W$, and $(X, (T', \bag'))$, and we used them to make the new tree decomposition $(T'', \bag'')$ of $G$. If $(T'', \bag'')$ is better than $(T, \bag)$ (in the sense of having strictly fewer bags of size $|W|$ and no larger bags) then $(T, \bag)$ already satisfies the conclusion of \Cref{lem:improveTDSimplified}. Hence, assume that $(T'', \bag'')$ is not better than $(T, \bag)$.
Our goal is to find a component $C$ of $G \setminus X$ and a separation $(P', S', Q')$ such that
$N(C) \subseteq S' \cup Q'$,
$W \subseteq S' \cup P'$,
$S'$ is linked to $N(C)$,
and $|S'| < |N(C)|$.
Then $(X, (T', \bag'))$, $Z = N(C)$ and the separation $(P', S', Q')$ will satisfy the premise of \Cref{lem:pullSimple}.
Setting $X^\star = (X \cap P') \cup S'$, \Cref{lem:pullSimple} implies that there exists a torso tree decomposition $(X^\star, (T^\star, \bag^\star))$ of width at most that of $(X, (T', \bag'))$.
Moreover, because $N(C) \subseteq X$ and $N(C)$ is disjoint from $P'$, it holds that
$$|X^\star| = |X| - |X \setminus P'| + |S'| \leq |X| - |N(C)| + |S'| < |X|\mbox{.}$$
Since $W \subseteq S' \cup P'$ and $W \subseteq X$, we have that $X^\star$ covers $W$, and we have found our improved torso tree decomposition  $(X^\star, (T^\star, \bag^\star))$ that covers $W$. 
We now show how such a component $C$ and separation $(P', S', Q')$ can be identified. 
%This part of the overview is somewhat technical and may be skipped on a first reading. Removed, figure makes better. 

Note that $(T', \bag')$ has no bags of size at least $|W|$. Therefore every bag of size at least $|W|$ in $(T'', \bag'')$ appears in $(T_C, \bag_C)$ for some component $C$ of $G \setminus X$. Observe also that for the root $r$ of $T$ and every component $C$ of $G \setminus X$ we have $|\bag_C(r_C)| < |W| = |\bag(r)|$. Indeed, for every copy $r_C$ of the root we have that $C$ is disjoint from $X$, and that therefore $\bag_C(r_C) = N(C)$. But $N(C)$ is a subset of some bag of  $(T', \bag')$, all of which have size at most $|W|-1$. Therefore, since $(T'', \bag'')$ is not better than $(T, \bag)$ at least one of the two following statements must hold.
{\em (i)} There exists a node $t$ in $V(T)$ and component $C$ of $G \setminus X$ such that $|\bag_C(t_C)| > |\bag(t)|$, or
{\em (ii)} There exists a node $t$ in $V(T)$ and two distinct components $C_1, C_2$ of $G \setminus X$ such that $|\bag_{C_1}(t_{C_1})| = |\bag_{C_2}(t_{C_2})| = |\bag(t)| = |W|$.

We show how to improve $X$ in the first case. To this end, let $t$ be a node in $V(T)$ and $C$ be a component $C$ of $G \setminus X$ such that $|\bag_C(t_C)| > |\bag(t)|$. We consider two separations of $G$: $(C, N(C), R)$ (where $R = V(G) \setminus N[C])$ is the ``rest'' and $(U, B, L)$ where $B = \bag(t)$, $L$ (the ``lower'' set) is the set of all non-$B$ vertices appearing in bags of $T$ below $t$, and $U$ (the ``upper'') set is defined as $U = V(G) \setminus (B \cup L)$ (consult \Cref{fig:separation} for a visualization of these separations and how they are used in the remainder of the argument). 

We have that $W$ is a subset of $X$ and therefore disjoint from $C$.
Similarly, all vertices of $W$ appear in at least one bag above $t$ (namely $r$), and therefore $W$ is disjoint from $L$.
It follows that $S$ defined as
$$S = (N(C) \setminus L) \cup (B \cap R)$$ 
separates $N(C)$ from $W$.
Furthermore, by choice of $t$ and $B = \bag(t)$ we have that
\begin{align}\label{eqn:bagCompare}
|B \cap N[C]| + |B \setminus N[C]| = |B| < |\bag_C(t_C)| = |B \cap N[C]| + |N(C) \cap L|\mbox{.}
\end{align}
Here $|\bag_C(t_C)| = |B \cap N[C]| + |N(C) \cap L|$ follows from the construction of the function $\bag_C$. 
From \Cref{eqn:bagCompare} we have that $|B \cap R| = |B \setminus N[C]| < |N(C) \cap L|$. But then we have that 
$$|N(C)| = |N(C) \setminus L| +  |N(C) \cap  L|  > |N(C) \setminus L| + |B \cap R| \geq |S|\mbox{.}$$

\begin{figure}[htb]
\begin{center}
\includegraphics[width=0.5\textwidth]{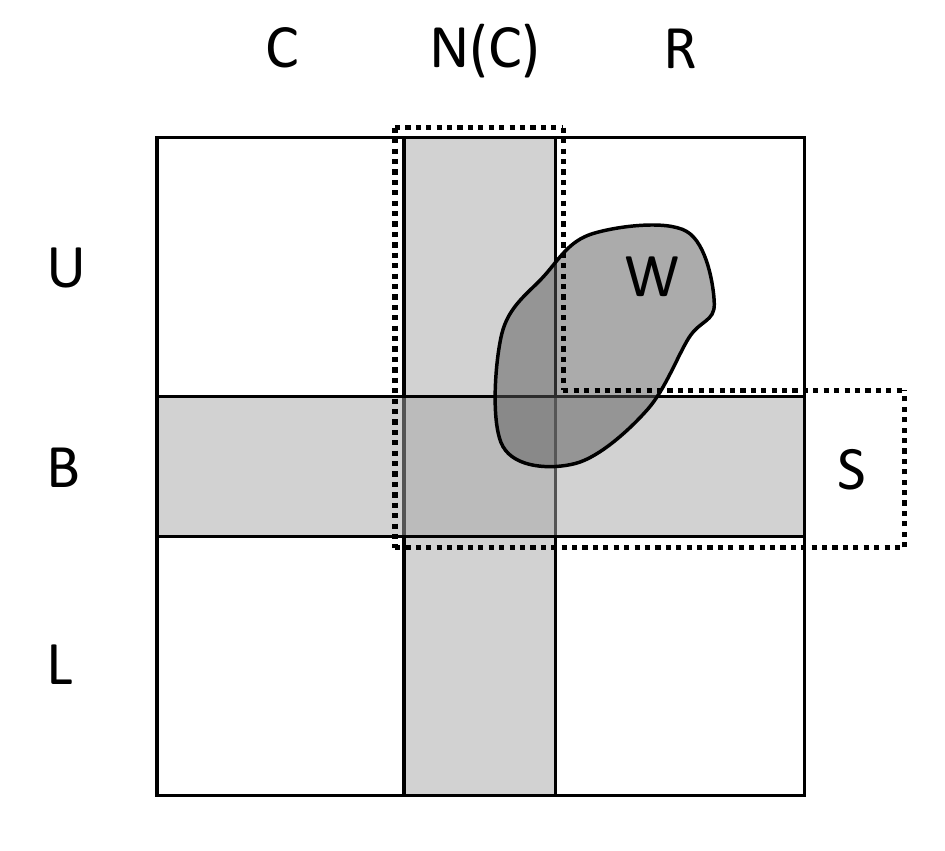}
\caption{$S$ separates $N(C)$ from $W$ and $|S| < |N(C)|$ because $|B \cap R| < |N(C) \cap L|$.}\label{fig:separation}
\end{center}
\end{figure}

Since $S$ separates $N(C)$ from $W$, there exists a separation $(P, S, Q)$ with $W \subseteq S \cup P$ and $N(C) \subseteq S \cup Q$, and $|S| < |N(C)|$.
Let $(P', S', Q')$ be a separation with $W \subseteq S' \cup P'$ and $N(C) \subseteq S' \cup Q'$ and $|S'|$ being of minimum size. Then $|S'| \leq |S| < |N(C)|$ and (by Menger's Theorem) the set $S'$ is linked into $N(C)$.
Now the component $C$ and separation  $(P', S', Q')$ satisfy all of the properties necessary to use  \Cref{lem:pullSimple} to improve $(X, (T, \bag))$.
This concludes case {\em (i)} (that there exists a node $t$ in $V(T)$ and component $C$ of $G \setminus X$ such that $|\bag_C(t_C)| > |\bag(t)|$).

The second case (when there exists a node $t$ in $V(T)$ and two distinct components $C_1, C_2$ of $G \setminus X$ such that $|\bag_{C_1}(t_{C_1})| = |\bag_{C_2}(t_{C_2})| = |\bag(t)| = |W|$) is handled in an analogous, but even more technical way.
In particular in this case we are not able to necessarily obtain an $X^\star$ with $|X^\star| < |X|$, but instead we obtain an  $X^\star$ with $|X^\star| = |X|$ and a lower value of a carefully chosen potential function. This concludes the proof sketch of \Cref{lem:improveTDSimplified}.

\subsection{Overview of \Cref{the:stwexalg,the:stwpalg}}
\label{subsec:branchoverview}
We now overview our algorithms for \stw and \pstw.
%\todo{TODO: Is separation defined? How about linked?}
Recall that every instance of \stw is also an instance of \pstw, so we will only work on instances of \pstw.
This will be useful also in the algorithm for \stw, since the recursive subproblems turn out to naturally correspond to \pstw.

We denote an instance of \pstw by $\ins = (G, \{W_1, \ldots, W_t\}, k)$.
We call the cliques $W_1, \ldots, W_t$ the \emph{terminal cliques} of the instance.
We say that a torso tree decomposition $(X, (T, \bag))$ in $G$ is a solution of $\ins$ if $(X, (T, \bag))$ covers $\bigcup_{i=1}^t W_i$ and has width at most $k$.
Here we do not anymore enforce that $|\bigcup_{i=1}^t W_i| \le k+2$, and it will in fact grow larger in the recursive subproblems (but $k$ or $t$ will not increase).
Both of our algorithms will either find a solution or conclude that no solution exists.
In particular, we do not use the freedom in the definitions of the problems that we could also determine that the treewidth of $G$ is more than $k$ without determining that no solution exists.

We will first sketch a $k^{\OO(k)} n^{\OO(1)}$ time algorithm for \pstw in the case when there are only two terminal cliques $W_1$ and $W_2$.
This algorithm showcases the most important concepts behind both the $k^{\OO(kt)} nm$ time algorithm of \Cref{the:stwpalg} and the $2^{\OO(k^2)} nm$ time algorithm of \Cref{the:stwexalg}, and in fact generalizing this to the $k^{\OO(kt)} nm$ algorithm does not require substantial new ideas but is rather a technical step.

\paragraph{Reduction rule.} Let $W_1$,$W_2$ be the two terminal cliques and $S$ be a minimum size $(W_1,W_2)$-separator, and $(A,S,B)$ the corresponding separation with $W_1 \subseteq A \cup S$ and $W_2 \subseteq B \cup S$.
We will argue that we can make $S$ into a new terminal clique and recursively solve the problem on the graphs $G[A \cup S]$ and $G[B \cup S]$.
More formally, we denote by $G \mcliq S$ the graph obtained from $G$ by making $S$ a clique, and then denote by $\ins \rescliqs (A, S)$ the instance $(G[A \cup S] \mcliq S, \{W_1, S\}, k)$ and by $\ins \rescliqs (B, S)$ the instance $(G[B \cup S] \mcliq S, \{W_2, S\}, k)$.
We argue that there exists a solution of $\ins$ if and only if there exists solutions of both $\ins \rescliqs (A, S)$ and $\ins \rescliqs (B, S)$.
%We argue that the instances $\ins \rescliqs (A, S)$ and $\ins \rescliqs (B, S)$ can then be solved independently of each other.

Observe that because both $\ins \rescliqs (A, S)$ and $\ins \rescliqs (B, S)$ contain the separator $S$ as a terminal clique but their graphs are disjoint otherwise, any solution of $\ins \rescliqs (A, S)$ can be combined with any solution of $\ins \rescliqs (B, S)$ into a solution of $\ins$ by simply connecting the tree decompositions by an edge between bags containing $S$.
To argue that if there exists a solution of $\ins$ then there exists solutions of both $\ins \rescliqs (A, S)$ and $\ins \rescliqs (B, S)$, we apply the pulling lemma (\Cref{lem:pullSimple}).
Because $S$ is a minimum size $(W_1,W_2)$-separator, by Menger's theorem $S$ is linked into $W_1$ and into $W_2$.
Therefore, in order to show that a solution of $\ins \rescliqs (A,S)$ exists, we consider a hypothetical solution $(X, (T, \bag))$ of $\ins$, and apply the pulling lemma with the separation $(A,S,B)$ and $Z = W_2$ as the subset of a bag with $Z \subseteq S \cup B$ into which $S$ is linked.
This constructs a torso tree decomposition $((X \cap A) \cup S, (T', \bag'))$ of width at most $k$ where $S$ is a bag, which can be observed to be a torso tree decomposition also in $G[A \cup S] \mcliq S$ because $S$ is a bag of $(T', \bag')$, and to cover $W_1 \cup S$ because $W_1 \subseteq A \cup S$ and $W_1 \subseteq X$, and therefore is a solution of $\ins \rescliqs (A,S)$.
The existence of a solution of $\ins \rescliqs (B,S)$ is proven in a symmetric way.

Observe that this reduction rule makes progress as long as $S \neq W_1$ and $S \neq W_2$, and thus we apply the rule as long as there exists any such minimum size $(W_1,W_2)$-separator $S$.
Motivated by this, we say that $W_1$ is \emph{strictly linked} into $W_2$ if $W_1$ is linked into $W_2$ and the only minimum size $(W_1,W_2)$-separators are $W_1$ and perhaps $W_2$ (if $|W_2| = |W_1|$).
%We also apply another simpler reduction rule that if there exists a separation $(A,S,B)$ with $S \subseteq W_1$ and both $A,B$ non-empty, then again solve the 

%$G \setminus W_1$ is not connected, then we delete from $G$ all connected components of $G \setminus W_1$ that do not intersect $W_2$ (and symmetrically if $G \setminus W_2$ is not connected we delete all connected components of $G \setminus W_2$ that do not intersect $W_1$).

\paragraph{Leaf pushing.}
Assume now that we cannot make any more progress by the reduction rule, and let $|W_1| \le |W_2|$, implying that $W_1$ is strictly linked into $W_2$.
Our goal is to now make progress by increasing the size of $W_1$.
We observe that for any solution $(X, (T, \bag))$ that minimizes $|X|$, it holds that if $l$ is a leaf node of $T$ and $p$ is the parent of $l$, then $\bag(l) \setminus \bag(p) \subseteq W_1 \cup W_2$.
Furthermore, we can assume that $\bag(p) = \bag(l) \setminus \{w\}$, where $w$ is a ``forget-vertex'' of $l$, and therefore $\bag(l) \setminus \bag(p) \subseteq W_1$ or $\bag(l) \setminus \bag(p) \subseteq W_2$.
Then, observe that if $\bag(l) \setminus \bag(p)$ intersects $W_i$, it must hold that $W_i \subseteq \bag(l)$ because $W_i$ is a clique.
Therefore, $(T, \bag)$ either contains a bag that contains both $W_1$ and $W_2$, in which case $|W_1 \cup W_2| \le k+1$ and there is a trivial single-bag solution, or $(T, \bag)$ has exactly two leaves and for one of them it holds that $W_1 \subseteq \bag(l)$ and $\bag(l) \setminus \bag(p) \subseteq W_1 \setminus W_2$.

Now, our goal will be, informally, to increase the size of $W_1$ by guessing a vertex in $\bag(l) \setminus W_1$ and adding it to $W_1$.
We let $w$ be the forget-vertex of $l$, and observe that the parent bag $\bag(p) = \bag(l) \setminus \{w\}$ is a $(W_1, W_2)$-separator.
This shows that $\bag(l) \setminus W_1$ must be non-empty, because otherwise $\bag(p)$ would be a $(W_1, W_2)$-separator of size $|W_1|-1$, contradicting that $W_1$ is linked into $W_2$.
Denote $G' = G \setminus (W_1 \setminus \{w\})$, and observe that in the graph $G'$ the set $\bag(l) \setminus W_1 = \bag(p) \setminus W_1$ is a $(\{w\}, W_2 \setminus W_1)$-separator.
We will then show that the subset $\bag(l) \setminus W_1$ of $\bag(l)$ can be replaced by an important $(\{w\}, W_2 \setminus W_1)$-separator (see \Cref{subsec:impsep} or~\cite[Chapter 8]{cygan2015parameterized} for definitions of important separators).
In particular, we will argue that there is an important $(\{w\}, W_2 \setminus W_1)$-separator $S \neq \{w\}$ in the graph $G'$ so that there exists a solution containing a bag $W_1 \cup S$.

Let $S$ be an important $(\{w\}, W_2 \setminus W_1)$-separator in the graph $G'$ so that it dominates $\bag(l) \setminus W_1$ and minimizes $|S|$ among all such important separators.
Denote the separation corresponding to $S$ by $(A,S,B) = (\reach_{G'}(\{w\}, S), S, V(G) \setminus (S \cup \reach_{G'}(\{w\}, S)))$, where $\reach_{G'}(\{w\}, S)$ denotes the vertices reachable from $\{w\}$ in the graph $G' \setminus S$.
It can be shown that $S$ is linked into $(A \cup S) \cap (\bag(l) \setminus W_1)$ (\Cref{lem:imp_sep_dom}).
Then, by adding $W_1 \setminus \{w\}$ back to the graph and to the separation, we get that $(A, S \cup W_1 \setminus \{w\}, B)$ is a separation of $G$ and $S \cup W_1 \setminus \{w\}$ is linked into $(A \cup S \cup W_1 \setminus \{w\}) \cap \bag(l)$ (the vertices in $W_1 \setminus \{w\}$ are linked by trivial one-vertex paths).
We then apply the pulling lemma (\Cref{lem:pullSimple}) with the hypothetical solution $(X, (T, \bag))$, the separation $(B, S \cup W_1 \setminus \{w\}, A)$, and the subset of a bag $Z = (A \cup S \cup W_1 \setminus \{w\}) \cap \bag(l)$, to argue that there exists a torso tree decomposition $((X \cap B) \cup S \cup W_1 \setminus \{w\}, (T', \bag'))$ of width at most $k$, containing a bag $S \cup W_1 \setminus \{w\}$.
As $|S| \le |\bag(l) \setminus W_1|$, this can be turned into a solution of $\ins$ by inserting $w$ into the bag $S \cup W_1 \setminus \{w\}$.
Therefore there exists a solution of $\ins$ with a bag $W_1 \cup S$, and in particular it is safe to replace the terminal clique $W_1$ by $W_1 \cup S$, also replacing $G$ by $G \mcliq (W_1 \cup S)$.

Now, we are able to increase the size of $W_1$ by guessing the forget-vertex $w \in W_1$ and an important separator $S$ and branching to $(G \mcliq (W_1 \cup S), \{W_1 \cup S, W_2\}, k)$.
However, by applying the reduction rule we might immediately lose most of the progress by finding a $(W_1 \cup S, W_2)$-separator $S'$ of size $|S'|<|W_1 \cup S|$ and ending up with an instance with terminal cliques $\{S', W_2\}$.
Nevertheless, we can ensure that such $S'$ must have size $|S'|>|W_1|$ by using the facts that $W_1$ is strictly linked into $W_2$ and the way $S$ was selected.
In particular, in the end, after applying the reduction rule possibly several times, we can guarantee that if initially $|W_1| = |W_2|$, then each resulting instance has terminal cliques of sizes at least $|W_1|+1$ and $|W_2|$, and if initially $|W_1| < |W_2|$, then each resulting instance has terminal cliques of sizes at least $|W_1|+1$ and $|W_1|+1$.
Therefore, if we consider $\min(|W_1|, |W_2|)+\min(\min(|W_1|, |W_2|)+1, \max(|W_1|, |W_2|))$ as our measure of progress, we are guaranteed to increase it by one by the branching.

As the sizes of terminal cliques are bounded by $k+1$, it is possible to increase this measure by at most $2k$ times.
Then, as the number of important separators of size at most $k$ is bounded by $4^k$~\cite{DBLP:journals/algorithmica/ChenLL09}, this results in a branching tree of degree $k 4^k$ and depth $2k$, resulting in a $(k 4^k)^{2k} n^{\OO(1)} = 2^{\OO(k^2)} n^{\OO(1)}$ time algorithm.
To improve this to $k^{\OO(k)} n^{\OO(1)}$ time, we observe that in order to make progress, it is sufficient to guess only one vertex of the important separator $S$ and add it to $W_1$, instead of guessing the whole important separator $S$.
To this end, we prove an ``important separator hitting set lemma'' (\Cref{lem:imp_sep_hit}) that gives a set of size $k$ that intersects all important separators of size at most $k$, and therefore allows to guess one vertex in an important separator of size at most $k$ by a branching degree of $k$ instead of $4^k$, resulting in a $(k^2)^{2k} n^{\OO(1)} = k^{\OO(k)} n^{\OO(1)}$ time algorithm.

\paragraph{More than two terminal cliques.}
Generalizing the just sketched $k^{\OO(k)} n^{\OO(1)}$ time algorithm for two terminal cliques into the $k^{\OO(kt)} n^{\OO(1)}$ algorithm for $t$ terminal cliques of \Cref{the:stwpalg} does not require major new ideas, but requires several technical considerations.
In the algorithm for $t$ terminal cliques, we will in addition to the leaf pushing branching do branching on merging two different terminal cliques into one, which should be done whenever we guess that there exists a solution where the two terminal cliques are in a same bag.
The ``real'' definition of the measure of the instance will also be more involved, in particular, instead of depending on the sizes of terminal cliques, the measure depends on a notion of ``flow potential'' of a terminal clique.
The flow potential has a technical definition, but for all terminal cliques $W_i$ except for a uniquely largest one it will be equal to the flow from $W_i$ into the union of the other terminal cliques.
The measure of a uniquely largest terminal clique must be special to encode that we make progress, for example, in the case when there are two terminal cliques $W_1$ and $W_2$ with $|W_1| = |W_2|$ and after branching we end up with two terminal cliques of sizes $|W_1| + 1$ and $|W_2|$.
The measure will also take into account the number of terminal cliques, in particular, it will ``encode'' that decreasing the number of terminal cliques with the expense of making the flow potential of one terminal clique worse still means making overall progress.

%decreasing the number of terminal cliques makes a lot of progress.

\paragraph{The $2^{\OO(k^2)} n^{\OO(1)}$ time algorithm.}
The $2^{\OO(k^2)} n^{\OO(1)}$ time algorithm for \stw of \Cref{the:stwexalg} also uses the same reduction rule and leaf pushing arguments.
In particular, even though the problem is originally \stw, applications of the reduction rule and leaf pushing will naturally turn the problem into \pstw.

For this algorithm, the main measure of progress will be a parameter $q$ that states that there are no solutions that contain ``internal separations'' of order $<q$.
Here, an internal separation of a solution $(X, (T,\bag))$ means a separation $(A,S,B)$ so that $S$ is a subset of some bag of $(T,\bag)$, and the terminal cliques intersect both $A$ and $B$.
The goal will be to increase $q$, by first pushing two terminal cliques to be of size at least $\ge q$ by using a version of leaf pushing that guesses the whole important separator instead of only one vertex, and then guessing how a hypothetical internal separation of order $q$ would split the terminal cliques and breaking the instance by an important separator of size $q$ pushed towards the side with two terminal cliques of size $\ge q$.
We will also argue about internal separations that contain only a small number of ``original'' terminal vertices behind them, in particular, we will use an observation that if a solution has an internal separation $(A,S,B)$ so that at most $k+1-|S|$ original terminal vertices are ``behind'' terminal cliques intersecting $A$, then the $A$-side of the solution can be replaced by just a single bag containing $S$ and the original terminal vertices behind it.

\section{Preliminaries}
\label{sec:preli}
We present definitions and preliminary results.
%\Cref{subsec:preligraph,subsec:prelilinked,subsec:prelitd} contain mostly standard material, but in \Cref{subsec:prelittd} the definition of torso tree decompositions is introduced and observations related to them are discussed.

For a positive integer $n$ we denote $[n] = \{1, 2, \ldots, n\}$ and for two integers $a,b$ with $a\le b$ we denote $[a,b] = \{a, a+1, \ldots, b\}$.

\subsection{Graphs}
\label{subsec:preligraph}
We denote the set of vertices of a graph $G$ by $V(G)$ and the set of edges by $E(G)$.
When the graph $G$ is clear from the context, we use $n = |V(G)|$ and $m = |V(G)|+|E(G)|$.
For a vertex $v \in V(G)$ we denote its neighborhood in $G$ by $N_G(v)$ and closed neighborhood by $N_G[v] = N_G(v) \cup \{v\}$.
For a set of vertices $S \subseteq V(G)$ their neighborhood is $N_G(S) = \bigcup_{v \in S} N(v) \setminus S$ and the closed neighborhood $N_G[S] = N_G(S) \cup S$.
We drop the subscript if the graph is clear from the context.
We denote the subgraph of $G$ induced by $S \subseteq V(G)$ by $G[S]$, and we also use the notation $G \setminus S = G[V(G) \setminus S]$.
We denote by $G \mcliq S$ the graph obtained from $G$ by making $S$ a clique.

A tripartition $(A, S, B)$ of $V(G)$ (with possibly empty parts) is a \emph{separation} of $G$ if there are no edges between $A$ and $B$.
The \emph{order} of the separation is $|S|$.
A separation of $G$ is a \emph{strict separation} if both $A$ and $B$ are non-empty.
For two sets $X,Y \subseteq V(G)$, an $(X,Y)$-separator is a set $S$ so that in the graph $G \setminus S$ there are no paths from $X \setminus S$ to $Y \setminus S$.
An $(X,Y)$-separator $S$ is a minimal $(X,Y)$-separator if no proper subset of $S$ is an $(X,Y)$-separator.
Note that $S$ is an $(X,Y)$-separator if and only if there exists a separation $(A,S,B)$ of $G$ with $X \subseteq A \cup S$ and $Y \subseteq B \cup S$.

For two sets of vertices $X,S \subseteq V(G)$, we denote by $\reach_G(X, S)$ the set of vertices in $G \setminus S$ reachable from $X \setminus S$.
We define $\reachn_G(X,S) = (X \cap S) \cup N(\reach_G(X,S)) \subseteq S$ to denote the subset of $S$ that can be seen from $X$.
Note that if $S$ is an $(X,Y)$-separator then $\reachn_G(X,S)$ is also an $(X,Y)$-separator and $\reach_G(X,\reachn_G(X,S)) = \reach_G(X,S)$.
It follows that if $S$ is a minimal $(X,Y)$-separator, then $S = \reachn_G(X, S)$.

For two sets of vertices $X,Y \subseteq V(G)$, we denote by $\flow_G(X,Y)$ the maximum number of vertex-disjoint paths in $G$ starting in $X$ and ending in $Y$.
We may omit the subscript if the graph is clear from the context.
By Menger's theorem, $\flow(X,Y)$ is equal to the size of a minimum size $(X,Y)$-separator.

We say that a set $X \subseteq V(G)$ is \emph{linked} into a set $Y \subseteq V(G)$ if $\flow(X,Y) = |X|$.
Note that here the definition of linked is asymmetric, in particular, the fact that $X$ is linked into $Y$ does not imply that $Y$ is linked into $X$.
We say that $X$ is \emph{strictly linked} into $Y$ if it is linked into $Y$ and for all $(X,Y)$-separators $S$ of size $|S| = |X|$ it holds that $S = X$ or $S = Y$.

\subsection{Tree decompositions}
\label{subsec:prelitd}
A tree decomposition of a graph $G$ is a pair $(T, \bag)$, where $T$ is a tree and $\bag$ is a function $\bag : V(T) \rightarrow 2^{V(G)}$ that satisfies
\begin{enumerate}
\item \label{def:td:vc} $V(G) = \bigcup_{t \in V(T)} \bag(t)$,
\item \label{def:td:ec} for every $uv \in E(G)$, there exists $t \in V(T)$ with $\{u,v\} \subseteq \bag(t)$, and
\item \label{def:td:cc} for every $v \in V(G)$, the set $\{t \in V(T) \mid v \in \bag(t)\}$ forms a connected subtree of $T$.
\end{enumerate}

We will call \Cref{def:td:vc} of the definition the \emph{vertex condition}, \Cref{def:td:ec} the \emph{edge condition}, and \Cref{def:td:cc} the \emph{connectedness condition}.
The width of a tree decomposition $(T, \bag)$ is $\max_{t \in V(T)} |\bag(t)|-1$ and the treewidth of a graph is the minimum width of a tree decomposition of it.
We usually call the vertices of the tree $T$ \emph{nodes} to distinguish them from the vertices of the graph $G$.

We will need the following standard utility lemma that transforms a tree decomposition into a no worse tree decomposition with at most $n$ nodes.

\begin{lemma}
\label{lem:tdlinear}
Given a tree decomposition $(T,\bag)$ of $G$ of width $k$ that has $h$ bags of size $k+1$, we can in time $k^{\OO(1)} |V(T)|$ construct a tree decomposition of $G$ of width $k$ that has at most $h$ bags of size $k+1$ and has at most $n$ nodes.
\end{lemma}
\begin{proof}
As long as there exists an edge $uv \in E(T)$ with $\bag(u) \subseteq \bag(v)$, we contract $uv$ and let the bag of the resulting node be $\bag(v)$.
This can be implemented in $k^{\OO(1)} |V(T)|$ time by depth-first search, and clearly does not increase the width or the number of bags of size $k+1$.
This results in a tree decomposition with at most $n$ nodes~(see~e.g.~\cite[Chapter 14.2]{fomin2019kernelization}).
\end{proof}

We sometimes view a tree decomposition $(T, \bag)$ as rooted on some specific node $r \in V(T)$.
In this setting we use standard rooted tree terminology, i.e., $v \in V(T)$ is an \emph{ancestor} of $u \in V(T)$ if it is on the unique path from $u$ to $r$ and a \emph{strict ancestor} if also $v \neq u$, and conversely $u$ is a (strict) \emph{descendant} of $v$.
We say that a node $t \in V(T)$ is the forget-node of a vertex $v \in V(G)$ if $v \in \bag(t)$ and either $t = r$ or for the parent $p$ of $t$ it holds that $v \notin \bag(p)$.
Note that every $v \in V(G)$ has a unique forget-node.

\subsection{Torso tree decompositions}
\label{subsec:prelittd}
Let $G$ be a graph and $X \subseteq V(G)$.
The graph $\torso_G(X)$ has set of vertices $V(\torso_G(X)) = X$ and has $uv \in E(\torso_G(X))$ if $u,v \in X$ and there is a path from $u$ to $v$ whose internal vertices are in $V(G) \setminus X$.
In particular, note that $E(\torso_G(X)) \supseteq E(G[X])$.
An equivalent definition of $\torso_G(X)$ is that it is the graph obtained from $G[X]$ by making $N_G(C)$ a clique for every connected component $C$ of $G \setminus X$.

We will need a following lemma about the interplay of the $\torso$ operation and induced subgraphs.

\begin{lemma}
\label{lem:torsoindsub}
Let $X,Y$ be subsets of $V(G)$.
Then $E(\torso_{G[Y]}(X \cap Y)) \subseteq E(\torso_G(X))$.
\end{lemma}
\begin{proof}
If $uv \in E(\torso_{G[Y]}(X \cap Y))$, then there is a path from $u$ to $v$ in $G[Y]$ with intermediate vertices in $Y \setminus X$.
This path exists also in $G$, implying that $uv \in E(\torso_G(X))$.
\end{proof}

A \emph{torso tree decomposition} in a graph $G$ is a pair $(X, (T, \bag))$, where $X \subseteq V(G)$ and $(T, \bag)$ is a tree decomposition of $\torso_G(X)$.
For a set $W \subseteq V(G)$, we say that $(X, (T, \bag))$ covers $W$ if $W \subseteq X$.

We observe the following equivalent viewpoint of torso tree decompositions that might be useful for intuition about them.

\begin{observation}
\label{obs:bigleaf}
There exists a torso tree decomposition $(X, (T, \bag))$ in $G$ if and only if there exists a tree decomposition of $G$ whose non-leaf nodes induce the tree decomposition $(T,\bag)$.
\end{observation}

For tree decompositions it holds that for any connected induced subgraph $G[Y]$, the set of bags intersecting $Y$ forms a connected subtree of the decomposition (see~e.g.~\cite[Chapter~14.1]{fomin2019kernelization}).
We will use a corresponding property of torso tree decompositions.

\begin{lemma}
\label{lem:indsub}
Let $(X, (T, \bag))$ be a torso tree decomposition in $G$, and let $Y \subseteq V(G)$ so that $G[Y]$ is connected.
The nodes $\{t \in V(T) \mid \bag(t) \cap Y \neq \emptyset\}$ induce a (possibly empty) connected subtree of $T$.
\end{lemma}
\begin{proof}
By the definition of $\torso_G(X)$, any $u-v$-path in $G[Y]$ with $u,v \in X$ can be mapped into an $u-v$-path in $\torso_G(X)[Y \cap X]$, and therefore $\torso_G(X)[Y \cap X]$ is connected, and therefore the lemma follows from the corresponding property of tree decompositions.
%Follows from \Cref{obs:bigleaf} and the corresponding property of tree decompositions.
\end{proof}

Alternatively, \Cref{lem:indsub} could be proven by using \Cref{obs:bigleaf} and the property of tree decompositions.
\Cref{lem:indsub} implies that if there are nodes $s,x,y \in V(T)$ so that $\{s\}$ is an $(\{x\}, \{y\})$-separator in $T$, then $\bag(s)$ is a $(\bag(x), \bag(y))$-separator in $G$.
This implication is proven by letting $Y$ to be any $\bag(x)-\bag(y)$-path and observing that by \Cref{lem:indsub}, $Y$ must now intersect $\bag(s)$.

\section{Toolbox}
\label{sec:toolbox}
In this section we provide two generic algorithmic tools that will be used in multiple sections of the paper.
First, in \Cref{subsec:impsep} we discuss known results about important separators and prove two observations about them.
In particular we believe that \Cref{lem:imp_sep_hit,lem:impsepfk} do not occur in the prior literature, while the rest of \Cref{subsec:impsep} is well-known material.
The results about important separators will be used in \Cref{sec:algopstw,sec:fasttw}.
Then, in \Cref{subsec:pull} we prove a ``pulling lemma'' about torso tree decompositions that will be crucial in many parts of our algorithm.
This lemma in the context of torso tree decompositions is novel, but prior analogues in the context of tree decompositions exist.
The pulling lemma will be used in \Cref{sec:redu,sec:algopstw,sec:fasttw}.

\subsection{Important separators}
\label{subsec:impsep}
The notion of important separator was introduced by Marx~\cite{DBLP:journals/tcs/Marx06}.

\begin{definition}[Important separator]
\label{def:impsep}
Let $A,B \subseteq V(G)$ be two sets of vertices.
A minimal $(A,B)$-separator $S$ is called an important $(A,B)$-separator if there exists no $(A,B)$-separator $S'$ with $|S'| \le |S|$ and $\reach_G(A, S) \subset \reach_{G}(A, S')$.
\end{definition}

We remark that \Cref{def:impsep} allows an important $(A,B)$-separator to be an empty set in the case when $B$ is not reachable from $A$ in $G$.

The following lemma is a straightforward consequence of \Cref{def:impsep}.

\begin{lemma}
\label{lem:impsepbasic}
For any $(A,B)$-separator $S$, there exists an important $(A,B)$-separator $S'$ so that $|S'| \le |S|$ and $\reach_G(A, S) \subseteq \reach_{G}(A, S')$.
\end{lemma}
\begin{proof}
By iteration of the definition.
\end{proof}

We say that an important $(A,B)$-separator $S'$ \emph{dominates} an $(A,B)$-separator $S$ if $S'$ satisfies the conditions of \Cref{lem:impsepbasic}.
For our algorithm, we need a property that a smallest important separator that dominates $S$ is linked into $S$ in a certain way.

\begin{lemma}
\label{lem:imp_sep_dom}
Let $S$ be an $(A,B)$-separator and $S'$ a smallest important $(A,B)$-separator that dominates $S$.
It holds that $S'$ is linked into $S \cap (R_G(A,S') \cup S')$.
\end{lemma}
\begin{proof}
Note that $\reachn_G(A,S) \subseteq \reach_G(A,S') \cup S'$ which implies that the set $S \cap (R_G(A,S') \cup S')$ is an $(A,S')$-separator and $\reach_G(A, S \cap (R_G(A,S') \cup S')) = \reach_G(A, S)$.
Suppose that $S'$ is not linked into $S \cap (R_G(A,S') \cup S')$, and let $S''$ be a $(S \cap (R_G(A,S') \cup S'), S')$-separator of size $|S''| < |S'|$.
Now, $S''$ is an $(A,B)$-separator and $\reach_G(A, S) \subseteq \reach_G(A, S'')$.
Therefore, an important $(A,B)$-separator that dominates $S''$ would also dominate $S$ and be smaller than $S'$, which contradicts the choice of $S'$.
\end{proof}

%We also need the following straightforward linkedness condition.

%\begin{lemma}
%\label{lem:imp_sep_slinked}
%Every important $(A,B)$-separator is strictly linked into $B$.
%\end{lemma}
%\begin{proof}
%Suppose that $S$ is an important $(A,B)$-separator but is not strictly linked into $B$, and let $S'$ be a minimum size $(S,B)$-separator with $|S'| \le |S|$ and $S' \neq S$.
%Now, because $S$ is an $(A,B)$-separator and $S'$ is a minimum size $(S,B)$-separator, it holds that $\reach_G(A, S) \subseteq \reach_G(A, S')$, which by $S' \neq S$ and the fact that $S$ is a minimal $(A,B)$-separator implies $\reach_G(A, S) \subset \reach_G(A, S')$, implying that $S$ is not an important $(A,B)$-separator.
%\end{proof}

We also need the following observation about important separators.

\begin{lemma}
\label{lem:mini_sep_impdom}
If $S$ is a minimal $(A,B)$-separator and $S'$ an important $(A,B)$-separator that dominates $S$, then $S'$ is an $(S,B)$-separator.
\end{lemma}
\begin{proof}
%\todo{TODO: maybe no need to prove this?}
The lemma clearly holds if all vertices of $S$ are in $R_G(A, S') \cup S'$.
Suppose there is a vertex $v \in S \setminus (R_G(A, S') \cup S')$.
However, now because $\reachn_G(A, S) \subseteq R_G(A', S') \cup S'$, we would have that $\reachn_G(A, S) \subset S$ implying that $S$ is not a minimal $(A,B)$-separator.
\end{proof}

The following bound on the important separators was given implicitly by~\cite{DBLP:journals/algorithmica/ChenLL09} and explicitly in~\cite{marx-razgon-stoc2011-multicut}.

\begin{lemma}[\cite{DBLP:journals/algorithmica/ChenLL09,marx-razgon-stoc2011-multicut}]
\label{lem:impsep4k}
For any two sets $A,B \subseteq V(G)$, there are at most $4^k$ important $(A,B)$-separators of size at most $k$, and they can be enumerated in time $4^k k^{\OO(1)} m$ and polynomial space.
\end{lemma}

We will prove more bounds about important separators.
For this, the basic tool will be the following property of important separators of minimum size given by Marx~\cite{DBLP:journals/tcs/Marx06}.

\begin{lemma}[\cite{DBLP:journals/tcs/Marx06}]
\label{lem:imp_sep_uniq}
For any two sets $A,B \subseteq V(G)$, there exists exactly one important $(A,B)$-separator $S$ of size $|S| = \flow(A,B)$, and moreover for all important $(A,B)$-separators $S'$ it holds that $\reach_G(A,S) \subseteq \reach_G(A,S')$.
\end{lemma}
%\begin{proof}
%\todo{TODO: maybe no need to prove this?}
%Let $S$ be any $(A,B)$-separator of size $\flow(A,B)$ and $S'$ an important $(A,B)$-separator.
%We will show that $S' = \rightsep_{A,B}(S,S')$.
%If $|\rightsep_{A,B}(S,S')| > |S'|$ would hold, then by submodularity we would have $|\leftsep_{A,B}(S,S')| < |S|$, which would contradict that $S$ is a minimum size $(A,B)$-separator.
%If $|\rightsep_{A,B}(S,S')| \le |S'|$ then $\rightsep_{A,B}(S,S')$ dominates $S'$, implying that $S' = \rightsep_{A,B}(S,S')$ because $S'$ is an important $(A,B)$-separator.

%This implies that all important $(A,B)$-separators $S'$ have $\reach_G(A,S') \supseteq \reach_G(A,S)$.
%In particular, this implies that there exists a unique important $(A,B)$-separator of size $\flow(A,B)$, because any other important $(A,B)$-separator of size $\flow(A,B)$ would dominate it.
%\end{proof}

We will next show that there exists a set of size at most $k$ that intersects every important $(A,B)$-separator of size at most $k$, i.e., a \emph{hitting set} for important $(A,B)$-separators of size at most $k$, and that such a hitting set can be computed efficiently.
To the best of our knowledge this is a novel lemma about important separators, though its proof is only a small variant of the proof of \Cref{lem:impsep4k}.

\begin{lemma}[Important separator hitting set lemma]
\label{lem:imp_sep_hit}
There is an algorithm that given two sets $A,B \subseteq V(G)$ and an integer $k$, in time $k^{\OO(1)} m$ outputs a set $H$ of size at most $k$ so that $H$ intersects every non-empty important $(A,B)$-separator of size at most $k$.
\end{lemma}
\begin{proof}
When $B$ is not reachable from $A$, we can let $H$ be the empty set.
When $B$ is reachable from $A$, we show by induction that there exists such a set $H$ of size at most $\max(0, k - \flow(A,B)+1)$, which implies the lemma because in this case $\flow(A,B) \ge 1$.

This holds in the base case $k < \flow(A,B)$ because then there exists no $(A,B)$-separators of size at most $k$ so we can take $H$ as the empty set.
Now assume that $k \ge \flow(A,B)$ and that this holds when the difference of $k$ and $\flow(A,B)$ is smaller.

By \Cref{lem:imp_sep_uniq}, let $S$ be the unique important $(A,B)$-separator of size $\flow(A,B)$.
If $S$ intersects $B$, then all important $(A,B)$-separators intersect $S \cap B$ and we are done by outputting any vertex of $S \cap B$.
Otherwise, assume that $S$ does not intersect $B$ and let $v$ be any vertex $v \in S$.
All important $(A,B)$-separators of size $|S|$ intersect $v$, and all important $(A,B)$-separators $S'$ of size $|S'|>|S|$ that do not intersect $v$ have $v \in \reach_G(A,S')$ and $\reach_G(A,S') \supset \reach_G(A,S)$ and therefore are important $(S \cup N(v), B)$-separators.
We observe that $\flow(S \cup N(v), B) > \flow(A,B)$, and therefore by induction assumption we construct $H$ by taking the union of $v$ and the hitting set for important $(S \cup N(v), B)$-separators of size at most $k$.
\end{proof} 

We will also need the following bound on important separators, which is also proven by a slight variation of the proof of \Cref{lem:impsep4k}.

\begin{lemma}
\label{lem:impsepfk}
For any two sets $A,B \subseteq V(G)$, there are at most $k^{k-\flow(A,B)}$ important $(A,B)$-separators of size at most $k$.
\end{lemma}
\begin{proof}
Again, we prove the lemma by induction on $k-\flow(A,B)$.
By \Cref{lem:imp_sep_uniq}, it holds in the base case $k-\flow(A,B)=0$, so assume $k-\flow(A,B) \ge 1$ and that the lemma holds for smaller values of $k-\flow(A,B)$.

By \Cref{lem:imp_sep_uniq}, let $S$ be the unique important $(A,B)$-separator of size $|S| = \flow(A,B) < k$.
Because for any important $(A,B)$-separator $S'$ it holds that $\reach_G(A, S) \subseteq \reach_G(A, S')$, we observe that for every important $(A,B)$-separator $S' \neq S$ there exists a vertex $v \in S \setminus B$ so that $S'$ is an important $(S \cup N(v), B)$-separator.
We observe that $\flow(S \cup N(v), B) > \flow(A,B)$, and therefore by induction get that the total number of important $(A,B)$-separators of size at most $k$ is
\[1 + |S| \cdot k^{k-(\flow(A,B)+1)} \le 1 + (k-1) \cdot k^{k-\flow(A,B)} / k \le k^{k-\flow(A,B)}.\]
\end{proof}

\subsection{Pulling lemma}
\label{subsec:pull}
Next we prove a lemma that will be used throughout \Cref{sec:redu,sec:algopstw,sec:fasttw} to argue that a separator $S$ can be incorporated as a bag of a torso tree decomposition if it satisfies certain properties.
We call it the ``pulling lemma'' because the separator $S$ will be ``pulled'' along disjoint paths into a bag of the tree decomposition.
Lemmas analogous to this have been used in the context of tree decompositions for example by~\cite{bellenbaum2002two,BodlaenderK06}.

\begin{lemma}[Pulling lemma]
\label{lem:pull}
Let $G$ be a graph and $(X, (T,\bag))$ a torso tree decomposition in $G$.
Let $(A, S, B)$ be a separation of $G$ so that there exists a node $r \in V(T)$ so that $S$ is linked into $\bag(r) \cap (S \cup B)$.
There exists a torso tree decomposition $((X \cap A) \cup S, (T', \bag'))$ so that
\begin{enumerate}
\item $T' = T$
\item for all $t \in V(T)$, $|\bag'(t)| \le |\bag(t)|$, and
\item $S \subseteq \bag'(r)$.
\end{enumerate}
Moreover, when $G$, $(X, (T,\bag))$, $(A,S,B)$, and $r$ are given as inputs, the torso tree decomposition $((X \cap A) \cup S, (T', \bag'))$ can be constructed in $k^{\OO(1)} (|V(T)| + m)$ time, where $k$ is the width of $(X, (T,\bag))$.
\end{lemma}
\begin{proof}
Index the vertices of $S$ by $S = \{s_1, s_2, \ldots, s_{|S|}\}$.
Because $S$ is linked into $\bag(r) \cap (S \cup B)$, there are vertex-disjoint paths $P_1, \ldots, P_{|S|}$, so that for each $i \in [|S|]$, $P_i$ is a path from $s_i$ to $\bag(r) \cap (S \cup B)$, and all vertices of $P_i$ are contained in $S \cup B$.

To construct $(T', \bag')$, we set $T' = T$, and for each $t \in V(T)$ we set
\[\bag'(t) = (\bag(t) \setminus (S \cup B)) \cup \{s_i \mid P_i \cap \bag(t) \neq \emptyset\}.\]
We have that $|\bag'(t)| \le |\bag(t)|$, because for each inserted vertex $s_i$ we removed a vertex in $P_i$ (note that the inserted vertex and the removed vertex could both be the same vertex $s_i$).
By definition every $P_i$ intersects $\bag(r)$, and thus $S \subseteq \bag'(r)$.
Denote $X' = (X \cap A) \cup S$.
It remains to show that $(T', \bag')$ is a tree decomposition of $\torso(X')$.

First, the tree decomposition $(T', \bag')$ satisfies the vertex condition because no vertices in $X \cap A$ were removed, and as argued before $S \subseteq \bag'(r)$.
Second, $(T', \bag')$ satisfies the connectedness condition because the occurrences of vertices in $X \cap A$ were not altered, and by \Cref{lem:indsub} the sets $\{t \mid P_i \cap \bag(t) \neq \emptyset\}$ induce connected subtrees of $T$.

For the edge condition, consider an edge $uv \in E(\torso(X'))$.
There is a path between $u$ and $v$ whose intermediate vertices are contained in $V(G) \setminus X'$.
If there would be an intermediate vertex in $B$, then $u,v \in S$, implying $\{u,v\} \subseteq \bag'(r)$, so it remains to consider the cases where there are no intermediate vertices or all intermediate vertices are in $A \setminus X' = A \setminus X$.
It follows that if in this case $u,v \in X$, then $uv \in E(\torso(X))$, so the edge condition of $(T', \bag')$ in this case holds by the edge condition of $(T, \bag)$.
Also if $u,v \in S$, then again $\{u,v\} \subseteq \bag'(r)$, so the remaining case is $uv = s_i v$, where $s_i \in S \setminus X$ and $v \in X \setminus S$.
Now, $s_i$ and the intermediate vertices on the path between $s_i$ and $v$ are in a connected component $C$ of $G \setminus X$.
Because $v \in X$ and $\bag(r) \subseteq X$, this implies that $N(C)$ contains both $v$ and at least one vertex on the path $P_i$, and therefore as $N(C)$ is a clique in $\torso(X)$ there is a node $t \in V(T)$ with $N(C) \subseteq \bag(t)$ and it will hold that $\{s_i, v\} \subseteq \bag'(t)$.

Because $(T,\bag)$ has width $k$ and $|S| \le k+1$, the construction clearly can be implemented in $k^{\OO(1)} (|V(T)| + m)$ time.
\end{proof}

Note that the condition $|\bag'(t)| \le |\bag(t)|$ implies that the width of $(T', \bag')$ is at most the width of $(T,\bag)$.

\section{Computing treewidth by \stw}
\label{sec:redu}
In this section we show that in order to improve a tree decomposition, it is sufficient to solve \stw.
In particular, we prove \Cref{the:stweximpl,the:stwapximpl}.

\subsection{Improving a tree decomposition}
\label{subsec:itimprtd}
We will define a weighted version of linkedness.
For a weight function $d : V(G) \rightarrow \mathbb{Z}$ and a set $S \subseteq V(G)$, we denote $d(S) = \sum_{v \in S} d(v)$.

\begin{definition}[$d$-linked]
\label{def:dlinked}
Let $G$ be a graph, $A,B\subseteq V(G)$, and $d : V(G) \rightarrow \mathbb{Z}$ a weight function.
The set $A$ is $d$-linked into $B$ if for any $(A,B)$-separator $S$ it holds either that $|S|>|A|$, or that $|S|=|A|$ and $d(S) \ge d(A)$.
\end{definition}

Note that if $A$ is $d$-linked into $B$ then $A$ is linked into $B$.
We say that an $(A,B)$-separator $S$ with $|S|<|A|$, or with $|S|=|A|$ and $d(S) < d(A)$ \emph{witnesses} that $A$ is not $d$-linked into $B$.
Then, we say that a torso tree decomposition $(X,(T,\bag))$ is $d$-linked into a set of vertices $W \subseteq V(G)$ if for every node $t \in V(T)$ it holds that $\bag(t)$ is $d$-linked into $W$.
We say that a pair $(t, S)$, where $t \in V(T)$ and $S$ is a $(\bag(t),W)$-separator witnessing that $\bag(t)$ is not $d$-linked into $W$ witnesses that $(X,(T,\bag))$ is not $d$-linked into $W$.

Our goal is to show that any torso tree decomposition that covers $W$ can be made to be $d$-linked into $W$.
In particular, we will show that if $(X, (T, \bag))$ is a torso tree decomposition that covers $W$, then given a pair $(t,S)$ that witnesses that $(X,(T,\bag))$ is not $d$-linked into $W$, we can, in some sense, improve $(X, (T, \bag))$ while maintaining that it covers $W$ and not increasing its width.
We define $\phi_d(X) = |X| \cdot n (k+1) + d(X)$ as the measure in which sense we will improve $(X, (T, \bag))$.

\begin{lemma}
\label{lem:make_d_linkedv2}
There is an algorithm that takes as input a graph $G$, a set of vertices $W \subseteq V(G)$, a torso tree decomposition $(X, (T, \bag))$ in $G$ of width $k$ that covers $W$, a weight function $d : V(G) \rightarrow [n]$, and a pair $(t,S)$ that witnesses that $(X,(T,\bag))$ is not $d$-linked into $W$, and in time $k^{\OO(1)} (|V(T)| + m)$ returns a torso tree decomposition $(X', (T', \bag'))$ that covers $W$, has width at most $k$, has at most $|V(T)|$ nodes, and has $\phi_d(X') < \phi_d(X)$.
\end{lemma}
\begin{proof}
After a $k^{\OO(1)} m$ time flow computation we may assume that $S$ is a minimum size $(\bag(t), W)$-separator, because if $S$ was not a minimum size $(\bag(t), W)$-separator then any minimum size $(\bag(t), W)$-separator also witnesses that $\bag(t)$ is not $d$-linked into $W$.
This implies that $S$ is linked into $\bag(t)$.

Let $A = \reach_G(W, S)$ and $B = V(G) \setminus (A \cup S)$.
Note that $W \subseteq A \cup S$ and $\bag(t) \subseteq B \cup S$.
Denote $X' = (X \cap A) \cup S$.
We apply the pulling lemma (\Cref{lem:pull}) with the torso tree decomposition $(X, (T, \bag))$, the separation $(A, S, B)$, and the node $t$ to construct a torso tree decomposition $(X', (T', \bag'))$ of width at most $k$ and at most $|V(T)|$ nodes.
As $W \subseteq X$ and $W \subseteq A \cup S$, we have that $W \subseteq X'$, so $(X', (T', \bag'))$ covers $W$.
It remains to prove that $\phi_d(X') < \phi_d(X)$.

Because $\bag(t) \subseteq S \cup B$ and $\bag(t) \subseteq X$, we have that $|X'| \le |X| - |\bag(t)| + |S|$ and $d(X') \le d(X) - d(\bag(t)) + d(S)$.
Therefore, if $|S|<|\bag(t)|$, then $|X'|<|X|$, implying $\phi_d(X') < \phi_d(X)$ because $d(S) < n(k+1)$.
If $|S|=|\bag(t)|$ and $d(S) < d(\bag(t))$, then $|X'| \le |X|$ and $d(X') < d(X)$,  implying $\phi_d(X') < \phi_d(X)$.
\end{proof}

Then, our goal is to show that either a torso tree decomposition $(X, (T_X,\bag_X))$ of width $k-1$ that covers a largest bag $W$ of a tree decomposition $(T, \bag)$ of width $k$ can be used to improve $(T, \bag)$, or we find a pair $(t,S)$ witnessing that $(X, (T_X,\bag_X))$ is not $d$-linked into $W$ for a certain function $d$, in which case we can improve $(X, (T_X,\bag_X))$ by applying \Cref{lem:make_d_linkedv2}.

Let $(T, \bag)$ be a tree decomposition of $G$ and $r \in V(T)$ a designated root-node of it.
For a vertex $v \in V(G)$, let $f_v \in V(T)$ be the node of $T$ with $v \in \bag(f_v)$ that has the smallest distance to the root $r$ in $T$ among all nodes whose bags contain $v$, that is, $f_v$ is the forget-node of $v$.
We define a weight function $d_{(T, \bag, r)} : V(G) \rightarrow [|V(T)|]$ for a vertex $v \in V(G)$ as the distance from $f_v$ to $r$ plus one.
Next we prove the main lemma of this section.

\begin{lemma}
\label{lem:imprmainv2}
Let $(T, \bag)$ be a tree decomposition of $G$ of width $k$, and $r$ a node of $(T, \bag)$ with $\bag(r) = W$ with $|W| = k+1$.
There is an algorithm that given a torso tree decomposition $(X, (T_X, \bag_X))$ that covers $W$ and has width at most $k-1$, in time $k^{\OO(1)} (|V(T)|+|V(T_X)|+m)$ either
\begin{enumerate}
\item \label{lem:imprmainv2:good} constructs a tree decomposition of $G$ of width at most $k$, having strictly less bags of size $k+1$ than $(T,\bag)$, and having at most $n$ nodes, or
\item \label{lem:imprmainv2:bad} returns a pair $(t, S)$ where $t \in V(T_X)$ and $S \subseteq V(G)$ that witnesses that $(X, (T_X, \bag_X))$ is not $d_{(T, \bag, r)}$-linked into $W$.
\end{enumerate}
\end{lemma}
\begin{proof}
We treat $(T,\bag)$ as rooted on the node $r$.
Our goal is to construct a tree decomposition $(T', \bag')$ of $G$, and then show that if it does not satisfy the conditions of \Cref{lem:imprmainv2:good}, then we find the pair $(t,S)$ of \Cref{lem:imprmainv2:bad}.

First, for every connected component $C$ of $G \setminus X$, we will construct a tree decomposition $(T_C, \bag_C)$ of $N[C]$, so that $N(C)$ is in the root bag of $(T_C, \bag_C)$.
We again use $f_v \in V(T)$ to denote the forget-node of $v$ in $(T, \bag)$.
For a node $t \in V(T)$, denote by $t^{N(C)}$ the vertices 
\[t^{N(C)} = \{v \in N(C) \mid f_v \text{ is a strict descendant of } t \text{ in } T\}.\]
To construct the tree decomposition $(T_C, \bag_C)$, we first set
\[T_C = T[\{t \in V(T) \mid C \cap \bag(t) \neq \emptyset\}],\]
i.e., $T_C$ is the subtree of $T$ induced by bags that intersect $C$.
Observe that $T_C$ is connected because $G[C]$ is connected.
Then for each $t \in V(T_C)$ we set
\[\bag_C(t) = (\bag(t) \cap N[C]) \cup t^{N(C)}.\]
We let the root node of $(T_C, \bag_C)$ to be the node $r_C \in V(T_C)$ that is the closest to $r$ in $T$.
Note that because $T_C$ is a connected subtree of $T$, the node $r_C$ is uniquely defined.

\begin{claim}
\label{lem:is_tdv2}
It holds that $(T_C, \bag_C)$ is a tree decomposition of $N[C]$ and $N(C) \subseteq \bag_C(r_C)$.
\end{claim}
\begin{proof}
First, for the vertices $C$ and edges in $G[C]$ the decomposition clearly satisfies the vertex and edge conditions because $(T, \bag)$ satisfied the conditions.
For edges between $C$ and $N(C)$ and vertices in $N(C)$, note that again each such edge must be in a bag that intersects $C$, and because for every vertex of $N(C)$ there exists such an edge we have that every vertex of $N(C)$ must occur in some bag that intersects $C$.
The decomposition satisfies the connectedness condition for vertices in $C$ directly by the connectedness condition of $(T,\bag)$.

For vertices $v \in N(C)$, either (1) $v \in \bag(r_C)$ and $v$ is not in $t^{N(C)}$ for any $t \in V(T_C)$, or (2) $f_v \in V(T_C) \setminus \{r_C\}$ and $v \in t^{N(C)}$ for all $t$ on the path from the parent of $f_v$ to the root $r_C$.
Therefore, the connectedness condition is maintained for vertices in $N(C)$.
This also shows that $N(C) \subseteq \bag_C(r_C)$, which finally implies the edge condition also for edges in $G[N(C)]$.
\end{proof}

Now, our complete construction of $(T', \bag')$ is to attach the tree decompositions $(T_C, \bag_C)$ for all components $C$ of $G \setminus X$ from their roots to the tree decomposition $(T_X, \bag_X)$.
Because $N(C)$ is a clique in $\torso(X)$, the decomposition $(T_X, \bag_X)$ contains a bag containing $N(C)$ to which $(T_C, \bag_C)$ can be attached.

Next we show that this construction can be implemented in $k^{\OO(1)} (|V(T)|+|V(T_X)|+m)$ time.
In particular, first, the connected components $C$ and their neighborhoods can be found in $k^{\OO(1)} m$ time.
Then, we observe that the sum of $|V(T_C)|$ over all components $C$ is at most $(k+1) |V(T)|$ because $(T, \bag)$ has width $k$ and the components $C$ are disjoint.
By first computing pointers from vertices of $G$ to bags containing them, and then using the fact that $|N(C)| \le k+1$, each tree decomposition $(T_C, \bag_C)$ can be constructed in $k^{\OO(1)} |V(T_C)|$ time, which sums up to $k^{\OO(1)} |V(T)|$.
Then, it remains to attach each tree decomposition $(T_C, \bag_C)$ to a node of $(T_X, \bag_X)$ whose bag contains $N(C)$.
For this, observe that if we consider $(T_X, \bag_X)$ rooted, and for $v \in N(C)$ denote by $f^X_v$ the forget-node of $v$ in $(T_X, \bag_X)$, then $N(C)$ is contained in the bag of the node $f^X_v$ for $v \in N(C)$ such that $f^X_v$ maximizes the distance from the root.

Next we give the main argument for extracting the witness of \Cref{lem:imprmainv2:bad} if $(T', \bag')$ does not satisfy \Cref{lem:imprmainv2:good}.

\begin{claim}
\label{lem:width_constrv2}
Let $C$ be a component of $G \setminus X$ and $x \in V(T_X)$ a node of $(T_X, \bag_X)$ with $N(C) \subseteq \bag_X(x)$.
For every node $t \in V(T_C)$ we have either that
\begin{enumerate}
\item \label{lem:width_constrv2:good} $|\bag_C(t)| < |\bag(t)|$ or $\bag_C(t) = \bag(t)$, or that
\item $(\bag_X(x) \setminus t^{N(C)}) \cup (\bag(t) \setminus N[C])$ witnesses that $\bag_X(x)$ is not $d_{(T, \bag, r)}$-linked into $W$.
\end{enumerate}
\end{claim}
\begin{proof}
\Cref{lem:width_constrv2:good} is true if $t^{N(C)}$ is empty, so suppose $t^{N(C)}$ is non-empty and $|\bag_C(t)| \ge |\bag(t)|$.
By the definition of $\bag_C(t)$ this implies that $|t^{N(C)}| \ge |\bag(t) \setminus N[C]|$.
Note that $t^{N(C)} \subseteq \bag_X(x)$.
We will show that in this case 
\[S = (\bag_X(x) \setminus t^{N(C)}) \cup (\bag(t) \setminus N[C])\] separates $\bag_X(x)$ from $W$.
Therefore $S$ witnesses that $\bag_X(x)$ is not $d_{(T, \bag, r)}$-linked into $W$, because by $|t^{N(C)}| \ge |\bag(t) \setminus N[C]|$ we have that $|S| \le |\bag_X(x)|$, and moreover we have $d_{(T, \bag, r)}(S) < d_{(T, \bag, r)}(\bag_X(x))$ because for every vertex $v_1 \in t^{N(C)}$ and $v_2 \in \bag(t)$, it holds that $d_{(T, \bag, r)}(v_1) > d_{(T, \bag, r)}(v_2)$ because $f_{v_1}$ is a strict descendant of $t$, and $f_{v_2}$ is an ancestor of $t$.

To show that $S$ separates $\bag_X(x)$ from $W$, it is sufficient to show that it separates $t^{N(C)}$ from $W$ because $\bag_X(x) \setminus S = t^{N(C)}$.
Consider a shortest path in $G \setminus S$ that starts in $t^{N(C)}$ and ends in $W$.
If this path would intersect $N[C]$ anywhere else than in its first vertex, then it would intersect $t^{N(C)}$ twice because $N(C) \setminus S = t^{N(C)}$ and  $W \cap C = \emptyset$, which would contradict that it is a shortest path.
Therefore, it intersects $N[C]$ only in its first vertex.
Then, because for each $v \in t^{N(C)}$ the node $t \in V(T)$ separates $f_v$ from $r$ in $T$, it holds that $\bag(t)$ separates $t^{N(C)}$ from $\bag(r) = W$.
Therefore, the path must intersect $\bag(t)$, and therefore as $\bag(t)$ and $t^{N(C)}$ are disjoint, it must intersect $\bag(t) \setminus N[C]$.
However, $\bag(t) \setminus N[C] \subseteq S$, and therefore no such path exists in $G \setminus S$.
\end{proof}

Now, for all nodes of the constructed decompositions $(T_C, \bag_C)$ we check if \Cref{lem:width_constrv2:good} of \Cref{lem:width_constrv2} holds, and if it does not hold we return the pair $(x, (\bag_X(x) \setminus t^{N(C)}) \cup (\bag(t) \setminus N[C]))$.
This can be done in $k^{\OO(1)} |V(T')| = k^{\OO(1)} |V(T)|$ time.

Then, it remains to prove that if \Cref{lem:width_constrv2:good} of \Cref{lem:width_constrv2} holds for all nodes of all decompositions $(T_C, \bag_C)$, then $(T', \bag')$ has width at most $k$ and has strictly less bags of size $k+1$ than $(T,\bag)$.
First, clearly $(T', \bag')$ has width at most $k$ as none of the decompositions $(T_C, \bag_C)$ have larger width than $(T,\bag)$ and $(T_X, \bag_X)$ has smaller width than $(T,\bag)$.
It remains to prove that $(T', \bag')$ has less bags of size $k+1$ than $(T, \bag)$.

Consider any node $t \in V(T)$, and suppose that there are two distinct components $C_1$ and $C_2$ of $G \setminus X$ so that both $C_1$ and $C_2$ intersect $\bag(t)$ and $|\bag_{C_1}(t)| = |\bag_{C_2}(t)| = |\bag(t)|$.
Now, by \Cref{lem:width_constrv2:good} of \Cref{lem:width_constrv2} it would hold that $\bag_{C_1}(t) = \bag_{C_2}(t) = \bag(t)$.
However, as $\bag_{C_1}(t) \subseteq N[C_1]$, this would contradict that $\bag(t)$ intersects $C_2$.
%Therefore, as $\bag_{C_1}(t) \subseteq N[C_1]$ and $\bag_{C_2}(t) \subseteq N[C_2]$, it holds that $\bag(t) \subseteq N(C_1) \cap N(C_2)$, which implies that $\bag(t)$ is a clique in $\torso(X)$, and therefore has size at most $k$.
Therefore, for any node $t \in V(T)$
%with a bag of size $|\bag(t)| = k+1$,
there is at most one corresponding node $t$ in the decompositions $(T_C, \bag_C)$ across all components $C$ with a bag of size $|\bag_C(t)| = |\bag(t)|$.
For the root node $r$, as $\bag(r) \subseteq X$, none of the components $C$ intersect $\bag(r)$, and therefore no decomposition $(T_C, \bag_C)$ contains a node corresponding to it.
All other bags of $(T', \bag')$ come from $(T_X, \bag_X)$ and have size at most $k$, so as $|\bag(r)| = k+1$, it follows that $(T', \bag')$ has strictly less bags of size $k+1$ than $(T, \bag)$.

Finally, by \Cref{lem:tdlinear} we can reduce the number of nodes of $(T', \bag')$ to at most $n$ within the same time.
\end{proof}

Then, we combine \Cref{lem:make_d_linkedv2,lem:imprmainv2} into a single lemma showing that to improve $(T,\bag)$ it is sufficient to find a torso tree decomposition in $G$ that covers a largest bag of $(T,\bag)$ and has width smaller than $(T,\bag)$.

\begin{lemma}
\label{lem:main_impr_lem}
Let $(T, \bag)$ be a tree decomposition of $G$ of width $k$ and $|V(T)| \le n$, and $r$ a node of $(T, \bag)$ with $\bag(r) = W$ with $|W| = k+1$.
There is an algorithm that given a torso tree decomposition $(X, (T_X, \bag_X))$ that covers $W$ and has width at most $k-1$, in time $k^{\OO(1)}(|V(T_X)|+n^3)$ constructs a tree decomposition of $G$ of width at most $k$, having strictly less bags of size $k+1$ than $(T, \bag)$, and having at most $n$ nodes.
\end{lemma}
\begin{proof}
First, we apply \Cref{lem:tdlinear} to reduce the number of nodes of $(T_X, \bag_X)$ to at most $n$.
Then, we repeatedly apply \Cref{lem:imprmainv2} together with \Cref{lem:make_d_linkedv2}, in particular, if \Cref{lem:imprmainv2} returns the tree decomposition of \Cref{lem:width_constrv2:good} we are done, and if it returns a pair $(t,S)$ that witnesses that $(X,(T_X,\bag_X))$ is not $d_{(T,\bag,r)}$-linked into $W$ then we apply \Cref{lem:make_d_linkedv2}, which decreases $\phi_{d_{(T,\bag,r)}}(X)$ by at least one.
Because $\phi_{d_{(T,\bag,r)}}(X)$ is initially $\OO(kn^2)$ and $\phi_{d_{(T,\bag,r)}}(X)$ must be non-negative, the total number of iterations is at most $\OO(kn^2)$, giving a total running time of $k^{\OO(1)} n^3$, plus $k^{\OO(1)}|V(T_X)|$ from the application of \Cref{lem:tdlinear}.
\end{proof}

\subsection{Reducing treewidth to \stw}
Now we can prove \Cref{the:stweximpl}, in particular that algorithms for \stw imply algorithms for treewidth (for definition of \stw see \Cref{subsec:overredustw}).
Recall that the running time function $T(k)$ is assumed to be increasing on $k$ and we denote $m = |V(G)|+|E(G)|$.

\eximpltheorem*
\begin{proof}
Let $G$ denote the input graph.
First, we use the 4-approximation algorithm of \cite{RobertsonS-GMXIII} to obtain a tree decomposition $(T, \bag)$ of $G$ of width at most $4k+3$ in time $2^{\OO(k)} n^2$ and polynomial space or to return that the treewidth of $G$ is larger than $k$.
By \Cref{lem:tdlinear}, within the same running time we assume that $|V(T)| \le n$, and we can also assume that $m \le \OO(kn)$ because otherwise the treewidth of $G$ would be larger than $k$.

Then, we repeat the following process as long as the width of $(T, \bag)$ is larger than $k$.
Let $W$ be a largest bag of $(T,\bag)$, and note that in this case $|W| \ge k+2$.
We use the algorithm for \stw to either get a torso tree decomposition that covers $W$ and has width at most $|W|-2$ or to 
conclude that the treewidth of $G$ is larger than $|W|-2 \ge k$.
If we conclude that the treewidth of $G$ is larger than $k$ we are ready and can immediately return.
%report that no such torso tree decomposition exists.
%If no such torso tree decomposition exists, we return that the treewidth of $G$ is larger than $k$.
%This is correct because any tree decomposition of $G$ of width $k$ would be also a torso tree decomposition that covers $W$ by setting $X = V(G)$ because $\torso_G(V(G)) = G$.
If the algorithm returns such a torso tree decomposition, we apply \Cref{lem:main_impr_lem} to improve $(T,\bag)$, in particular to decrease the number of bags of size $|W|$ and not increase the width.

We can decrease the number of largest bags while not increasing the width at most $\OO(kn)$ times before the width decreases from $4k+3$ to $k$, and therefore the algorithm works with $\OO(kn)$ applications of the algorithm for \stw and \Cref{lem:main_impr_lem}.
In all of the applications, the parameter $k$ for \stw is at most $4k+2$, where $k$ is the original parameter for treewidth.
This results in a total running time of $T(\OO(k)) \cdot \OO((nk)^{c+1}) + k^{\OO(1)} n^4 + 2^{\OO(k)} n^2$.
\end{proof}

We then turn to \Cref{the:stwapximpl}, in particular, to proving that algorithms for \pstw imply approximation algorithms for treewidth (for the definition of \pstw see \Cref{subsec:overredustw}).
The crucial lemma for this will the following.

\begin{lemma}
\label{lem:wpartapx}
Let $G$ be a graph of treewidth at most $k$, $\varepsilon \in (0,1)$ a rational, and $W \subseteq V(G)$ a set of vertices of size $|W| \le 4k+4$.
There exists a partition of $W$ into $t = \OO(1/\varepsilon)$ parts $W_1, \ldots, W_t$, so that after making each part into a clique the treewidth of $G$ is at most $k+\varepsilon k$.
\end{lemma}
\begin{proof}
If $\varepsilon < 1/k$ we can return the trivial partition of $W$ into single vertices.
Therefore we can assume that $\varepsilon k \ge 1$.

Consider a rooted tree decomposition $(T, \bag)$ of $G$ of width $k$.
By turning $(T, \bag)$ into a ``nice tree decomposition'', we can assume that the root bag of $(T,\bag)$ is empty, each node of $T$ has at most two children, and that $|\bag(u) \setminus \bag(v)| + |\bag(v) \setminus \bag(u)| \le 1$ holds for any two adjacent nodes $u,v \in V(T)$ (see e.g.~\cite[Chapter~7]{cygan2015parameterized}).
Recall that a node $t \in V(T)$ with a parent $p \in V(T)$ is a forget-node of a vertex $v \in V(T)$ if $v \in \bag(t) \setminus \bag(p)$.
Respectively, such $v$ is a forget-vertex of $t$.
Note that each node of $T$ has at most one forget-vertex and each vertex of $G$ has exactly one forget-node.
By further stretching $(T,\bag)$ we can also assume that each forget-node has exactly one child.
We say that a node is a $W$-forget node if it is a forget-node of a vertex $w \in W$.

Let us process $(T, \bag)$ from the leaves towards the root, i.e., in an order of a post-order traversal, and maintain a set of ``removed'' nodes $R \subseteq V(T)$.
Suppose we are processing a node $t$ and let $D \subseteq V(T)$ be the nodes of $T$ that are descendants of $t$ and reachable from $t$ in $T \setminus R$.
Note that $t \in D$ and $D \subseteq V(T) \setminus R$.
Now, if $D$ contains at least $\varepsilon k/2$ $W$-forget-nodes or $t$ is the root we add a part to the partition of $W$ and modify the tree decomposition as follows.
We let $W' \subseteq W$ be the vertices in $W$ whose forget-nodes are in $D$.
We add $W'$ as a part of the partition, and add $W'$ to the bags of all nodes in $D$.
Then, we add all nodes in $D$ to $R$.

Observe that $|W'| \le \varepsilon k$ follows from the facts that we process the tree in post-order, each node can have at most two children, each node can be a forget-node of at most one vertex, each forget-node has one child, and $\varepsilon k \ge 1$.
Therefore, the sizes of the bags of nodes in $D$ increased by at most $\varepsilon k$, and moreover they will not increase again because they were added to $R$.
Therefore, the resulting tree decomposition has width at most $k+\varepsilon k$.
We also observe that the resulting tree decomposition is indeed a tree decomposition after making such $W'$ into a clique: All the new edges are contained in the bags of all nodes in $D$, and the subtree condition is maintained because the forget-nodes of vertices in $W'$ are in $D$.

Now, each created part of the partition except the part corresponding to the root has size at least $\varepsilon k/2$, so in total the number of parts is at most $|W|/(\varepsilon k/2)+1 \le \frac{16}{\varepsilon}+1 = \OO(1/\varepsilon)$.
\end{proof}

Now, by using \Cref{lem:wpartapx} we can prove \Cref{the:stwapximpl} similarly to \Cref{the:stweximpl}.

\apximpltheorem*
\begin{proof}
Let $G$ denote the input graph.
First, we use the 4-approximation algorithm of \cite{RobertsonS-GMXIII} to obtain a tree decomposition $(T, \bag)$ of $G$ of width at most $4k+3$ in time $2^{\OO(k)} n^2$ and polynomial space or to return that the treewidth of $G$ is larger than $k$.
By \Cref{lem:tdlinear}, within the same running time we assume that $|V(T)| \le n$, and we can also assume that $m \le \OO(kn)$ because otherwise the treewidth of $G$ would be larger than $k$.

Then, we repeat the following process as long as the width of $(T, \bag)$ is larger than $k+\varepsilon k$.
Let $W$ be a largest bag of $(T,\bag)$, and note that in this case $|W| \ge k+\varepsilon k+2$ and $|W| \le 4k+4$.
We try all partitions of $W$ into $t = \OO(1/\varepsilon)$ parts (where the bound for $t$ is from \Cref{lem:wpartapx}).
For each partition $W_1, \ldots, W_t$, we make the parts $W_1, \ldots, W_t$ into cliques in $G$, and then use the algorithm for \pstw with this partition of $W$.
By \Cref{lem:wpartapx}, there exists such a partition so that after making $W_1, \ldots, W_t$ into cliques the treewidth of $G$ is at most $k+\varepsilon k$, and therefore if the algorithm for \pstw returns for every partition that the treewidth of $G$ is larger than $|W|-2 \ge k+\varepsilon k$, we can return that the treewidth of $G$ is larger than $k$.
%a negative answer for every partition we can return that the treewidth of $G$ is larger than $k$.
Otherwise, the algorithm for \pstw returned a torso tree decomposition that covers $W$ and has width at most $|W|-2$, and we proceed applying \Cref{lem:main_impr_lem} similarly as in the proof of \Cref{the:stweximpl}.

The running time follows from the fact that there are at most $t^{\OO(k)} = (1+1/\varepsilon)^{\OO(k)}$ partitions of $W$ into $t = \OO(1/\varepsilon)$ parts, and we can decrease the number of largest bags while not increasing the width at most $\OO(nk)$ times, and therefore there we use in total $\OO(nk) \cdot (1+1/\varepsilon)^{\OO(k)}$ applications of the algorithm for \pstw with $t = \OO(1/\varepsilon)$.
The parameter $k$ for \pstw is at most $4k+2$, where $k$ is the original parameter for treewidth.
\end{proof}

\section{Algorithm for \pstw}
\label{sec:algopstw}
This section is devoted to proving \Cref{the:stwpalg}, in particular to giving a $k^{\OO(kt)} nm$ time algorithm for \pstw.
We start by giving a slightly more general definition of \pstw than was given in \Cref{subsec:overredustw}.

An instance of \pstw is a triple $\ins = (G, \{W_1, \ldots, W_t\}, k)$, where $G$ is a graph, $t$ and $k$ positive integers with $t \le k+2$, and $\{W_1, \ldots, W_t\}$ is a set of $t$ terminal cliques, each $W_i \subseteq V(G)$ being a clique of size at most $k+1$ in $G$.
We denote the union of the terminal cliques by $\allw_{\ins} = \bigcup_{i=1}^t W_i$.
%The terminal cliques can overlap, in particular, unlike in the definition of \Cref{sec:overview}\todo{??}, $\{W_1, \ldots, W_t\}$ is not necessarily a partition of $\allw_{\ins}$.
Note that unlike in the definition of \Cref{subsec:overredustw}, we do not require that $|\allw_{\ins}| = k+2$ here, and in particular while $|\allw_{\ins}| = k+2$ holds initially, the set $|\allw_{\ins}|$ can become larger in the recursive calls (but $t$ or $k$ will not increase).

A solution of an instance of \pstw is a torso tree decomposition $(X, (T, \bag))$ that covers $\allw_{\ins}$ and has width at most $k$.
We say that an instance $\ins$ is a \emph{yes-instance} if there exists a solution of it and a \emph{no-instance} otherwise.
Our algorithm will either return a solution or conclude that $\ins$ is a no-instance, in particular, it will not use the freedom in the definition to return that the treewidth of $G$ is larger than $k$ without concluding that $\ins$ is a no-instance.

The rest of this section is organized as follows.
In \Cref{subsec:algoinvas} we introduce a measure called \emph{flow potential} for quantifying the progress of the algorithm.
In \Cref{subsec:safesep} we give a reduction rule to simplify instances by safe separations.
In \Cref{subsec:branching} we give the two branching rules of the algorithm, and in \Cref{subsec:algotogether} we describe the algorithm and put together its correctness proof.
Finally, in \Cref{subsec:algo1tc} we analyze the running time.

\subsection{Flow potential}
\label{subsec:algoinvas}
We define the measure \emph{flow potential} that will be used for quantifying the progress of the algorithm.

Let $G$ be a graph and $X,Y \subseteq V(G)$ two sets of vertices.
If $X \subset V(G)$, then the flow potential from $X$ to $Y$, denoted by $\flp_G(X,Y)$ is the minimum order of a separation $(A,S,B)$, where $X \subseteq A \cup S$, $Y \subseteq B \cup S$, and $B \neq \emptyset$.
If $X = V(G)$, then $\flp_G(X,Y) = |X|$.
We observe that
\[\flow_G(X,Y) \le \flp_G(X,Y) \le |X|.\]
Note that $\flow_G(X,Y) < \flp_G(X,Y)$ if and only if $|X| > |Y|$, $Y$ is strictly linked into $X$, and $X$ intersects all connected components of $G \setminus Y$.
We will also use the property that if $Y_1 \subseteq Y_2$, then $\flp_G(X,Y_1) \le \flp_G(X,Y_2)$.

Let $\ins = (G, \{W_1, \ldots, W_t\}, k)$ be an instance.
For a terminal clique $W_i$, we denote by $\ow_\ins(W_i) = \bigcup_{j \in [t] \setminus \{i\}} W_j$ the union of the other terminal cliques in $\ins$.
We define the flow potential of $W_i$ in $\ins$ to be 
\[\flp_{\ins}(W_i) = \flp_G(W_i, \ow_\ins(W_i)).\]
Note that $\flp_{\ins}(W_i) \neq \flow_G(W_i, \ow_\ins(W_i))$ can hold only if $W_i$ is the unique largest terminal clique of $\ins$.

\subsection{Safe separations}
\label{subsec:safesep}
In this subsection we introduce a reduction rule based on identifying a strict separation $(A,S,B)$, making $S$ a clique and enforcing $S \subseteq X$, and then solving the different sides of $S$ independently of each other.
In particular, we show that this reduction is safe if $(A,S,B)$ satisfies certain conditions that we define next.

\begin{definition}[Safe separation]
\label{def:safesep}
Let $\ins = (G, \{W_1, \ldots, W_t\}, k)$ be an instance and $(A, S, B)$ a strict separation of $G$.
We say that $(A, S, B)$ is a safe separation of $\ins$ if there exists terminal cliques $W_a$ and $W_b$ (possibly $a=b$) so that $S$ is linked into $(A \cup S) \cap W_a$ and into $(B \cup S) \cap W_b$.
\end{definition}

We say that such terminal cliques $W_a$ and $W_b$ are the \emph{spanning terminal cliques} of the safe separation.
Note that safe separations can be classified into two types: those where $S$ is a subset of some terminal clique $W_i$ and we can take $W_a = W_b = W_i$, and those where $a \neq b$ and $S$ is a minimum size $(W_a,W_b)$-separator.
The purpose of our definition is to be a common generalization of these two types.

Next we introduce notation for reduction by safe separations.
In \Cref{sec:fasttw} this notation will also be used with separations that are not necessarily safe.

Let $\ins = (G, \{W_1, \ldots, W_t\}, k)$ be an instance and $(A, S, B)$ a separation with $|S|\le k+1$.
We denote by $\{W_1, \ldots, W_t\} \rescliqs (A,S)$ the set obtained from $\{W_1, \ldots, W_t\}$ by first removing each $W_i$ with $W_i \cap A = \emptyset$, and then inserting $S$ if no superset of $S$ is present.
Recall that $G \mcliq S$ denotes the graph obtained from $G$ by making $S$ a clique.
We define $G \rescliqs (A,S) = G[A \cup S] \mcliq S$, and then $\ins \rescliqs (A,S) = (G \rescliqs (A,S), \{W_1, \ldots, W_t\} \rescliqs (A, S), k)$.

Note that if $(A,S,B)$ is a safe separation, then $|S|\le k+1$ because $S$ is linked into a spanning terminal clique $W_a$ and $|W_a| \le k+1$.
Also, observe that if $(A,S,B)$ is a safe separation, then $\ins \rescliqs (A,S)$ has at most $t$ terminal cliques.
This is because if all terminal cliques of $\ins$ intersect $A$, then $(A,S,B)$ can be a safe separation only if the spanning terminal clique $W_b$ is a superset of $S$.

The reduction rule used in the algorithm will be that if there exists a safe separation $(A,S,B)$, then solve the instances $\ins \rescliqs (A,S)$ and $\ins \rescliqs (B,S)$ independently of each other.
If either of them returns \no, then return \no, and if both of them return a solution, denoted by $(X_A, (T_A, \bag_A))$ and $(X_B, (T_B, \bag_B))$, respectively, then return the solution obtained by combining these solutions on the separator $S$.
More formally, if $(T_A, \bag_A)$ and $(T_B, \bag_B)$ are tree decompositions that both contain a bag containing $S$, then we denote by $(T_A, \bag_A) \cup_S (T_B, \bag_B)$ the tree decomposition obtained by taking the disjoint union of $(T_A, \bag_A)$ and $(T_B, \bag_B)$ and connecting them by an edge between bags containing $S$.
Then, the combined solution is denoted by $(X_A \cup X_B, (T_A, \bag_A) \cup_S (T_B, \bag_B))$.

In the next two lemmas we show that this reduction is correct.
We start by proving that if both $\ins \rescliqs (A,S)$ and $\ins \rescliqs (B,S)$ return a solution, then the constructed solution is a solution of $\ins$.
This holds in fact for any separation $(A,S,B)$.

\begin{lemma}
\label{lem:sepredgen}
Let $\ins = (G, \{W_1, \ldots, W_t\}, k)$ be an instance and $(A,S,B)$ a separation.
If $(X_A, (T_A, \bag_A))$ is a solution of $\ins \rescliqs (A,S)$ and $(X_B, (T_B, \bag_B))$ a solution of $\ins \rescliqs (B,S)$, then $(X_A \cup X_B, (T_A, \bag_A) \cup_S (T_B, \bag_B))$ is a solution of $\ins$.
\end{lemma}
\begin{proof}
First, note that both $\{W_1, \ldots, W_t\} \rescliqs (A,S)$ and $\{W_1, \ldots, W_t\} \rescliqs (B,S)$ contain a superset of $S$, so both $X_A$ and $X_B$ are supersets $S$.
Also, any terminal clique in $\{W_1, \ldots, W_t\}$ is either in $\{W_1, \ldots, W_t\} \rescliqs (A,S)$, in $\{W_1, \ldots, W_t\} \rescliqs (B,S)$, or is a subset of $S$, so $X_A \cup X_B$ is a superset of $\allw_{\ins}$.

Because $S$ is a clique in both $G \rescliqs (A,S)$ and in $G \rescliqs (B,S)$, it is contained in a bag of both $(T_A, \bag_A)$ and $(T_B, \bag_B)$, and therefore the construction $(T_A, \bag_A) \cup_S (T_B, \bag_B)$ is well-defined.
Because $G \rescliqs (A,S)$ and $G \rescliqs (B,S)$ intersect only in $S$, it holds that $X_A \cap X_B = S$, which implies that $(T_A, \bag_A) \cup_S (T_B, \bag_B)$ satisfies the connectedness condition to be a tree decomposition of $\torso_G(X_A \cup X_B)$.
The vertex condition is trivially satisfied.
For the edge condition, consider an edge $uv \in E(\torso_G(X_A \cup X_B))$.
Because $(A,S,B)$ is a separation and $S \subseteq X_A \cup X_B$, it must hold that $u,v \in A \cup S$ or $u, v \in B \cup S$.
By symmetry, assume that $u,v \in A \cup S$.
The internal vertices of the path between $u$ and $v$ must be in $A$, and therefore $uv \in E(\torso_{G[A \cup S]}(X_A))$, implying that $uv$ is in some bag of $(T_A, \bag_A)$.
Therefore $(T_A, \bag_A) \cup_S (T_B, \bag_B)$ satisfies the edge condition.
\end{proof}

We then show the other direction of correctness by applying the pulling lemma (\Cref{lem:pull}).

\begin{lemma}
\label{lem:safesepyesyes}
Let $\ins = (G, \{W_1, \ldots, W_t\}, k)$ be an instance and $(A,S,B)$ a safe separation.
If $\ins$ is a yes-instance, then both $\ins \rescliqs (A,S)$ and $\ins \rescliqs (B,S)$ are yes-instances.
\end{lemma}
\begin{proof}
By the symmetry of the definition of safe separation, it suffices to show that $\ins \rescliqs (A,S)$ is a yes-instance.

Let $W_a$ and $W_b$ be the spanning terminal cliques of $(A,S,B)$ and $(X, (T, \bag))$ a solution of $\ins$.
Because $W_b$ is a clique in $G$ and is contained in $X$, there exists a node $r \in V(T)$ with $W_b \subseteq \bag(r)$, which implies that $S$ is linked into $\bag(r) \cap (B \cup S)$.
We use the pulling lemma (\Cref{lem:pull}) with the separation $(A, S, B)$ and the node $r$ to obtain a torso tree decomposition $((X \cap A) \cup S, (T', \bag'))$ of width at most $k$ containing a bag containing $S$.
By \Cref{lem:torsoindsub}, $((X \cap A) \cup S, (T', \bag'))$ is a torso tree decomposition also in $G[A \cup S]$, and because it covers $S$ and $(T', \bag')$ contains a bag containing $S$, it is also a torso tree decomposition in $G \rescliqs (A,S)$.
Because every terminal clique in $\{W_1, \ldots, W_t\} \rescliqs(A,S)$ is a subset of $(X \cap A) \cup S$, we have that $((X \cap A) \cup S, (T', \bag'))$ covers $\allw_{\ins \rescliqs (A,S)}$.
\end{proof}

Next we will give three lemmas arguing that the flow potentials ``behave well'' when breaking the instance by a safe separation.
In particular, our goal is to show that when breaking an instance $\ins$ by a safe separation $(A,S,B)$, the terminal cliques in $\ins \rescliqs (A,S)$ do not have smaller flow potentials than the corresponding terminal cliques in $\ins$, except in some limited cases.
Note that this is quite easy to see if we would measure flow instead of flow potential, and indeed we prove that the properties that naturally hold for flow also hold for flow potential.
We start by considering a situation where the breaking does not change the terminal cliques.

\begin{lemma}
\label{lem:sepsame}
Let $\ins = (G, \{W_1, \ldots, W_t\}, k)$ be an instance and $(A,S,B)$ a separation.
If $\{W_1, \ldots, W_t\} \rescliqs (A,S) = \{W_1, \ldots, W_t\}$, then $\flp_{\ins \rescliqs (A,S)}(W_i) \ge \flp_{\ins}(W_i)$ for all $W_i$.
\end{lemma}
\begin{proof}
Observe that in this case, $\ow_{\ins \rescliqs (A,S)}(W_i) = \ow_{\ins}(W_i)$, and therefore it suffices to prove that $\flp_{G}(W_i, \ow_{\ins}(W_i)) \le \flp_{G \rescliqs (A, S)}(W_i, \ow_{\ins}(W_i))$.

We argue by the definition of flow potential.
First, if $W_i = A \cup S$, this holds trivially, so assume that $W_i \subset A \cup S$.
Let $(A',S',B')$ be a separation in $G \rescliqs (A,S)$ of minimum order so that $W_i \subseteq A' \cup S'$, $\ow_{\ins}(W_i) \subseteq B' \cup S'$, and $B'$ is non-empty.
Now, $S$ is a clique in $G \rescliqs (A,S)$, so either $S \subseteq A' \cup S'$ or $S \subseteq B' \cup S'$.
If $S \subseteq A' \cup S'$ then $(A' \cup B, S', B')$ is a separation in $G$, and if $S \subseteq B' \cup S'$ then $(A', S', B' \cup B)$ is a separation in $G$.
In either case, $\flp_G(W_i, \ow_{\ins}(W_i)) \le |S'| = \flp_{G \rescliqs (A,S)}(W_i, \ow_{\ins}(W_i))$.
\end{proof}

We then argue that for most of the terminal cliques the flow potential does not decrease when going from $\ins$ to $\ins \rescliqs (A,S)$.
In particular, all except one terminal cliques of $\ins \rescliqs (A,S)$ satisfy the conditions to be the terminal clique $W_i$ in the next lemma.

\begin{lemma}
\label{lem:sepflp}
Let $\ins = (G, \{W_1, \ldots, W_t\}, k)$ be an instance, $(A,S,B)$ a separation, and $W_i$ a terminal clique of both $\ins$ and $\ins \rescliqs (A,S)$ so that there exists some other terminal clique $W_j \neq W_i$ of $\ins \rescliqs (A,S)$ with $W_j \supseteq S$.
Then, $\flp_{\ins \rescliqs (A, S)}(W_i) \ge \flp_{\ins}(W_i)$.
\end{lemma}
\begin{proof}
We observe that because $W_j \supseteq S$, it holds that $\ow_{\ins \rescliqs (A,S)}(W_i) = (\ow_{\ins}(W_i) \cap A) \cup S$.
Therefore, it suffices to prove that
\[\flp_{G}(W_i, \ow_{\ins}(W_i)) \le \flp_{G \rescliqs (A, S)}(W_i, (\ow_{\ins}(W_i) \cap A) \cup S).\]

We argue by the definition of flow potential.
First, if $W_i = A \cup S$, the lemma holds trivially, so assume that $W_i \subset A \cup S$.
Let $(A',S',B')$ be a separation in $G \rescliqs (A,S)$ of minimum order so that $W_i \subseteq A' \cup S'$, $(\ow_{\ins}(W_i) \cap A) \cup S \subseteq B' \cup S'$, and $B'$ is non-empty.
Now, because $S \subseteq B' \cup S'$, we have that $(A', S', B' \cup B)$ is a separation in $G$.
Note that $\ow_{\ins}(W_i) \subseteq S' \cup B' \cup B$, so it follows that 
\[\flp_G(W_i, \ow_{\ins}(W_i)) \le |S'| = \flp_{G \rescliqs (A,S)}(W_i, (\ow_{\ins}(W_i) \cap A) \cup S).\]
\end{proof}

Then, we reduce the task of analyzing the total flow potentials of $\ins \rescliqs (A,S)$ into three cases.

\begin{lemma}
\label{lem:ssflp2}
Let $\ins = (G, \{W_1, \ldots, W_t\}, k)$ be an instance and $(A,S,B)$ a safe separation.
Now, either
\begin{enumerate}
\item \label{lem:ssflp2:ltq} $\ins \rescliqs (A,S)$ has less terminal cliques than $\ins$, or
\item \label{lem:ssflp2:etq} $\{W_1, \ldots, W_t\} \rescliqs (A,S) = \{W_1, \ldots, W_t\}$, or
\item \label{lem:ssflp2:nt} $S$ is a terminal clique of $\ins \rescliqs (A,S)$ but not of $\ins$, and there exists a terminal clique $W_i$ of $\ins$ with $W_i \cap B \neq \emptyset$ and $\flp_{\ins \rescliqs (A,S)}(S) \ge \flp_{\ins}(W_i)$.
\end{enumerate}
\end{lemma}
\begin{proof}
Let $W_a$ and $W_b$ be the spanning terminal cliques of $(A,S,B)$.
If $W_b \subseteq A \cup S$, then $W_b$ is also a superset of $S$ because $S$ is linked into $W_b \cap (S \cup B)$, and either \Cref{lem:ssflp2:ltq} or \Cref{lem:ssflp2:etq} holds.
Then, if $W_b$ intersects $B$ and also some other terminal clique intersects $B$, \Cref{lem:ssflp2:ltq} holds.
It remains to prove that if $W_b$ is the only terminal clique that intersects $B$ and $\ins \rescliqs (A,S)$ has the same number of terminal cliques as $\ins$, then \Cref{lem:ssflp2:nt} holds.

In this case, $S$ is a terminal clique of $\ins \rescliqs (A,S)$ but not of $\ins$, and $\ow_{\ins \rescliqs (A,S)}(S) = \ow_{\ins}(W_b)$, implying that it suffices to prove that
$\flp_{G}(W_b, \ow_{\ins}(W_b)) \le \flp_{G \rescliqs (A,S)}(S, \ow_{\ins}(W_b))$.
We argue by the definition of flow potential.
Because $(A,S,B)$ is a strict separation, it holds that $S \subset A \cup S$.
Let $(A',S',B')$ be a separation in $G \rescliqs (A,S)$ of minimum order so that $S \subseteq A' \cup S'$, $\ow_{\ins}(W_b) \subseteq B' \cup S'$, and $B'$ is non-empty.
Now, because $S \subseteq A' \cup S'$, we have that $(A' \cup B, S', B')$ is a separation in $G$.
As $W_b \subseteq A' \cup B \cup S'$, it follows that \[\flp_G(W_b, \ow_{\ins}(W_b)) \le |S'| = \flp_{G \rescliqs (A,S)}(S, \ow_{\ins}(W_b)).\]
\end{proof}

We then show that safe separations can be found efficiently.

\begin{lemma}
\label{lem:safeseptime}
There is a $k^{\OO(1)} m$ time algorithm for finding a safe separation or deciding that none exist.
\end{lemma}
\begin{proof}
We try all pairs of terminal cliques $W_a$ and $W_b$ and find safe separations whose spanning terminal cliques $W_a$ and $W_b$ are.

Let $(A,S,B)$ be a safe separation and $W_a$ and $W_b$ its spanning terminal cliques.
First, consider safe separations $(A,S,B)$ where $S \subseteq W_a$ or $S \subseteq W_b$.
Assume without loss of generality that $S \subseteq W_a$.
We can find such safe separations by checking if $G \setminus W_a$ has at least two connected components, and also trying all $w \in W_a$ and checking if $G \setminus (W_a \setminus \{w\})$ has at least two connected components.

Then, consider safe separations $(A,S,B)$ spanned by $W_a$ and $W_b$ so that $S$ is not a subset of $W_a$ or $W_b$.
In this case $W_a$ intersects $A$ and $W_b$ intersects $B$, and $S$ is a $(W_a,W_b)$-separator.
By the symmetry of the definition we may assume that $|W_a| \le |W_b|$.
If $W_a$ is strictly linked into $W_b$, then no such safe separators exists.
Then, if $W_a$ is not strictly linked into $W_b$, any minimum size $(W_a,W_b)$-separator $S$ with $S \neq W_a$ and $S \neq W_b$ corresponds to a safe separator, and can be found by standard flow computations in $k^{\OO(1)} m$ time.
\end{proof}

Then, we state the fact that applying safe separations makes all pairs of terminal cliques strictly linked into each other.

\begin{lemma}
\label{lem:strictlinked}
If $\ins$ has no safe separations, then for each pair of terminal cliques $W_i,W_j$ with $|W_i| \le |W_j|$, it holds that $W_i$ is strictly linked into $W_j$.
\end{lemma}
\begin{proof}
If $W_i$ would not be strictly linked into $W_j$, then the separator contradicting strict linkedness would give a safe separation.
\end{proof}

\subsection{Branching}
\label{subsec:branching}
We do two types of branching in our algorithm: terminal clique merging and leaf pushing.

\subsubsection{Terminal clique merging}
\label{subsec:tercliqmerg}
The first type of branching is that we guess that two terminal cliques $W_i$ and $W_j$ are in a same bag in some solution, and therefore we can actually merge them into one terminal clique.
We introduce some notation for this operation and analyze the flow potential under it.

Let $\ins = (G, \{W_1, \ldots, W_t\}, k)$ be an instance.
For two distinct terminal cliques $W_i$ and $W_j$ with $|W_i \cup W_j| \le k+1$, we denote by $\{W_1, \ldots, W_t\} \mergcliqs (W_i, W_j)$ the set obtained by removing $W_i$ and $W_j$ from $\{W_1, \ldots, W_t\}$ and inserting $W_i \cup W_j$ (if $W_i \cup W_j$ is already present, nothing is inserted).
We make $W_i \cup W_j$ into a clique in this operation so we denote $G \mergcliqs (W_i, W_j) = G \mcliq (W_i \cup W_j)$ and by $\ins \mergcliqs (W_i, W_j)$ we denote the instance $(G \mergcliqs (W_i, W_j), \{W_1, \ldots, W_t\} \mergcliqs (W_i, W_j), k)$.

Next we observe that the terminal clique merging does not decrease the flow potentials of the other terminal cliques.

\begin{lemma}
\label{lem:mergcliqfld}
Let $\ins = (G, \{W_1, \ldots, W_t\}, k)$ be an instance and $W_i, W_j, W_l \in \{W_1, \ldots, W_t\}$ three distinct terminal cliques.
It holds that $\flp_{\ins \mergcliqs(W_i,W_j)}(W_l) \ge \flp_{\ins}(W_l)$.
\end{lemma}
\begin{proof}
This follows directly from the facts that $\ow_{\ins}(W_l) = \ow_{\ins \mergcliqs(W_i,W_j)}(W_l)$ and any separation of $G \mergcliqs (W_i,W_j)$ is also a separation of $G$.
\end{proof}

We also observe that any solution of $\ins \mergcliqs(W_i,W_j)$ is also a solution of $\ins$.

We say that $\ins$ is \emph{maximally merged} if for any pair of two distinct terminal cliques $W_i$ and $W_j$ it holds that either $|W_i \cup W_j| > k+1$ or $\ins \mergcliqs (W_i,W_j)$ is a no-instance.
In particular, we can conclude that $\ins$ is maximally merged after branching on all different ways to merge two terminal cliques and not finding a solution.

\subsubsection{Leaf pushing}
\label{subsec:leafpush}
A terminal clique $W_i$ is a \emph{potential forget-clique} of $\ins$ if there exists a solution $(X, (T, \bag))$ of $\ins$ so that $T$ contains a leaf node $l$ with a parent $p$ so that $W_i \subseteq \bag(l)$ and $\bag(p) = \bag(l) \setminus \{w\}$ for some $w \in W_i \setminus \ow_{\ins}(W_i)$.
The leaf pushing operation will make progress by adding a vertex to a potential forget-clique.

Next we show that a maximally merged yes-instance has at least two potential forget-cliques.
The fact that there are at least two of them will be important since we do not necessarily make progress by leaf pushing the uniquely largest terminal clique.

\begin{lemma}
\label{lem:twopotforcli}
Let $\ins = (G, \{W_1, \ldots, W_t\}, k)$ be a yes-instance that is maximally merged and has $t \ge 2$.
The instance $\ins$ has at least two potential forget-cliques.
\end{lemma}
\begin{proof}
Let $(X, (T, \bag))$ be a solution of $\ins$ that first minimizes $|X|$ and subject to that minimizes $|V(T)|$.
First, if $|V(T)| = 1$, then all terminal cliques would be contained in the only bag of $(T, \bag)$, and $\ins$ would not be maximally merged.
Therefore $|V(T)| \ge 2$ and $T$ has at least two leaves.

\begin{claim}
\label{claim:wforget}
For any leaf node $l$ of $T$ with parent $p$, it holds that $\bag(l) \setminus \bag(p) \subseteq \allw_{\ins}$.
\end{claim}
\begin{proof}
Suppose that $\bag(l) \setminus \bag(p)$ contains a vertex $x \in V(G) \setminus \allw_{\ins}$.
Let $(T', \bag')$ be tree decomposition obtained from $(T, \bag)$ by removing $x$ from $\bag(l)$.
We claim that then $(X \setminus \{x\}, (T', \bag'))$ is a solution of $\ins$ that would contradict the minimality of $|X|$.
It holds that $X \setminus \{x\}$ covers $\allw_{\ins}$ and that the width of $(T', \bag')$ is at most $k$, so it remains to argue that $(T', \bag')$ is a tree decomposition of $\torso_G(X \setminus \{x\})$.
Because $x$ occurred only in the bag $\bag(l)$, $(T', \bag')$ satisfies the vertex condition and the connectedness condition.
For the edge condition, it suffices to prove that any path from $x$ to $X \setminus \{x\}$ intersects $\bag(l) \setminus \{x\}$.
This follows from \Cref{lem:indsub} by considering a modified version of $(T, \bag)$ where a bag containing $\bag(l) \setminus \{x\}$ is inserted between $l$ and $p$.
\end{proof}

Now, let $l,p \in V(T)$ be some leaf-parent pair.
We have that $\bag(l) \setminus \bag(p)$ is non-empty because otherwise we could decrease $|V(T)|$ by removing $l$.
Therefore, by \Cref{claim:wforget} there exists a terminal clique $W_i$ that intersects $\bag(l) \setminus \bag(p)$.
Because $W_i$ is a clique and the decomposition covers $W_i$, we know that $W_i \subseteq \bag(l)$.
We can modify $(T, \bag)$ by adding nodes between $l$ and $p$ so that the vertices in $\bag(l) \setminus \bag(p)$ are forgotten one-by-one, and that a vertex $w \in W_i \cap (\bag(l) \setminus \bag(p))$ is the first to be forgotten.
In particular, this results in a decomposition where the parent of $l$ is a node $p'$ with $\bag(p') = \bag(l) \setminus \{w\}$.
Now, if $w$ would be in some other terminal clique $W_j \neq W_i$, then $W_j \subseteq \bag(l)$ would hold because $\bag(l)$ is the only bag containing $w$, but then $\ins$ would not be maximally merged.
Therefore, $W_i$ is a potential forget-clique.

Finally, to show that there are at least two potential forget-cliques, note that if $W_i$ intersects $\bag(l) \setminus \bag(p)$ for two different leaf-parent pairs $l_1, p_1$ and $l_2, p_2$, then because $W_i \subseteq \bag(l_1)$ and $W_i \subseteq \bag(l_2)$, by the connectivity condition it would hold that $W_i \subseteq \bag(p_1)$, contradicting that $W_i$ intersects $\bag(l_1) \setminus \bag(p_1)$.
Therefore, for every leaf-parent pair $l,p$ we can assign a unique terminal clique $W_i$, and therefore as there are at least two leafs there are at least two potential forget-cliques.
\end{proof}

We introduce notation for defining the leaf pushing operation.
Let $\ins = (G, \{W_1, \ldots, W_t\}, k)$ be an instance, $W_i$ a terminal clique, and $A \subseteq V(G)$ a set of vertices that is disjoint from $W_i$ and $|W_i \cup A| \le k+1$.
We denote by $\{W_1, \ldots, W_t\} \addv (W_i,A)$ the set obtained from $\{W_1, \ldots, W_t\}$ by replacing the terminal clique $W_i$ by $W_i \cup A$.
Again, if $W_i \cup A$ already exists, we just remove $W_i$.
We denote $G \addv (W_i, A) = G \mcliq (W_i \cup A)$ and by $\ins \addv (W_i,A)$ we denote the instance $(G \addv (W_i, A), \{W_1, \ldots, W_t\} \addv (W_i,A), k)$.
Observe that any solution of $\ins \addv (W_i,A)$ is also a solution of $\ins$.
Observe also that if $\ins \addv (W_i, A)$ is a yes-instance and $A' \subseteq A$, then $\ins \addv (W_i, A')$ is also a yes-instance.

Next we prove the main leaf pushing lemma, in particular that we can increase the size of a potential forget-clique by guessing an important separator.

\begin{lemma}
\label{lem:leafpush_impsep}
Let $\ins = (G, \{W_1, \ldots, W_t\}, k)$ be a maximally merged yes-instance with $t \ge 2$ and no safe separators and $W_i$ a potential forget-clique of $\ins$.
There exists a vertex $w \in W_i \setminus \ow_{\ins}(W_i)$ and in the graph $G \setminus (W_i \setminus \{w\})$ a non-empty important $(\{w\}, \allw_{\ins} \setminus W_i)$-separator $S$ disjoint from $W_i$ so that $\ins \addv (W_i, S)$ is a yes-instance.
\end{lemma}
\begin{proof}
By the definition of potential forget-clique, let $(X, (T, \bag))$ be a solution so that $(T, \bag)$ contains a leaf node $l$ with a parent $p$ so that $W_i \subseteq \bag(l)$ and $\bag(p) = \bag(l) \setminus \{w\}$ for $w \in W_i \setminus \ow_{\ins}(W_i)$.
Denote $W_i^f = W_i \setminus \{w\}$.

By \Cref{lem:indsub} it holds that $\bag(p) = \bag(l) \setminus \{w\}$ separates $w$ from $X \setminus \{w\}$, and therefore separates $w$ from $\allw_{\ins} \setminus \{w\}$.
Therefore, in the graph $G \setminus W_i^f$, the set $\bag(l) \setminus W_i$ is a $(\{w\}, \allw_{\ins} \setminus W_i)$-separator.
The set $\allw_{\ins} \setminus W_i$ is non-empty because $\ins$ is maximally merged and $t \ge 2$.

Let $S$ be a smallest important $(\{w\}, \allw_{\ins} \setminus W_i)$-separator in $G \setminus W_i^f$ that dominates $\bag(l) \setminus W_i$.
The separator $S$ does not contain $w$ because $w \in \reach_{G \setminus W_i^f}(\{w\}, \bag(l) \setminus W_i)$, and therefore $S$ is disjoint from $W_i$.
The separator $S$ is non-empty because otherwise $W_i^f$ would separate $w$ from $\allw_{\ins} \setminus W_i$ and be a safe separator.
Let $(A, S, B)$ be the separation in $G \setminus W_i^f$ with $B = \reach_{G \setminus W_i^f}(\{w\}, S)$ and $A = V(G) \setminus (W_i^f \cup B \cup S)$, which implies $\allw_{\ins} \setminus W_i \subseteq A \cup S$.
By \Cref{lem:imp_sep_dom}, $S$ is linked into $(\bag(l) \setminus W_i) \cap (B \cup S)$ in $G \setminus W_i^f$.

By adding $W_i^f$ back to the graph and to the separator, we get that $(A, S \cup W_i^f, B)$ is a separation of $G$, and moreover $S \cup W_i^f$ is linked into $\bag(l) \cap (B \cup S \cup W_i^f)$ (the vertices in $W_i^f$ are linked by trivial one-vertex paths).
Let $X' = (X \cap A) \cup S \cup W_i^f$ and $(X', (T', \bag'))$ be the torso tree decomposition obtained by applying the pulling lemma (\Cref{lem:pull}) with $(X, (T, \bag))$, the separation $(A, S \cup W_i^f, B)$, and the node $l$ of $T$.

Now, $(X', (T', \bag'))$ has width at most $k$ and it covers $\allw_{\ins} \setminus \{w\}$.
Also, $S \cup W_i^f \subseteq \bag'(l)$.
We construct a torso tree decomposition $(X' \cup \{w\}, (T'', \bag''))$ from $(X', (T', \bag'))$ by adding a leaf $l'$ adjacent to $l$ with $\bag''(l') = S \cup W_i$.
Because $|S| \le |\bag(l) \setminus W_i|$, it follows that $|\bag''(l')| \le |\bag(l)| \le k+1$.
Also, $(X' \cup \{w\}, (T'', \bag''))$ covers $S \cup \allw_{\ins}$, and therefore it remains to prove that $(T'', \bag'')$ is indeed a tree decomposition of $\torso(X' \cup \{w\})$.
It satisfies the vertex condition because $(T', \bag')$ satisfied the vertex condition for $X'$ and the vertex $w$ is in the bag $\bag''(l')$.
It satisfies the connectivity condition because $(T', \bag')$ satisfied the connectivity condition, the vertex $w$ is in no bag of $(T', \bag')$, and $S \cup W_i^f \subseteq \bag'(l)$.
It remains to prove that $(T'', \bag'')$ satisfies the edge condition, which follows from the edge condition of $(T', \bag')$ and the fact that $S \cup W_i^f$ separates $w$ from $X'$.
\end{proof}

In the algorithm, we will apply \Cref{lem:leafpush_impsep} together with the important separator hitting set lemma (\Cref{lem:imp_sep_hit}).
In particular, we will add only a single vertex of $S$ to $W_i$ in the actual leaf pushing branching.
Next we show that if $W_i$ is not a unique largest terminal clique, adding any vertex to $W_i$ increases its flow potential.

\begin{lemma}
\label{lem:addvfld}
Let $\ins = (G, \{W_1, \ldots, W_t\}, k)$ be an instance with $t \ge 2$ that has no safe separations.
Let also $i \in [t]$ so that $|W_i| \le k$ and exists $j \neq i$ so that $|W_j| \ge |W_i|$ and $W_j$ is not a superset of $W_i$, and let $v \in V(G) \setminus W_i$.
Then $\flp_{\ins \addv (W_i,\{v\})}(W_i \cup \{v\}) \ge \flp_{\ins}(W_i)+1$.
\end{lemma}
\begin{proof}
It suffices to show that $W_i \cup \{v\}$ has the maximum possible flow potential, in particular that $\flp_{\ins \addv (W_i,\{v\})}(W_i \cup \{v\}) = |W_i \cup v|$.
Suppose otherwise, in particular suppose that there exists a separation $(A,S,B)$ with $W_i \cup \{v\} \subseteq A \cup S$, $\ow_{\ins \addv (W_i,\{v\})}(W_i \cup \{v\}) \subseteq B \cup S$, $B$ non-empty, and $|S| < |W_i \cup \{v\}|$.
Note that $|S| < |W_i \cup \{v\}|$ implies that also $A$ is non-empty, and note that because $W_j \neq W_i \cup \{v\}$, we have that $W_j \subseteq \ow_{\ins \addv (W_i,\{v\})}(W_i \cup \{v\}) \subseteq B \cup S$.
In particular, $S$ is a $(W_i,W_j)$-separator.
Because $\ins$ has no safe separations, by \Cref{lem:strictlinked} $W_i$ is strictly linked into $W_j$.
Therefore, because $|S| \le |W_i| \le |W_j|$, either $S = W_i$ or $S = W_j$.
However, in either case $(A,S,B)$ would be a safe separation of $\ins$, which is a contradiction.
\end{proof}

\subsection{The algorithm}
\label{subsec:algotogether}
In this subsection, we put the reduction rules and branching together to a complete algorithm for \pstw, and analyze the algorithm.

First, we need the following lemma to handle corner cases.

\begin{lemma}
\label{lem:smallcases}
Instances with $t=1$ or $|V(G)| \le k+2$ can be solved in $k^{\OO(1)} m$ time.
\end{lemma}
\begin{proof}
If $t=1$, then because $|W_1| \le k+1$, there is a trivial solution where $X=W_1$ and the tree decomposition has a single bag containing $X$.
When $|V(G)| \le k+2$, we consider the following cases.
First, if $|\allw_{\ins}| \le k+1$, then again the trivial single-bag solution suffices.
Otherwise, we have that $X = \allw_{\ins} = V(G)$, and there exists a solution if and only if $G$ is not a complete graph.
\end{proof}

\begin{algorithm}[!b]
\caption{A $k^{\OO(kt)} nm$ time algorithm for \pstw.\label{alg:partsubtw}}
\textbf{Input:} Instance $\ins = (G, \{W_1, \ldots, W_t\}, k)$.\\
\textbf{Output:} Either a solution of $\ins$ or \no.
\begin{algorithmic}[1]
\If {$t = 1$ or $|V(G)| \le k+2$} \Return Case-analysis($\ins$)\label{alg:partsubtw:bf} \Comment{\Cref{lem:smallcases}}
\EndIf
\If {Exists a safe separation $(A,S,B)$}\label{alg:partsubtw:safesepcheck}
\State \Return Combine(Solve($\ins \rescliqs (A,S)$), Solve($\ins \rescliqs (B,S)$))\label{alg:partsubtw:safesepcheckret}
%\State $sol_A \gets$ Solve($\ins \rescliqs (A,S)$)\label{alg:partsubtw:safesepchecksola}
%\State $sol_B \gets$ Solve($\ins \rescliqs (B,S)$)\label{alg:partsubtw:safesepchecksolb}
%\If {$sol_A = \no$ or $sol_B = \no$}\label{alg:partsubtw:safesepifno}
%\State \Return \no\label{alg:partsubtw:retnocomb}
%\Else\label{alg:partsubtw:safesepelse}
%\State Let $sol_A = (X_A, (T_A, \bag_A))$ and $sol_B = (X_B, (T_B, \bag_B))$\label{alg:partsubtw:safesepcombi}
%\State \Return $(X_A \cup X_B, (T_A, \bag_A) \cup_S (T_B, \bag_B))$\label{alg:partsubtw:retcomb}
%\EndIf
\EndIf
\ForAll{$i, j \in [t]$ with $i \neq j$ and $|W_i \cup W_j| \le k+1$}\label{alg:partsubtw:forcliqupairs}
\State $sol \gets$ Solve($\ins \mergcliqs (W_i,W_j)$)\label{alg:partsubtw:mergrec}
\If {$sol \neq \no$} \Return $sol$\label{alg:partsubtw:retmc}
\EndIf
\EndFor
\If {Exists $i,j \in [t]$ with $W_i \subset W_j$} \Return \no\label{alg:partsubtw:retsubset}
\EndIf
\ForAll{$i \in [t]$ so that $|W_i| \le k$ and exists $j \neq i$ with $|W_j| \ge |W_i|$}\label{alg:partsubtw:itforgcliqs}
\ForAll{$w \in W_i$}\label{alg:partsubtw:itforgcliqw}
\State $H \gets$ $\text{ImpSepHittingSet}_{G \setminus (W_i \setminus \{w\})}$($\{w\}$, $\allw_{\ins} \setminus W_i$, $k$)\label{alg:partsubtw:impsephitset} \Comment{\Cref{lem:leafpush_impsep,lem:imp_sep_hit}}
\ForAll{$v \in H \setminus \{w\}$}\label{alg:partsubtw:ithitset}
\State $sol \gets$ Solve($\ins \addv (W_i,\{v\})$)\label{alg:partsubtw:addrec}
\If {$sol \neq \no$} \Return $sol$\label{alg:partsubtw:retadd}
\EndIf
\EndFor
\EndFor
\EndFor
\State \Return \no\label{alg:partsubtw:finalno}
\end{algorithmic}
\end{algorithm}

The algorithm for \pstw is presented in the pseudocode \Cref{alg:partsubtw}.
We denote recursive applications of the algorithm by the function ``Solve''.

First, on \Cref{alg:partsubtw:bf}, \Cref{alg:partsubtw} handles the special cases of $t=1$ and $|V(G)| \le k+2$ by \Cref{lem:smallcases}.
Then, on \Cref{alg:partsubtw:safesepcheck,alg:partsubtw:safesepcheckret} the reduction by safe separations discussed in \Cref{subsec:safesep} is implemented.
In particular, if there exists a safe separation $(A,S,B)$, then the instances $\ins \rescliqs (A,S)$ and $\ins \rescliqs (B,S)$ are solved recursively, and if both of them return a solution the solutions are combined to a solution of $\ins$, and if either of them returns \no, then we return \no.
The function ``Combine'' on \Cref{alg:partsubtw:safesepcheckret} denotes an operation that returns \no if either of its arguments is \no, and if its arguments are a solution $(X_A, (T_A, \bag_A))$ of $\ins \rescliqs (A,S)$ and a solution $(X_B, (T_B, \bag_B))$ of $\ins \rescliqs (B,S)$ then it returns the solution $(X_A \cup X_B, (T_A, \bag_A) \cup_S (T_B, \bag_B))$ of $\ins$.

Then, the terminal clique merging branching discussed in \Cref{subsec:tercliqmerg} is implemented on \Cref{alg:partsubtw:forcliqupairs,alg:partsubtw:mergrec,alg:partsubtw:retmc}.
In particular, the algorithm branches on merging all pairs of terminal cliques $W_i,W_j$ with $|W_i \cup W_j| \le k+1$ and returns a solution if any of the resulting instances were yes-instances.
After this, $\ins$ is maximally merged, and this is immediately used on \Cref{alg:partsubtw:retsubset} to return \no if some terminal clique is a subset of some other terminal clique.

Then, the leaf pushing branching discussed in \Cref{subsec:leafpush} is implemented on \Cref{alg:partsubtw:itforgcliqs,alg:partsubtw:itforgcliqw,alg:partsubtw:impsephitset,alg:partsubtw:ithitset,alg:partsubtw:addrec,alg:partsubtw:retadd}.
The algorithm branches on all candidates for a potential forget-clique $W_i$ that is not a uniquely largest terminal clique, and a vertex $w \in W_i$ for which there exists an important $(\{w\}, \allw_{\ins} \setminus W_i)$-separator $S$ in the graph $G \setminus (W_i \setminus \{w\})$ so that $w$ and $S$ satisfy the conditions of \Cref{lem:leafpush_impsep}.
The algorithm does not branch on all such important separators $S$, but instead uses the important separator hitting set lemma (\Cref{lem:imp_sep_hit}) to obtain a vertex set $H$ of size at most $k$ that intersects all important $(\{w\}, \allw_{\ins} \setminus W_i)$-separator $S$ of size at most $k$ in the graph $G \setminus (W_i \setminus \{w\})$.
Then, a single vertex of such an important separator can be guessed by guessing a single vertex in $H$, so the algorithm branches on all vertices in $H \setminus \{w\}$ to add to $W_i$.
Finally, on \Cref{alg:partsubtw:finalno} the algorithm concludes that $\ins$ is a no-instance if none of the branches returned a solution.

%On \Cref{alg:partsubtw:itforgcliqs,alg:partsubtw:itforgcliqw,alg:partsubtw:impsephitset,alg:partsubtw:ithitset,alg:partsubtw:addrec,alg:partsubtw:retadd} the leaf pushing branching of \Cref{subsec:leafpush} is implemented, in particular on \Cref{alg:partsubtw:impsephitset} we use \Cref{lem:imp_sep_hit} to find a set $H$ that intersects all important separators that fulfill the conditions of \Cref{lem:leafpush_impsep} with $W_i$ and $w \in W_i$ fixed by the loops of \Cref{alg:partsubtw:itforgcliqs,alg:partsubtw:itforgcliqw}.

The algorithm clearly works in polynomial space.
We then prove the correctness of \Cref{alg:partsubtw}.
Its running time will be analyzed in \Cref{subsec:algo1tc}.

We start by proving that \Cref{alg:partsubtw} is correct when it returns a solution.

\begin{lemma}
\label{lem:algpartsubtwc1}
If \Cref{alg:partsubtw} returns a solution, then it is a solution of $\ins$.
\end{lemma}
\begin{proof}
For the case analysis of \Cref{alg:partsubtw:bf} this is by \Cref{lem:smallcases}.
Then, we use induction on the recursion tree, assuming that the lemma holds for recursive calls of the algorithm.

Now, whenever \Cref{alg:partsubtw} returns on \Cref{alg:partsubtw:safesepcheckret} after finding a safe separation $(A,S,B)$ and combining solutions of $\ins \rescliqs (A,S)$ and $\ins \rescliqs (B,S)$ into a solution of $\ins$, it is correct by induction and \Cref{lem:sepredgen}.
The cases when the algorithm returns a solution after terminal clique merging on \Cref{alg:partsubtw:retmc} or after leaf pushing on \Cref{alg:partsubtw:retadd} are correct by induction and the facts that any solution of $\ins \mergcliqs (W_i,W_j)$ is also a solution of $\ins$ and any solution of $\ins \addv (W_i,\{v\})$ is also a solution of $\ins$.
\end{proof}

We then show that \Cref{alg:partsubtw} is correct when it returns \no.

\begin{lemma}
\label{lem:algpartsubtwc2}
If \Cref{alg:partsubtw} returns \no, then $\ins$ is a no-instance.
\end{lemma}
\begin{proof}
For the case analysis of \Cref{alg:partsubtw:bf} this is by \Cref{lem:smallcases}.
Then we use induction on the recursion tree, assuming that the lemma holds for recursive calls of the algorithm.
The correctness of safe separation reduction on \Cref{alg:partsubtw:safesepcheckret} follows from induction and \Cref{lem:safesepyesyes}.

Then, after the safe separation reduction of \Cref{alg:partsubtw:safesepcheck,alg:partsubtw:safesepcheckret} we can assume that $\ins$ has no safe separations, and by the terminal clique merging of \Cref{alg:partsubtw:forcliqupairs,alg:partsubtw:mergrec,alg:partsubtw:retmc} and induction we can assume that $\ins$ is maximally merged.
The correctness of returning \no on \Cref{alg:partsubtw:retsubset} if there are terminal cliques $W_i \subset W_j$ follows from the facts that $\ins$ is maximally merged and if $\ins$ would be a yes-instance, then $\ins \mergcliqs (W_i, W_j)$ would also be a yes-instance.

It remains to argue that if \Cref{alg:partsubtw} returns from the final \Cref{alg:partsubtw:finalno}, then $\ins$ is a no-instance.
For the sake of contradiction, assume that \Cref{alg:partsubtw} returns \no from \Cref{alg:partsubtw:finalno} but $\ins$ is a yes-instance.
Now, by \Cref{lem:twopotforcli}, $\ins$ has at least two potential forget-cliques.
Let $W_i$ be a smallest potential forget-clique of $\ins$.
By \Cref{lem:leafpush_impsep} we have that $|W_i| \le k$, and therefore $W_i$ satisfies the conditions of \Cref{alg:partsubtw:itforgcliqs}.
By \Cref{lem:leafpush_impsep}, there exists a vertex $w \in W_i$ and in $G \setminus (W_i \setminus \{w\})$ a non-empty important $(\{w\}, \allw_{\ins} \setminus W_i)$-separator $S$ disjoint from $W_i$ so that $\ins \addv (W_i, S)$ is a yes-instance.
Some iteration of \Cref{alg:partsubtw:itforgcliqw} will choose this vertex $w \in W_i$, and it holds that $H \cap S \neq \emptyset$, so some iteration of \Cref{alg:partsubtw:ithitset} will choose a vertex $v \in S$.
Because $\ins \addv (W_i, S)$ is a yes-instance, $\ins \addv (W_i, \{v\})$ is also a yes-instance, so by induction we get that \Cref{alg:partsubtw} would return on \Cref{alg:partsubtw:retadd}, which is a contradiction.
\end{proof}

\subsection{Running time analysis}
\label{subsec:algo1tc}
We then prove that the running time of \Cref{alg:partsubtw} is $k^{\OO(kt)} nm$.
For this, we introduce the measures $\fld_{\ins}(W_i)$ of a terminal clique and $\fld(\ins)$ of the instance.

We define the measure of a terminal clique based on its flow potential as

\[\fld_{\ins}(W_i) = 3k+3-\flp_{\ins}(W_i).\]
Observe that $2k+2 \le \fld_{\ins}(W_i) \le 3k+3$, which in particular implies that the sum of measures of two terminal cliques is always at least $k+1$ larger than the measure of a single terminal clique.

Then, the measure of the instance is defined as

\[\fld(\ins) = \sum_{i=1}^t \fld_{\ins}(W_i).\]

Note that $(2k+2)t \le \fld(\ins) \le (3k+3)t$.
We will show the running time of the algorithm to be of form $k^{\OO(\fld(\ins))} nm = k^{\OO(kt)} nm$.

To this end, we will show that breaking the instance by a safe separation does not increase the measure, and that both the terminal clique merging branching of \Cref{alg:partsubtw:mergrec} and the leaf pushing branching of \Cref{alg:partsubtw:addrec} decrease the measure by at least one.
We start by proving the property for safe separations, using \Cref{lem:sepsame,lem:sepflp,lem:ssflp2}.

\begin{lemma}
\label{lem:safeseptcfld}
Let $\ins = (G, \{W_1, \ldots, W_t\}, k)$ be an instance and $(A,S,B)$ a safe separation.
It holds that $\fld(\ins \rescliqs (A,S)) \le \fld(\ins)$.
\end{lemma}
\begin{proof}
We consider the three cases of \Cref{lem:ssflp2}.
First, if $\ins \rescliqs (A,S)$ has less terminal cliques than $\ins$, let $t' < t$ be the number of terminal cliques of $\ins \rescliqs (A,S)$.
Recall from the definition of $\ins \rescliqs (A,S)$ that all terminal cliques of $\ins \rescliqs (A,S)$ except possibly $S$ are also terminal cliques of $\ins$, and moreover $\ins \rescliqs (A,S)$ has at least one terminal clique that is a superset of $S$.
Therefore, the conditions of \Cref{lem:sepflp} apply for at least $t'-1$ terminal cliques $W_i$ of $\ins \rescliqs (A,S)$, in particular, for them $\flp_{\ins \rescliqs (A,S)}(W_i) \ge \flp_{\ins}(W_i)$ holds by \Cref{lem:sepflp} and therefore for them $\fld_{\ins \rescliqs (A,S)}(W_i) \le \fld_{\ins}(W_i)$.
Because the measure of a terminal clique is at least $2k+2$ and at most $3k+3$, it follows that
\[\fld(\ins \rescliqs (A,S)) \le \fld(\ins) + 3k+3 - (2k+2)(t-(t'-1)) \le \fld(\ins) + 3k+3 - (2k+2) \cdot 2 \le \fld(\ins) - k-1\]
in this case.

Then, if $\{W_1, \ldots, W_t\} \rescliqs (A,S) = \{W_1, \ldots, W_t\}$, the lemma follows directly from \Cref{lem:sepsame}.

Then, if both $\ins$ and $\ins \rescliqs (A,S)$ have $t$ terminal cliques and there is a terminal clique $W_i$ of $\ins$ with $W_i \cap B \neq \emptyset$, then $\ins$ does not contain any terminal clique that is a subset of $A \cup S$ and a superset of $S$ and therefore for all terminal cliques $W_j$ of $\ins \rescliqs (A,S)$ except $S$ we have by \Cref{lem:sepflp} that $\flp_{\ins \rescliqs (A,S)}(W_j) \ge \flp_{\ins}(W_j)$.
Therefore, we have that
\[\fld(\ins \rescliqs (A,S)) \le \fld(\ins) + \fld_{\ins \rescliqs (A,S)}(S) - \fld_{\ins}(W_i),\]
which by $\flp_{\ins \rescliqs (A,S)}(S) \ge \flp_{\ins}(W_i)$ (from \Cref{lem:ssflp2}) implies that $\fld(\ins \rescliqs (A,S)) \le \fld(\ins)$.
\end{proof}

Next, we observe that the terminal clique merging branching of \Cref{alg:partsubtw:mergrec} decreases the measure of the instance.

\begin{lemma}
\label{lem:cmtcfld}
Let $\ins = (G, \{W_1, \ldots, W_t\}, k)$ be an instance and $W_i$, $W_j$ two distinct terminal cliques.
It holds that $\fld(\ins \mergcliqs (W_i,W_j)) \le \fld(\ins)-k-1$.
\end{lemma}
\begin{proof}
This follows from \Cref{lem:mergcliqfld} and the facts that $\fld_{\ins}(W_i)+\fld_{\ins}(W_j) \ge 4k+4$ and $\fld_{\ins \mergcliqs (W_i,W_j)}(W_i \cup W_j) \le 3k+3$.
\end{proof}

Then, we observe that \Cref{lem:addvfld} implies that the leaf pushing branching of \Cref{alg:partsubtw:addrec} decreases the measure of the instance.

\begin{lemma}
\label{lem:addvtcfld}
Let $\ins = (G, \{W_1, \ldots, W_t\}, k)$ be an instance with $t \ge 2$ that has no safe separations.
Let also $i \in [t]$ so that $|W_i| \le k$ and exists $j \neq i$ so that $|W_j| \ge |W_i|$ and $W_j$ is not a superset of $W_i$, and let $v \in V(G) \setminus W_i$.
It holds that $\fld(\ins \addv (W_i, \{v\})) \le \fld(\ins)-1$.
\end{lemma}
\begin{proof}
Observe that adding a vertex to a terminal clique does not decrease the flow potentials of other terminal cliques.
Therefore, the lemma holds by \Cref{lem:addvfld}.
\end{proof}

Then for the reduction by safe separations we have to argue that the sum of the sizes of the instances $\ins \rescliqs (A,S)$ and $\ins \rescliqs (B,S)$ is less than the size of $\ins$.
For this, we define the graph size of $\ins$ to be
\[\vertsz(\ins) = \max(1, (k+2)|V(G)|-(k+2)^2).\]
Note that $\vertsz(\ins) = 1$ if and only if $|V(G)| \le k+2$.

We show that if $(A,S,B)$ is a strict separation and $|S| \le k+1$, then with respect to the $\vertsz(\ins)$ measure, the instances $\ins \rescliqs (A,S)$ and $\ins \rescliqs (B,S)$ are in total smaller than $\ins$ if $|V(G)| \ge k+3$.

\begin{lemma}
\label{lem:gs_sep}
Let $\ins = (G, \{W_1, \ldots, W_t\}, k)$ be an instance with $|V(G)| \ge k+3$ and $(A,S,B)$ a strict separation with $|S| \le k+1$.
It holds that $\vertsz(\ins \rescliqs (A,S)) + \vertsz(\ins \rescliqs (B,S)) \le \vertsz(\ins)-1$.
\end{lemma}
\begin{proof}
First, because $(A,S,B)$ is a strict separation, $|V(G)| \ge k+3$, and $k \ge 1$, both of $\vertsz(\ins \rescliqs (A,S)) + 2 \le \vertsz(\ins)$ and $\vertsz(\ins \rescliqs (B,S)) + 2 \le \vertsz(\ins)$ hold.
This implies that the lemma holds whenever $\vertsz(\ins \rescliqs (A,S)) = 1$ or $\vertsz(\ins \rescliqs (B,S)) = 1$.
It remains to consider the case where $\vertsz(\ins \rescliqs (A,S)) > 1$ and $\vertsz(\ins \rescliqs (B,S)) > 1$, in particular where $|A \cup S| \ge k+3$ and $|B \cup S| \ge k+3$.

In this case
\begin{align}
\vertsz(\ins \rescliqs (A,S)) + \vertsz(\ins \rescliqs (B,S)) =& (k+2)|A \cup S| - (k+2)^2 + (k+2)|B \cup S| - (k+2)^2 \nonumber\\
=& (k+2)(|A \cup S|+|B \cup S|) - 2(k+2)^2 \nonumber\\
=& (k+2)(|V(G)|+|S|) - 2(k+2)^2 \nonumber\\
=& (k+2)(|V(G)|+|S|-k-2)-(k+2)^2 \le \vertsz(\ins)-1\nonumber.
\end{align}
\end{proof}

We then put the running time analysis together.

\begin{lemma}
\label{lem:algpartsubtwtime}
\Cref{alg:partsubtw} runs in time $k^{\OO(kt)} nm$.
\end{lemma}
\begin{proof}
First, we observe that all of the operations in a single call of the recursive procedure can be performed in $k^{\OO(1)} m'$ time, where $m'$ is the number of edges in the instance given to the recursive call.
In particular, the case analysis of \Cref{alg:partsubtw:bf} can be implemented in $k^{\OO(1)} m'$ time by \Cref{lem:smallcases}, reducing by safe separations on \Cref{alg:partsubtw:safesepcheck,alg:partsubtw:safesepcheckret} can be implemented in $k^{\OO(1)} m'$ time by \Cref{lem:safeseptime}, for terminal clique merging on \Cref{alg:partsubtw:forcliqupairs,alg:partsubtw:mergrec,alg:partsubtw:retmc} this is trivial, and for leaf pushing on \Cref{alg:partsubtw:itforgcliqs,alg:partsubtw:itforgcliqw,alg:partsubtw:impsephitset,alg:partsubtw:ithitset,alg:partsubtw:addrec,alg:partsubtw:retadd} it is an application of the important separator hitting set lemma (\Cref{lem:imp_sep_hit}).

By the definition of $\ins \rescliqs (A,S)$, observe that at each recursive call the current graph can be obtained from an induced subgraph of the original graph by adding all edges inside the terminal cliques, and therefore we can bound $m' \le k^{\OO(1)} m$, where $m$ is the number of original edges.
Therefore, the running time of the algorithm can be bounded by $k^{\OO(1)} m \cdot R(\ins)$, where $R(\ins)$ is the total number of recursive calls.

We show by induction that the number of recursive calls is bounded by \[R(\ins) \le \vertsz(\ins) \cdot ((k+2)^3)^{\fld(\ins)} = k^{\OO(1)} n \cdot k^{\OO(kt)} = k^{\OO(kt)} n.\]
First, when the algorithm returns from the case analysis of \Cref{alg:partsubtw:bf}, this holds because $\vertsz(\ins) \ge 1$ and $\fld(\ins) \ge 1$ always.
Then we can assume that $t \ge 2$ and $|V(G)| \ge k+3$.
If there exists a safe separation $(A,S,B)$, then the number of recursive calls is by induction
\begin{align}
R(\ins) =& 1+R(\ins \rescliqs (A,S)) + R(\ins \rescliqs (B,S)) \nonumber\\
\le& 1 + (\vertsz(\ins \rescliqs (A,S)) + \vertsz(\ins \rescliqs (B,S))) \cdot ((k+2)^3)^{\fld(\ins)} && \text{(by \Cref{lem:safeseptcfld})} \nonumber\\
\le& \vertsz(\ins) \cdot ((k+2)^3)^{\fld(\ins)}. && \text{(by \Cref{lem:gs_sep})} \nonumber
\end{align}

If no safe separators exist, then all recursive calls are from terminal clique merging on \Cref{alg:partsubtw:mergrec} and leaf pushing on \Cref{alg:partsubtw:addrec}.
By \Cref{lem:cmtcfld}, for recursive calls made from terminal clique merging on \Cref{alg:partsubtw:mergrec} it holds that $\fld(\ins \mergcliqs (W_i, W_j)) \le \fld(\ins) -1$, and by \Cref{lem:addvtcfld}, for recursive calls made from leaf pushing on \Cref{alg:partsubtw:addrec} it holds that $\fld(\ins \addv (W_i, \{v\})) \le \fld(\ins) -1$.
The total number of recursive calls from \Cref{alg:partsubtw:mergrec,alg:partsubtw:addrec} is at most $t^2 + tk^2 \le (k+2)^3-1$, so we get that
\begin{align}
R(\ins) \le& 1 + ((k+2)^3-1) \cdot \vertsz(\ins) \cdot ((k+2)^3)^{\fld(\ins)-1} \le \vertsz(\ins) \cdot ((k+2)^3)^{\fld(\ins)}.\nonumber
\end{align}
\end{proof}

This finishes the proof of \Cref{the:stwpalg}, and together with \Cref{the:stwapximpl} they imply \Cref{the:mainapx}.

\section{Faster algorithm for \stw}
\label{sec:fasttw}
This section is devoted to proving \Cref{the:stwexalg}, in particular to giving a $2^{\OO(k^2)} nm$ time algorithm for \stw.
We will re-use many definitions and lemmas of \Cref{sec:algopstw}.
In particular, we use the definition of an instance of \pstw from \Cref{sec:algopstw}, observing that an instance of \stw can be seen as an instance of \pstw having initially $t = |W| = k+2$ terminal cliques of size $1$.

The rest of this section is organized as follows.
In \Cref{subsec:degenseps} we introduce the concept of a terminal clique covering an original terminal vertex and based on that the concept of degenerate separation.
In \Cref{subsec:maintvalid} we introduce a new parameter to the measure the progress of the algorithm and argue how different operations on instances preserve so called ``valid instances''.
In \Cref{subsec:algok2branch} we discuss the branching rules of the algorithm and in \Cref{subsec:thealgok2} we describe the algorithm and put together its correctness proof.
In \Cref{subsec:alg2tcana} we analyze the running time.

\subsection{Terminal covers and degenerate separations}
\label{subsec:degenseps}
We extend the definition of an instance of \pstw given in \Cref{sec:algopstw}.
We now keep track also of the original graph $\og$ and the set of original terminal vertices $\ogw$, which were the original input to the \stw problem.
In particular, here $|\ogw| = k+2$.
Observe that the recursive algorithm of \Cref{sec:algopstw} maintains that if $\ins = (G, \{W_1, \ldots, W_t\}, k)$ is an instance of some recursive subproblem, then $V(G) \subseteq V(\og)$ and $E(G) \supseteq E(\torso_{\og}(V(G)))$.
In particular, the operations $\ins \mergcliqs (W_i,W_j)$ and $\ins \addv (W_i, A)$ trivially maintain this because they change $G$ only by adding edges, and the $\ins \rescliqs (A,S)$ operation with a separation $(A,S,B)$ maintains this because $S$ becomes a clique in $G \rescliqs (A,S)$.

%In particular, the following lemma holds

%\begin{lemma}
%Let $\ins = (G, \{W_1, \ldots, W_t\}, k)$ be an instance so that $V(G) \subseteq V(\og)$ and $E(G) \supseteq E(\torso_{\og}(V(G)))$.
%Then,
%\begin{enumerate}
%\item $V(G \mergcliqs (W_i,W_j)) \subseteq V(\og)$ and $E(G \mergcliqs (W_i,W_j)) \supseteq E(\torso_{\og}(V(G \mergcliqs (W_i,W_j))))$ for any $W_i,W_j \in \{W_1, \ldots, W_t\}$ with $|W_i \cup W_j| \le k+1$,
%\item $V(G \addv (W_i,A)) \subseteq V(\og)$ and $E(G \addv (W_i,A)) \supseteq E(\torso_{\og}(V(G \addv (W_i,A))))$ for any $W_i \in \{W_1, \ldots, W_t\}$ and $A \subseteq V(G) \setminus W_i$, and
%\item $V(G \mergcliqs (W_i,W_j)) \subseteq V(\og)$ and $E(G \mergcliqs (W_i,W_j)) \supseteq E(\torso_{\og}(V(G \mergcliqs (W_i,W_j))))$
%\end{enumerate}
%\end{lemma}
%with the operations $\ins \mergcliqs (W_i,W_j)$ and $\ins \addv (W_i, A)$ we only add edges to the 

\paragraph{Terminal covers.}
Let $\ins = (G, \{W_1, \ldots, W_t\}, k)$ be an instance, $\og$ the original graph, and $\ogw$ the original terminal vertices.
We say that a terminal clique $W_i$ \emph{covers} an original terminal vertex $w \in \ogw$ if $W_i$ is a $(\{w\}, V(G))$-separator in $\og$.
We observe that in the algorithm of the previous section, at every point for every original terminal vertex there exists a terminal clique that covers it, and moreover every terminal clique covers at least one original terminal vertex.

In the algorithm of this section we maintain a mapping $\tc : \ogw \rightarrow \{W_1, \ldots, W_t\}$ from the original terminal vertices to the current terminal cliques, so that for all $w \in \ogw$, the terminal clique $\tc(w)$ covers $w$.
We extend the definition of an instance to include the mapping $\tc$, in particular an instance is now a 4-tuple $\ins = (G, \{W_1, \ldots, W_t\}, k, \tc)$.

Let us now define how the mapping is maintained under the operation $\ins \rescliqs (A,S)$, where $(A,S,B)$ is a separation.
Let $S'$ be a terminal clique of $\ins \rescliqs (A,S)$ that is a superset of $S$.
If there are multiple such terminal cliques, then let $S'$ be the lexicographically first choice.
Then, we define 
\[
\tc(w) \rescliqs (A,S) = 
\begin{cases}
S' \text{ if } \tc(w) \subseteq B \cup S\\
\tc(w) \text{ otherwise}.
\end{cases}
\]

Now, we define $\ins \rescliqs (A,S) = (G \rescliqs (A,S), \{W_1, \ldots, W_t\} \rescliqs (A,S), k, \tc \rescliqs (A,S))$.
The following lemma shows that this correctly maintains the mapping $\tc$.

\begin{lemma}
Let $\ins = (G, \{W_1, \ldots, W_t\}, k, \tc)$ be an instance, $\og$ the original graph, and $w \in \ogw$ an original terminal vertex.
Let also $(A,S,B)$ be a separation of $G$.
If $\tc(w) \subseteq B \cup S$, then any terminal clique of $\ins \rescliqs (A,S)$ that is a superset of $S$ covers $w$ in $\ins \rescliqs (A,S)$.
Otherwise, $\tc(w)$ is a terminal clique of $\ins \rescliqs (A,S)$ and covers $w$ in $\ins \rescliqs (A,S)$.
\end{lemma}
\begin{proof}
First, in the case when $\tc(w)$ intersects $A$, in which case $\tc(w) \in \{W_1, \ldots, W_t\} \rescliqs (A,S)$, the fact that $\tc(w)$ covers $w$ in $\ins \rescliqs (A,S)$ holds directly by the fact that $\tc(w)$ covers $w$ in $\ins$.

Then, consider the case when $\tc(w) \subseteq B \cup S$.
Recall that by definition $\tc(w)$ is a $(\{w\}, V(G))$-separator in $\og$.
We will show that $S$ is a $(\{w\}, V(G \rescliqs (A,S)))$-separator in $\og$.
Consider any path from $w$ to $V(G \rescliqs (A,S)) = A \cup S$ in $\og$.
Because $\tc(w)$ covers $w$ in $\ins$, this path intersects $\tc(w)$, and therefore it has a suffix that is a $\tc(w)-A \cup S$ path in $\og$.
Now, because $E(G) \supseteq E(\torso_{\og}(V(G)))$, we can map the suffix into a $\tc(w)-A \cup S$ path in $G$ by just removing vertices in $\og \setminus V(G)$ from it.
Then, because $S$ is a $(\tc(w), A \cup S)$-separator in $G$, this path must intersect $S$, and therefore the path from $w$ to $A \cup S$ in $\og$ must also intersect $S$, and therefore $S$ is a $(\{w\}, A \cup S)$-separator in $\og$.
\end{proof}

We also define the maintenance of $\tc$ under terminal clique merging by
\[
\tc(w) \mergcliqs (W_i, W_j) =
\begin{cases}
W_i \cup W_j \text{ if } \tc(w) = W_i \text{ or } \tc(w) = W_j\\
\tc(w) \text{ otherwise}.
\end{cases}
\]

Now, $\ins \mergcliqs (W_i, W_j) = (G \mergcliqs (W_i, W_j), \{W_1, \ldots, W_t\} \mergcliqs (W_i, W_j), k, \tc \mergcliqs (W_i, W_j))$.
This maintains the mapping $\tc$ correctly because if $W_i$ is a $(\{w\}, V(G))$-separator in $\og$, then also $W_i \cup W_j$ is a $(\{w\}, V(G))$-separator in $\og$.

Similarly, when adding vertices to terminal cliques it is defined by 
\[
\tc(w) \addv (W_i, A) =
\begin{cases}
W_i \cup A \text{ if } \tc(w) = W_i\\
\tc(w) \text{ otherwise}.
\end{cases}
\]

Then, $\ins \addv (W_i, A) = (G \addv (W_i, A), \{W_1, \ldots, W_t\} \addv (W_i, A), k, \tc \addv (W_i, A))$.
Again, this clearly maintains the mapping $\tc$ correctly.

For a terminal clique $W_i$, we denote by $\ctc_{\ins}(W_i)$ the number of original terminal vertices mapped to $W_i$ by $\tc$, i.e., $\ctc_{\ins}(W_i) = |\{w \in \ogw \mid \tc(w) = W_i\}|$.
Note that the operations $\ins \mergcliqs (W_i, W_j)$ and $\ins \addv (W_i, A)$ preserve the invariant that $\ctc_{\ins}(W_i) \ge 1$ for all terminal cliques $W_i$, and the operation $\ins \rescliqs (A,S)$ for a separation $(A,S,B)$ preserves this if there is at least one terminal clique that is a subset of $B \cup S$ or a superset of $S$, which will be always the case when this operation is used.

\paragraph{Degenerate separations.}
Let $\ins = (G, \{W_1, \ldots, W_t\}, k, \tc)$ be an instance, and $(X, (T, \bag))$ a solution of $\ins$.
We say that a separation $(A,S,B)$ of $G$ is an \emph{internal separation} of the solution $(X, (T, \bag))$ for $\ins$ if $S \subseteq \bag(t)$ for some $t \in V(T)$ and $\allw_{\ins}$ intersects both $A$ and $B$.

We say that an internal separation $(A,S,B)$ is \emph{degenerate} if
\[|S| + \sum_{W_i \mid W_i \cap A \neq \emptyset} \ctc_{\ins}(W_i) \le k+1.\]

Note that if a solution contains a degenerate internal separation $(A,S,B)$, then for the purpose of obtaining a solution of the original instance we can, slightly informally speaking, replace the decomposition on the $A$-side of the separation by just a single bag $S \cup \{w \in \ogw \mid \tc(w) \cap A \neq \emptyset\}$ because the definition ensures that $|S \cup \{w \in \ogw \mid \tc(w) \cap A \neq \emptyset\}| \le k+1$.
In particular, we observe that for the original instance there always exists a solution where every degenerate internal separation is a separation between a leaf bag and the rest of the decomposition, with the leaf bag containing only the separator $S$ and the original terminal vertices ``behind'' $S$.
Now, our goal is to perform a \emph{pre-branching} step that by using important separators guesses, in some sense, a maximal set of degenerate internal separations, and after that arrives to an instance where no solution has a degenerate internal separation.

We say that an instance $\ins = (G, \{W_1, \ldots, W_t\}, k, \tc)$ is \emph{valid} if it is a yes-instance and has no solution $(X, (T, \bag))$ that has a degenerate internal separation for $\ins$.
Otherwise, we say that $\ins$ is \emph{invalid}.
Next we show that by performing the pre-branching step, we can assume that we start with a valid instance.

\begin{lemma}[Pre-branching]
\label{lem:prebranching}
There is an algorithm, that given a graph $G$, integer $k$, and a set $W$ with $|W|=k+2$, in time $2^{\OO(k^2)} m$ and polynomial space enumerates $2^{\OO(k^2)}$ instances $(G, \{W_1, \ldots, W_t\}, k, \tc)$ with $t \le k+2$ so that any solution of any of the instances can in time $k^{\OO(1)} m$ be turned into a torso tree decomposition of width at most $k$ in $G$ that covers $W$, and moreover if such a torso tree decomposition exists, then at least one of the returned instances is valid.
\end{lemma}
\begin{proof}
We initially consider the instance $\ins = (G, \{W_1, \ldots, W_{|W|}\}, k, \tc)$, where $\{W_1, \ldots, W_{|W|}\}$ is a partition of $W = \ogw$ into single vertices and $\tc(w) = \{w\}$.
We then branch on all possible ways to perform terminal clique merging operations.
There are $k^{\OO(k)}$ possible sequences of terminal clique merging.

Observe that any solution of any of the resulting instances is a torso tree decomposition of width at most $k$ in $G$ that covers $W$, and moreover if a solution exists, then at least one of the resulting instances is a maximally merged yes-instance.
Notice also that $\ctc_{\ins}(W_i) = |W_i|$ holds for all terminal cliques $W_i$ in instances obtained in this manner.
We will then prove the lemma with the assumption that we start with an instance for which $\ctc_{\ins}(W_i) = |W_i|$ holds, and in particular, if the starting instance is a maximally merged yes-instance, then at least one of the outputs will be a valid instance.

We will do branching that maintains a partition of terminal cliques into \emph{processed} terminal cliques and \emph{unprocessed} terminal cliques.
Initially, all of the terminal cliques are unprocessed.
This branching will always maintain that for unprocessed terminal cliques $W_i$ it holds that $\ctc_{\ins}(W_i) \ge |W_i|$, and in the ``success'' branches the following invariants will be maintained

\begin{enumerate}
\item \label{lem:prebranching:yes} $\ins$ is a yes-instance,
\item \label{lem:prebranching:invproc} for every processed terminal clique $W_i$ there exists no solution that contains a degenerate internal separation $(A,S,B)$ so that $A$ intersects $W_i$, and
\item \label{lem:prebranching:invmany} there exists no solution that contains a degenerate internal separation $(A,S,B)$ so that more than one terminal clique intersects $A$.
\end{enumerate}

Initially, the invariant $\ctc_{\ins}(W_i) \ge |W_i|$ holds as discussed earlier.
Also, if $\ins$ is initially a maximally merged yes-instance, the invariant of \Cref{lem:prebranching:yes} holds by definition, and the invariant of \Cref{lem:prebranching:invproc} initially holds by the fact that there are no processed terminal cliques.
The invariant of \Cref{lem:prebranching:invmany} initially holds if $\ins$ is maximally merged, because if there would be a solution $(X, (T, \bag))$ with a degenerate internal separation $(A,S,B)$, so that a set $\{W_{i_1}, \ldots, W_{i_l}\}$ of at least two terminal cliques intersects $A$, then by the definition $|S| + \sum_{W_i \mid W_i \cap A \neq \emptyset} \ctc_{\ins}(W_i) \le k+1$ and the fact that $\ctc_{\ins}(W_i) \ge |W_i|$ we could replace the solution on the subgraph induced by $A \cup S$ by just a single bag $S \cup \bigcup_{W_i \in \{W_{i_1}, \ldots, W_{i_l}\}} W_i$, and conclude that $\ins$ is not maximally merged.

The branching works as follows.
While there exists an unprocessed terminal clique $W_i$, we branch on the cases that either there exists no solution that contains a degenerate internal separation $(A,S,B)$ with $W_i \cap A \neq \emptyset$, recursing to the case where we just mark $W_i$ as processed, or that there exists a solution that contains a degenerate internal separation $(A,S,B)$ so that $W_i \cap A \neq \emptyset$, and in this case we recurse on all cases that are obtained by taking an important $(W_i, \ow_{\ins}(W_i))$-separator $S'$ of size $|S'| \le k+1-\ctc_{\ins}(W_i)$, letting $(A', S', B') = (\reach_G(W_i, S'), S', V(G) \setminus (\reach_G(W_i, S') \cup S'))$, and recursing to the instance $\ins \rescliqs (B', S')$ with $S'$ marked as processed if it is a terminal clique of $\ins \rescliqs (B', S')$.

Observe that the branching cannot decrease $\ctc_{\ins}(W_i)$ for any unprocessed terminal clique $W_i$, so the invariant that $\ctc_{\ins}(W_i) \ge |W_i|$ is maintained.
Note also that in each branch, the number of unprocessed terminal cliques decreases.
In particular, in the corner case when $S'$ is not a terminal clique of $\ins \rescliqs (B', S')$, it holds that the terminal cliques of $\ins \rescliqs (B', S')$ are a subset of the terminal cliques of $\ins$ but do not contain $W_i$.
As the number of important separators $S'$ of size at most $k$ is at most $4^k$ by \Cref{lem:impsep4k}, we branch to at most $4^k+1$ directions every time, and therefore as initially there are at most $k+2$ unprocessed terminal cliques, the branching tree has size at most $(4^k+1)^{k+2} = 2^{\OO(k^2)}$.

The leafs of the branching tree have no unprocessed terminal cliques and will be the instances we output.
Note that because when taking the important separator $S'$ we impose the condition $|S'| \le k+1-\ctc_{\ins}(W_i)$, which by $\ctc_{\ins}(W_i) \ge |W_i|$ implies that $|S'|+|W_i| \le k+1$, we can construct from a solution of $\ins \rescliqs (B', S')$ a solution of $\ins$ by just attaching a bag $W_i \cup S'$ as a neighbor of a bag containing $S'$.
Therefore, from any solution of the outputted instance we can construct a solution of the original instance.
Note that if the invariants of \Cref{lem:prebranching:yes,lem:prebranching:invproc,lem:prebranching:invmany} hold for some outputted instance, then it is a valid instance.
Therefore, it remains to argue that if the invariants of \Cref{lem:prebranching:yes,lem:prebranching:invproc,lem:prebranching:invmany} initially hold, then they hold for at least one of the branches.

Suppose that the invariants of \Cref{lem:prebranching:yes,lem:prebranching:invproc,lem:prebranching:invmany} hold for $\ins$, and let $W_i$ be an unprocessed terminal clique that we are branching on.
First, if there exists no solution that contains a degenerate internal separation $(A,S,B)$ so that $A$ intersects $W_i$, the invariants are clearly maintained by just marking $W_i$ processed as it does not change the set of solutions of the instance or the set $\allw$.
Then, suppose there exists a solution $(X, (T, \bag))$ that contains a degenerate internal separation $(A,S,B)$ so that $A$ intersects $W_i$.
First, by \Cref{lem:prebranching:invmany}, $W_i$ must be the only terminal clique that intersects $A$.
Then, we consider a hypothetical solution $(X, (T, \bag))$ and a degenerate internal separation $(A,S,B)$ of it so that $\reach_G(W_i, S)$ is the largest possible over all such $(X, (T, \bag))$ and $(A,S,B)$.
We note that $W_i \subseteq A \cup S$ and $\ow_{\ins}(W_i) \subseteq B \cup S$, so $S$ is a $(W_i, \ow_{\ins}(W_i))$-separator.
Then we let $S'$ to be a smallest important $(W_i, \ow_{\ins}(W_i))$-separator that dominates $S$, and will argue that the branch that selects $S'$ as the important separator will maintain the invariants.
We denote $(A', S', B') = (\reach_G(W_i, S'), S', V(G) \setminus (\reach_G(W_i, S') \cup S'))$.
Note that $W_i$ intersects $A'$ because $W_i$ intersects $\reach_G(W_i, S)$.
We will argue that $\ins \rescliqs (B',S')$ satisfies the invariants of \Cref{lem:prebranching:yes,lem:prebranching:invproc,lem:prebranching:invmany}.

\paragraph{\Cref{lem:prebranching:yes}.} First, to show that $\ins \rescliqs (B',S')$ is a yes-instance (i.e., \Cref{lem:prebranching:yes}), we use that by \Cref{lem:imp_sep_dom}, $S'$ is linked into $S \cap (A' \cup S')$.
In particular, we apply the pulling lemma (\Cref{lem:pull}) with the separation $(B', S', A')$, the torso tree decomposition $(X, (T, \bag))$, and the bag of $(T, \bag)$ that contains $S$, and obtain a torso tree decomposition $((X \cap B') \cup S', (T', \bag'))$ of no larger width that covers $\allw_{\ins \rescliqs (B', S')}$, showing that $\ins \rescliqs (B',S')$ is a yes-instance.

\paragraph{\Cref{lem:prebranching:invmany}.}
Then, to argue that $\ins \rescliqs (B',S')$ does not have solutions with degenerate internal separations $(A,S,B)$ with multiple terminal cliques intersecting $A$ (i.e., \Cref{lem:prebranching:invmany}), suppose that $(X_d, (T_d, \bag_d))$ is a solution of $\ins \rescliqs (B',S')$ that contains a degenerate internal separation $(A_d,S_d,B_d)$ so that at least two terminal cliques of $\ins \rescliqs (B',S')$ intersect $A_d$.
Now, because $W_i$ is the only terminal clique of $\ins$ that intersects $A'$ and by the fact that $|S'| \le |S|$ and $|W_i| + |S| \le k+1$ we can turn $(X_d, (T_d, \bag_d))$ into a solution $(X_d \cup W_i, (T_d', \bag_d'))$ of $\ins$ by just attaching a bag $W_i \cup S'$ to a bag of $(T_d', \bag_d')$ that contains $S'$.

Then, if $S' \subseteq S_d \cup B_d$, we consider the separation $(A_d, S_d, B_d \cup A')$ of $G$.
The set $\allw_{\ins}$ intersects $A_d$ because at least two terminal cliques of $\ins \rescliqs (B',S')$ intersect $A_d$, and it intersects $B_d \cup A'$ because $W_i$ intersects $A'$, and therefore $(A_d, S_d, B_d \cup A')$ is an internal separation for the solution $(X_d \cup W_i, (T_d', \bag_d'))$ of $\ins$.
In this case, as $A_d \subseteq B'$, all terminal cliques of $\ins \rescliqs (B',S')$ that intersect $A_d$ are also terminal cliques of $\ins$ that intersect $A_d$ and vice versa, and therefore at least two terminal cliques of $\ins$ intersect $A_d$ and $(A_d, S_d, B_d \cup A')$ is degenerate also for $\ins$, which would contradict that $\ins$ satisfies the invariant of \Cref{lem:prebranching:invmany}.

The other case is that $S'$ intersects $A_d$ and is a subset of $A_d \cup S_d$, in which case we consider the separation $(A_d \cup A', S_d, B_d)$ of $G$.
The set $\allw_{\ins}$ intersects $A_d \cup A'$ because $W_i$ intersects $A'$, and it intersects $B_d$ because $B_d \subseteq B'$ and $\allw_{\ins} \cap B' = \allw_{\ins \rescliqs (B',S')} \cap B'$, and therefore it is an internal separation for the solution $(X_d \cup W_i, (T_d', \bag_d'))$ of $\ins$.
At least two terminal cliques of $\ins$ intersect $A_d \cup A'$ because $W_i$ intersects $A'$, and some terminal clique of $\ins \rescliqs (B',S')$ that is also a terminal clique of $\ins$ must intersect $A_d$ because all but at most one terminal clique of $\ins \rescliqs (B',S')$ is a terminal clique of $\ins$ and at least two terminal cliques of $\ins \rescliqs (B',S')$ intersect $A_d$.
Now, because $S'$ intersects $A_d$ and all terminal vertices covered by terminal cliques of $\ins$ that are subsets of $A' \cup S'$ were mapped into $S'$ or to the superset of $S'$ in $\ins \rescliqs (B',S')$, the internal separation $(A_d \cup A', S_d, B_d)$ is degenerate for $\ins$, which would contradict that $\ins$ satisfies the invariant of \Cref{lem:prebranching:invmany}.

\paragraph{\Cref{lem:prebranching:invproc}.}
Then, to argue that $\ins \rescliqs (B',S')$ has no solution with a degenerate internal separation $(A,S,B)$ with a processed terminal clique intersecting $A$ (i.e., \Cref{lem:prebranching:invproc}), let $W_j$ be a processed terminal clique of $\ins \rescliqs (B',S')$ and suppose that $(X_d, (T_d, \bag_d))$ is a solution of $\ins \rescliqs (B',S')$ that contains a degenerate internal separation $(A_d,S_d,B_d)$ so that $W_j$  intersects $A_d$ but no other terminal clique of $\ins \rescliqs (B',S')$ intersects $A_d$. (Note that if also some other terminal clique intersects $A_d$, then we are in the already proven case of \Cref{lem:prebranching:invmany}.)
Again, because $W_i$ is the only terminal clique of $\ins$ that intersects $A'$ and by the fact that $|S'| \le |S|$ and $|W_i| + |S| \le k+1$ we can turn $(X_d, (T_d, \bag_d))$ into a solution $(X_d \cup W_i, (T_d', \bag_d'))$ of $\ins$ by just attaching a bag $W_i \cup S'$ to a bag of $(T_d, \bag_d)$ that contains $S'$.

Then, if $S' \subseteq S_d \cup B_d$, we consider the separation $(A_d, S_d, B_d \cup A')$ of $G$.
Now, because $S' \subseteq S_d \cup B_d$ but $W_j$ intersects $A_d$, we have that $W_j \neq S'$, implying that $W_j$ is also a terminal clique of $\ins$ and therefore $A_d$ intersects $\allw_{\ins}$, and $B_d \cup A'$ intersects $\allw_{\ins}$ because $W_i$ intersects $A'$, and therefore $(A_d, S_d, B_d \cup A')$ is an internal separation of the solution $(X_d \cup W_i, (T_d', \bag_d'))$ for $\ins$.
Because $W_j$ is a terminal clique of $\ins$ and $W_j \neq W_i$ and $W_j \neq S'$, it is also a processed terminal clique of $\ins$, and moreover because $A_d \subseteq B'$, also no other terminal cliques of $\ins$ intersect $A_d$, and therefore we contradict that $\ins$ satisfies the invariant of \Cref{lem:prebranching:invproc}.

%which would contradict that $\ins$ satisfies the invariant of \Cref{lem:prebranching:invproc}.

The other case is that $S'$ intersects $A_d$ and is a subset of $A_d \cup S_d$.
Note that in this case $S' \subseteq W_j$ and $W_j$ is the only terminal clique of $\ins \rescliqs (B',S')$ that is a superset of $S'$.
We then consider the separation $(A_d \cup A', S_d, B_d)$ of $G$.
Because $W_i$ intersects $A'$ and $B_d \subseteq B'$ this is an internal separation for $\ins$.
Note that because $W_i \subseteq A' \cup S'$, the original terminal vertices mapped to $W_i$ in $\ins$ are mapped to $W_j$ in $\ins \rescliqs (B',S')$.
If any other terminal clique of $\ins$ than $W_i$ intersects $A_d \cup A'$, then it must either be also a terminal clique of $\ins \rescliqs (B',S')$ that intersects $A_d$ (in particular, $W_j$) or a subset of $A' \cup S'$ whose covered original terminal vertices are mapped to $W_j$ in $\ins \rescliqs (B',S')$.
In that case, we contradict that $\ins$ satisfies the invariant of \Cref{lem:prebranching:invmany}.
Then, if the only terminal clique of $\ins$ that intersects $A_d \cup A'$ is $W_i$, we observe that $S_d$ is a $(W_i,\ow_{\ins}(W_i))$-separator, and moreover because $S_d$ does not intersect $A'$ and $S'$ intersects $A_d$ we have that $\reach_G(W_i, S') \subset \reach_G(W_i, S_d)$.
Because $\ctc_{\ins \rescliqs (B',S')}(W_j) \ge \ctc_{\ins}(W_i)$, it holds that $|W_i| + |S_d| \le k+1$, and therefore $(A_d \cup A', S_d, B_d)$ is a degenerate internal separation for $\ins$, and therefore as $\reach_G(W_i, S) \subseteq \reach_G(W_i, S') \subset \reach_G(W_i, S_d)$, it contradicts the choice of $(A,S,B)$.
\end{proof}

\subsection{Maintaining valid instances}
\label{subsec:maintvalid}
We introduce a new parameter of the instance based on the minimum order of a an internal separation in a solution.
In particular, we maintain a integer $q$ so that, informally speaking, in the success branches it is guaranteed that there exists no solution that contains an internal separation of order less than $q$.
More formally, we further extend the definition of an instance to now be a 5-tuple $\ins = (G, \{W_1, \ldots, W_t\}, k, \tc, q)$, re-using all previous definitions but now also including an integer $q \in [0,k+2]$.
We now say that $\ins$ is valid if it is a yes-instance, there exists no solution that contains a degenerate internal separation, and there exists no solution that contains an internal separation of order less than $q$.
Otherwise, we say that $\ins$ is invalid.
Note that a valid instance in the sense of \Cref{subsec:degenseps} can be turned into valid instance of this sense by just setting $q=0$.
The definitions $\rescliqs$, $\mergcliqs$, and $\addv$ used for manipulating the instance are extended so that they do not change $q$.
The rest of this section will be devoted to designing a branching algorithm for either finding a solution of $\ins$ or concluding that $\ins$ is invalid.

We first give a general lemma that will be used for arguing that if we break the instance by a separation $(A,S,B)$, then the resulting instances $\ins \rescliqs (A,S)$ and $\ins \rescliqs (B,S)$ are valid.

\begin{lemma}
\label{lem:breakvalid}
Let $\ins = (G, \{W_1, \ldots, W_t\}, k, \tc, q)$ be a valid instance and $(A,S,B)$ a separation of $G$ so that at least one terminal clique either intersects $A$ or is a superset of $S$ and at least one terminal clique either intersects $B$ or is a superset of $S$.
If both $\ins \rescliqs (A,S)$ and $\ins \rescliqs (B,S)$ are yes-instances, then both $\ins \rescliqs (A,S)$ and $\ins \rescliqs (B,S)$ are valid.
\end{lemma}
\begin{proof}
By symmetry it suffices to prove that $\ins \rescliqs (A,S)$ is valid, so for the sake of contradiction suppose that $\ins \rescliqs (A,S)$ is invalid.
Because $\ins \rescliqs (A,S)$ is a yes-instance, it has a solution that contains a degenerate internal separation or an internal separation of order $<q$.
Let $(X_A, (T_A, \bag_A))$ be such a solution of $\ins \rescliqs (A,S)$ and $(X_B, (T_B, \bag_B))$ any solution of $\ins \rescliqs (B,S)$.

By \Cref{lem:sepredgen}, $(X_A \cup X_B, (T_A, \bag_A) \cup_S (T_B, \bag_B))$ is a solution of $\ins$.
Let $(A', S', B')$ be the internal separation of $(X_A, (T_A, \bag_A))$.
Now, as $S$ is a clique in $G \rescliqs (A,S)$, we have two cases, either $S \subseteq B' \cup S'$ or $S \subseteq A' \cup S'$ and $S$ intersects $A'$.

First, if $S \subseteq B' \cup S'$, then consider the separation $(A', S', B \cup B')$ of $G$.
In this situation we have that $\allw_{\ins}$ intersects $A'$ because $\allw_{\ins \rescliqs (A,S)}$ intersects $A'$ and $S \subseteq B' \cup S'$, and that $\allw_{\ins}$ intersects $B \cup B'$ because if it does not intersect $B$, then $\allw_{\ins \rescliqs (A,S)} = \allw_{\ins}$, in which case it must intersect $B'$.
Therefore, $(A', S', B \cup B')$ is an internal separation of the solution $(X_A \cup X_B, (T_A, \bag_A) \cup_S (T_B, \bag_B))$ of $\ins$, and therefore if $|S'| < q$ we are done in this case.
If $(A', S', B')$ is degenerate internal separation of $(X_A, (T_A, \bag_A))$ in $\ins \rescliqs (A,S)$, then it is also a degenerate internal separation of $(X_A \cup X_B, (T_A, \bag_A) \cup_S (T_B, \bag_B))$ because the original terminal vertices mapped to terminal cliques that intersect $A'$ are the same in both $\ins$ and $\ins \rescliqs (A,S)$ because $A' \subseteq A$.

Then, if $S \subseteq A' \cup S'$ and $S$ intersects $A'$, consider the separation $(A' \cup B, S', B')$ of $G$.
By the same arguments as in the earlier case, we get that $(A' \cup B, S', B')$ is an internal separation of the solution $(X_A \cup X_B, (T_A, \bag_A) \cup_S (T_B, \bag_B))$ of $\ins$, and again if $|S'| < q$ we are immediately done.
Now, if $(A', S', B')$ is a degenerate internal separation of $(X_A, (T_A, \bag_A))$ in $\ins \rescliqs (A,S)$, then $(A' \cup B, S', B')$ is degenerate in $\ins$ because all original terminal vertices mapped to terminal cliques intersecting $B$ in $\ins$ are mapped to a superset of $S$ in $\ins \rescliqs (A,S)$, which intersects $A'$.
\end{proof}

It follows that breaking the instance by safe separations preserves the validity.

\begin{lemma}
\label{lem:safesepvalid}
If $\ins$ is valid and $(A,S,B)$ is a safe separation, then both $\ins \rescliqs (A,S)$ and $\ins \rescliqs (B,S)$ are valid.
\end{lemma}
\begin{proof}
By \Cref{lem:safesepyesyes} both $\ins \rescliqs (A,S)$ and $\ins \rescliqs (B,S)$ are yes-instances.
Because $(A,S,B)$ is a safe separation, at least one terminal clique intersects $A$ or is a superset of $S$ and at least one terminal clique intersects $B$ or is a superset of $S$.
Therefore by \Cref{lem:breakvalid} both $\ins \rescliqs (A,S)$ and $\ins \rescliqs (B,S)$ are valid.
\end{proof}

Then, we observe that safe separations $(A,S,B)$ of order $<q$ can be turned into internal separations of order $<q$ if $\allw_{\ins}$ intersects both $A$ and $B$.

\begin{lemma}
\label{lem:invalidsafesep}
If an instance $\ins = (G, \{W_1, \ldots, W_t\}, k, \tc, q)$ has a safe separation $(A,S,B)$ of order $<q$ so that $\allw_{\ins}$ intersects both $A$ and $B$, then $\ins$ is invalid.
\end{lemma}
\begin{proof}
If $\ins$ would be valid, then by \Cref{lem:safesepyesyes} both $\ins \rescliqs (A,S)$ and $\ins \rescliqs (B,S)$ are yes-instances.
Let $(X_A, (T_A, \bag_A))$ be a solution of $\ins \rescliqs (A,S)$ and $(X_B, (T_B, \bag_B))$ be a solution of $\ins \rescliqs (B,S)$.
By \Cref{lem:sepredgen}, $(X_A \cup X_B, (T_A, \bag_A) \cup_S (T_B, \bag_B))$ is a solution of $\ins$. 
However, now $(A,S,B)$ is an internal separation of order $<q$ and thus $\ins$ is invalid.
\end{proof}

%\begin{lemma}
%If $\ins = (G, \{W_1, \ldots, W_t\}, k, \tc, q)$ is valid and contains a terminal clique of size $\ge q$, then for each terminal clique $W_i$ it holds that $\flp_{\ins}(W_i) \ge \min(q, |W_i|)$.
%\end{lemma}
%\begin{proof}
%Let $W_j$ be a terminal clique of size $|W_j| \ge q$ and suppose that $\flow_G(W_i, W_j) < \min(q, |W_i|)$.
%Now, the minimum size separator between $W_i$ and $W_j$ would give a safe separation $(A,S,B)$ of order $<q$ with $W_i \cap A \neq \emptyset$ and $W_j \cap B \neq \emptyset$, contradicting \Cref{lem:invalidsafesep}.
%Therefore, $\flow_G(W_i, W_j) \ge \min(q, |W_i|)$, implying $\flp_{\ins}(W_i) \ge \min(q, |W_i|)$.
%\end{proof}

Then, we show that the notion of maximally merged plays well together with the definition of valid instances.

\begin{lemma}
\label{lem:maxmergedinvalid}
Let $\ins = (G, \{W_1, \ldots, W_t\}, k, \tc, q)$ be an instance.
If for all pairs of distinct terminal cliques $W_i,W_j$ either $|W_i \cup W_j| > k+1$ holds or $\ins \mergcliqs (W_i,W_j)$ is invalid, then either $\ins$ is maximally merged or $\ins$ is invalid.
\end{lemma}
\begin{proof}
Suppose this holds and $\ins$ is valid but not maximally merged.
Now, there exists a pair of distinct terminal cliques $W_i,W_j$ so that $|W_i \cup W_j| \le k+1$ and $\ins \mergcliqs (W_i,W_j)$ is a yes-instance but invalid.
Let $(X, (T, \bag))$ be a solution of $\ins \mergcliqs (W_i,W_j)$ and $(A,S,B)$ an internal separation of $(X, (T, \bag))$ that is either degenerate or has $|S| < q$.
Now, $(X, (T, \bag))$ is also a solution of $\ins$, and because $\allw_{\ins} = \allw_{\ins \mergcliqs (W_i, W_j)}$, $(A,S,B)$ is also an internal separation for $\ins$.
First, if $|S| < q$, then $\ins$ is invalid.
Then, if $(A,S,B)$ is degenerate for $\ins \mergcliqs (W_i,W_j)$, then for $\ins$ only a smaller number of original terminal vertices are mapped into terminal cliques that intersect $A$, so $(A,S,B)$ is also degenerate for $\ins$.
\end{proof}

\subsection{Branching}
\label{subsec:algok2branch}
In our algorithm, if there are less than 2 terminal cliques of size $\ge q$ we will perform leaf pushing branching to create more terminal cliques of size $\ge q$.
Unlike in the leaf pushing of \Cref{sec:algopstw}, in the leaf pushing of this section we will add the whole important separator to the terminal clique.
We show that like this, we can make a terminal clique of size $<q$ into size $\ge q$ by guessing a single important separator.
%This leaf pushing will be based on the following lemma.

\begin{lemma}
\label{lem:validleafpush}
Let $\ins = (G, \{W_1, \ldots, W_t\}, k, \tc, q)$ be a maximally merged valid instance with $t \ge 2$, no safe separators, and $W_i$ a potential forget-clique of $\ins$.
There is a vertex $w \in W_i \setminus \ow_{\ins}(W_i)$ and in the graph $G \setminus (W_i \setminus \{w\})$ a non-empty important $(\{w\}, \allw_{\ins} \setminus W_i)$-separator $S$ disjoint from $W_i$ so that $\ins \addv (W_i, S)$ is a valid instance and $q+1 \le |W_i \cup S| \le k+1$.
\end{lemma}
\begin{proof}
Let $w$ and $S$ be chosen so that $\ins \addv (W_i, S)$ is a yes-instance, which can be done by \Cref{lem:leafpush_impsep}.
Now, suppose that $\ins \addv (W_i, S)$ is invalid and let $(X, (T, \bag))$ and $(A',S',B')$ be the solution and the internal separation that show that $\ins \addv (W_i, S)$ is invalid.
Note that $(X, (T, \bag))$ is also a solution of $\ins$.

First, if $\allw_{\ins}$ intersects both $A'$ and $B'$, then $(X, (T, \bag))$ and $(A', S', B')$ directly show that also $\ins$ is invalid.
In particular, in the case when $(A', S', B')$ is degenerate for $\ins \addv (W_i, S)$, note that if a terminal clique of $\ins$ intersects $A'$ then also the corresponding terminal clique of $\ins \addv (W_i, S)$ also intersects $A'$, and therefore $(A', S', B')$ is also degenerate for $\ins$.

Then, if $\allw_{\ins}$ does not intersect $A'$ but $S$ does, we have that $W_i \subseteq S'$ and $S \subseteq A' \cup S'$.
In this case, denote $\reach_w = \reach_G(\{w\}, S \cup W_i \setminus \{w\})$ and consider the separation $(A'', S'', B'') = (A' \cup \reach_w, S' \setminus \reach_w, B' \setminus \reach_w)$.
This is a separation because the neighborhood of $\reach_w$ is a subset of $S \cup W_i \subseteq A' \cup S'$.
Moreover, it is an internal separation of $(X, (T, \bag))$ for $\ins$ because $w \in A''$ and $\allw_{\ins} \cap B' = \allw_{\ins} \cap B''$ because $\reach_w \cap \allw_{\ins} = \{w\}$ and  $w \notin B'$.
If $|S'| < q$ then this immediately shows that $\ins$ is invalid.
If $(A', S', B')$ is degenerate for $\ins \addv (W_i, S)$, then $(A'', S'', B'')$ is degenerate for $\ins$, because by the fact that $w \in W_i \setminus \ow_{\ins}(W_i)$, the only terminal clique of $\ins$ that intersects $A''$ is $W_i$, and $\ctc_{\ins}(W_i) = \ctc_{\ins \addv (W_i, S)}(W_i \cup S)$ in this case.

Then, if $\allw_{\ins}$ does not intersect $B'$ but $S$ does, we do an analogous argument, in particular we again denote $\reach_w = \reach_G(\{w\}, S \cup W_i \setminus \{w\})$ and consider the separation $(A'', S'', B'') = (A' \setminus \reach_w, S' \setminus \reach_w, B' \cup \reach_w)$.
By the same argument as previously, this is an internal separation of $(X, (T, \bag))$ for $\ins$.
Again, if $|S'| < q$ then $\ins$ is invalid.
If $(A', S', B')$ is degenerate for $\ins \addv (W_i, S)$, then $(A'', S'', B'')$ is degenerate for $\ins$, because if a terminal clique of $\ins$ intersects $A''$ then a corresponding terminal clique of $\ins \addv (W_i, S)$ intersects $A'$.

Finally, to show that $|W_i \cup S| \ge q+1$, we have that if $|W_i \cup S| \le q$ would hold, then $(W_i \setminus \{w\}) \cup S$ would be a $(\{w\}, \allw_{\ins} \setminus W_i)$-separator of size $< q$.
This would give an internal separation $(A,(W_i \setminus \{w\}) \cup S,B)$ with $w$ intersecting $A$ and $\allw_{\ins} \setminus (W_i \cup S)$ intersecting $B$.
Note that $\allw_{\ins} \setminus (W_i \cup S) \neq \emptyset$ because otherwise $\ins$ would not be maximally merged.
\end{proof}

Then, once there are at least two terminal cliques of size $\ge q$, the algorithm will guess how a hypothetical internal separation of order $q$ would separate the terminal cliques, and find a corresponding separation by guessing an important separator.
For this argument it will be crucial that the separator $S$ of such an internal separation $(A,S,B)$ will be linked into the terminal cliques of size $\ge q$.

\begin{lemma}
\label{lem:intseplinked}
Let $\ins = (G, \{W_1, \ldots, W_t\}, k, \tc, q)$ be a valid instance, and $(A,S,B)$ an internal separation of a solution of $\ins$ of order $|S| = q$.
Then $S$ is linked into any terminal clique of $\ins$ of size at least $q$.
\end{lemma}
\begin{proof}
Let $W_i$ be a terminal clique of $\ins$ of size $|W_i| \ge q$ and suppose $S$ is not linked into $W_i$.
By symmetry suppose $W_i \subseteq A \cup S$, and by definition of internal separation let $W_j$ be a terminal clique that intersects $B$.
Now, let $S'$ be a minimum size $(W_i, S)$-separator, in particular having size $|S'| < q$ and both $W_i$ and $S$ linked into $S'$.
Note that $S'$ also separates $W_j$ from $W_i$, in particular $S'$ gives a separation $(A', S', B')$ so that $W_i \subseteq A' \cup S'$ and $S \cup W_j \subseteq B' \cup S'$ and moreover $W_j \cap B' \neq \emptyset$.

Let $(X, (T, \bag))$ be the solution of $\ins$ whose internal separation $(A,S,B)$ is.
Note that $(T, \bag)$ has a bag containing $S$, and a bag containing $W_i$.
We use the pulling lemma (\Cref{lem:pull}) with $(X, (T, \bag))$, $(A', S', B')$ and the bag containing $S$ to construct a solution of $\ins \rescliqs (A',S')$ and then with $(X, (T, \bag))$, $(B', S', A')$ and the bag containing $W_i$ to construct a solution of $\ins \rescliqs (B', S')$.
Now, by combining the solutions of $\ins \rescliqs (A',S')$ and $\ins \rescliqs (B', S')$ using \Cref{lem:sepredgen} we get a solution of $\ins$ whose internal separation $(A', S', B')$ is.
This implies that $\ins$ is invalid because $|S'| < q$, $W_i$ intersects $A'$, and $W_j$ intersects $B'$.
\end{proof}

We then introduce notation for arguing about guessing how a hypothetical internal separation of order $q$ separates the terminal cliques.
Let $\ins = (G, \{W_1, \ldots, W_t\}, k, \tc, q)$ be an instance. 
For $t' \subseteq [t]$, we denote $\allw_{\ins}[t'] = \bigcup_{i \in t'} W_i$.
Let $(t_L, t_R)$ be a partition of $[t]$ into two non-empty sets, in particular representing a partition of the terminal cliques.
We call $(t_L, t_R)$ $q$-biased if $|W_i| \ge q$ implies that $i \in t_R$.

We then give the main lemma that asserts how internal separations of order $q$ can be guessed by guessing the partition $(t_L, t_R)$ of terminal cliques induced by them and an important $(\allw_{\ins}[t_L], \allw_{\ins}[t_R])$-separator.

\begin{lemma}
\label{lem:impsepbreak}
Let $\ins = (G, \{W_1, \ldots, W_t\}, k, \tc, q)$ be a maximally merged valid instance that has no safe separations and has at least two terminal cliques of size at least $q$.
Suppose that $\ins$ has a solution that has an internal separation of order $q$.
Then, there exists a $q$-biased partition $(t_L, t_R)$ of $[t]$ and an important $(\allw_{\ins}[t_L], \allw_{\ins}[t_R])$-separator $S$ of size $|S|=q$ with $\reach_G(\allw_{\ins}[t_L], S) \neq \emptyset$, corresponding to a separation $(A, S, B) = (\reach_G(\allw_{\ins}[t_L],S), S, V(G) \setminus (\reach_G(\allw_{\ins}[t_L],S) \cup S))$ so that both $\ins \rescliqs (A, S)$ and $\ins \rescliqs (B, S)$ are valid instances.
\end{lemma}
\begin{proof}
Let $(A,S,B)$ be an internal separation of order $q$ of a solution $(X, (T, \bag))$ of $\ins$.
Note that because $\ins$ does not have safe separations, either all terminal cliques of size $\ge q$ intersect $A$ or all terminal cliques of size $\ge q$ intersect $B$.
By permuting $A$ and $B$ if necessary, assume that all terminal cliques of size $\ge q$ intersect $B$.
Let $(t_L, t_R)$ be the partition of $[t]$ that is obtained by assigning terminal cliques that intersect $B$ into $t_R$ and the others into $t_L$.
Note that at least one terminal clique intersects $A$ because $(A,S,B)$ is an internal separation and at least two terminal cliques of size $\ge q$ intersect $B$.

Now, $S$ is a $(\allw_{\ins}[t_L], \allw_{\ins}[t_R])$-separator.
Moreover, $S$ is a minimal $(\allw_{\ins}[t_L], \allw_{\ins}[t_R])$-separator, because otherwise the subset of $S$ would give an internal separation of order $<q$, meaning that $\ins$ would not be valid.
Let $W_i$ be a terminal clique of size $\ge q$.
By \Cref{lem:intseplinked}, $S$ is linked into $W_i$.
Let $S'$ be a smallest important $(\allw_{\ins}[t_L], \allw_{\ins}[t_R])$-separator that dominates $S$.
Because $S$ is a minimal $(\allw_{\ins}[t_L], \allw_{\ins}[t_R])$-separator, by \Cref{lem:mini_sep_impdom} $S'$ is a $(S, \allw_{\ins}[t_R])$-separator, which implies that $|S'| = q$ and $S'$ is linked into $W_i$.
Also, by minimality of $S$ and \Cref{lem:imp_sep_dom}, we have that $S'$ is linked into $S$.
Moreover, as $\allw_{\ins}[t_L]$ intersects $A$, we have that $\allw_{\ins}[t_L]$ intersects $\reach_G(\allw_{\ins}[t_L], S')$ and in particular $\reach_G(\allw_{\ins}[t_L], S') \neq \emptyset$.

Now, let $(A',S',B') = (\reach_G(\allw_{\ins}[t_L],S'), S', V(G) \setminus (\reach_G(\allw_{\ins}[t_L],S') \cup S'))$.
Observe that $S \subseteq A' \cup S'$ and $W_i \subseteq B' \cup S'$.
We then use the pulling lemma (\Cref{lem:pull}) to construct solutions of $\ins \rescliqs (A', S')$ and $\ins \rescliqs (B', S')$.
In particular, a solution of $\ins \rescliqs (A', S')$ is constructed by applying the lemma with the torso tree decomposition $(X, (T, \bag))$, the separation $(A', S', B')$, and the node of $(T, \bag)$ whose bag contains $W_i$.
Symmetrically, a solution of $\ins \rescliqs (B', S')$ is constructed by applying the lemma with the torso tree decomposition $(X, (T, \bag))$, the separation $(B', S', A')$, and the node of $(T, \bag)$ whose bag contains $S$.

Now, both $\ins \rescliqs (A', S')$ and $\ins \rescliqs (B', S')$ are yes-instances, so it remains to prove that they are valid.
First, because $\ins$ is maximally merged and there are at least two terminal cliques of size $\ge q$, it holds that $|\allw_{\ins}[t_R]| \ge q+1$, implying that $\allw_{\ins}[t_R]$ intersects $B'$.
Also, as argued before $\allw_{\ins}[t_L]$ intersects $A'$.
Therefore, by \Cref{lem:breakvalid} both $\ins \rescliqs (A', S')$ and $\ins \rescliqs (B', S')$ are valid.
\end{proof}

\subsection{The algorithm}
\label{subsec:thealgok2}
We then describe the $2^{\OO(k^2)} nm$ time algorithm for \stw.

Given input $(G, W, k)$, the algorithm first uses pre-branching (\Cref{lem:prebranching}) to enumerate $2^{\OO(k^2)}$ instances of \pstw, so that any solution to any of the instances can in $k^{\OO(1)} m$ time be turned into a torso tree decomposition in $G$ of width $k$ that covers $W$, and moreover if such a torso tree decomposition exists, then at least one of the instances is valid.
For each resulting instance, we then use a recursive procedure that either concludes that a given instance is invalid, or returns a solution to the instance.
This recursive procedure is described in pseudocode \Cref{alg:exsubtw}, and we also give a detailed description of it next.

\begin{algorithm}[t]
\caption{The recursive procedure of a $2^{\OO(k^2)} nm$ time algorithm for \stw.\label{alg:exsubtw}}
\textbf{Input:} Instance $\ins = (G, \{W_1, \ldots, W_t\}, k, \tc, q)$.\\
\textbf{Output:} Either a solution of $\ins$ or \invalid.
\begin{algorithmic}[1]
\If {$t \le 1$ or $|V(G)| \le k+2$} \Return Case-analysis($\ins$)\label{alg:exsubtw:bf} \Comment{\Cref{lem:smallcases}}
\EndIf
\If {Exists a safe separation $(A,S,B)$}\label{alg:exsubtw:safesepcheck}
\If {$|S|<q$ and $\allw_{\ins}$ intersects both $A$ and $B$} \label{alg:exsubtw:safesepif}
\State \Return \invalid \label{alg:exsubtw:safesepcheckinvalid}
\Else \label{alg:exsubtw:safesepelse}
\State \Return Combine(Solve($\ins \rescliqs (A,S)$), Solve($\ins \rescliqs (B,S)$))\label{alg:exsubtw:safesepcheckcombine}
\EndIf
\EndIf

\ForAll{$i, j \in [t]$ with $i \neq j$ and $|W_i \cup W_j| \le k+1$}\label{alg:exsubtw:forcliqupairs}
\State $sol \gets$ Solve($\ins \mergcliqs (W_i,W_j)$)\label{alg:exsubtw:mergrec}
\If {$sol \neq \invalid$} \Return $sol$\label{alg:exsubtw:retmc}
\EndIf
\EndFor
\If {$q > k+1$} \Return \invalid \label{alg:exsubtw:retqlarge}
\EndIf
%\If {Exists $i \neq j$ with $W_i = W_j$} \Return \invalid \label{alg:exsubtw:reteqcliqs} %% no longer needed because W_1,...,W_t is a set 
%\EndIf

\If {Less than $2$ terminal cliques of size $\ge q$} \Comment{\Cref{lem:validleafpush}} \label{alg:exsubtw:caseless}
\ForAll{$i \in [t]$ so that $|W_i| < q$ and exists $j \neq i$ with $|W_j| \ge |W_i|$}\label{alg:exsubtw:itleafpushclique}
\ForAll{$w \in W_i$}\label{alg:exsubtw:itleafpushw}
\ForAll{Important $(\{w\}, \allw_{\ins} \setminus W_i)$-separators $S$ in $G \setminus (W_i \setminus \{w\})$ with $|S| \le k$}\label{alg:exsubtw:itleafpushimps}
\If {$q+1 \le |W_i \cup S| \le k+1$}\label{alg:exsubtw:itleafpushif}
\State $sol \gets$ Solve($\ins \addv (W_i,S)$)\label{alg:exsubtw:leafpush1rec}
\If {$sol \neq \invalid$} \Return $sol$\label{alg:exsubtw:retadd1}
\EndIf
\EndIf
\EndFor
\EndFor
\EndFor
\Else \Comment{\Cref{lem:impsepbreak}} \label{alg:exsubtw:casemore}
\ForAll{$q$-biased bipartitions $(t_L,t_R)$ of $[t]$}\label{alg:exsubtw:bbfor}
\ForAll{Important $(\allw_{\ins}[t_L], \allw_{\ins}[t_R])$-separators $S$ with $|S| = q$}\label{alg:exsubtw:impsepfor}
\State Let $(A,S,B) = (\reach_G(\allw_{\ins}[t_L], S), S, V(G) \setminus (\reach_G(\allw_{\ins}[t_L], S) \cup S))$\label{alg:exsubtw:asb}
\If {$\sum_{W_i \mid W_i \cap A \neq \emptyset} \ctc_{\ins}(W_i) > k+1-q$} \label{alg:exsubtw:degencheck}
\State $sol \gets$ Combine($S$, Solve($\ins \rescliqs (A,S)$), Solve($\ins \rescliqs (B,S)$))\label{alg:exsubtw:leafpush2rec}
\If {$sol \neq \invalid$} \Return $sol$\label{alg:exsubtw:retadd2}
\EndIf
\EndIf
\EndFor
\EndFor
\EndIf
\State \Return Solve($(G, \{W_1, \ldots, W_t\}, k, \tc, q+1)$) \label{alg:exsubtw:lastret}
\end{algorithmic}
\end{algorithm}

First, on \Cref{alg:exsubtw:bf} the algorithm uses \Cref{lem:smallcases} to handle the corner cases of $t=1$ and $|V(G)|\le k+2$.
Then, on \Cref{alg:exsubtw:safesepcheck,alg:exsubtw:safesepif,alg:exsubtw:safesepcheckinvalid,alg:exsubtw:safesepelse,alg:exsubtw:safesepcheckcombine} the reduction by safe separations is performed.
In particular, if there exists a safe separation $(A,S,B)$, then if $|S|<q$ and $\allw$ intersects both $A$ and $B$, we can by \Cref{lem:invalidsafesep} conclude that the instance is invalid.
Otherwise, we recursively solve the instances $\ins \rescliqs (A,S)$ and $\ins \rescliqs (B,S)$ (recursive application of the algorithm is denoted by the function ``Solve'' in the pseudocode), and if both of them return a solution then we return the solution obtained from combining them, and if either of them return \invalid then we return \invalid.
In particular, the function ``Combine'' on \Cref{alg:exsubtw:safesepcheckcombine} denotes an operation that returns \invalid if either of its arguments is \invalid, and if its arguments are a solution $(X_A, (T_A, \bag_A))$ of $\ins \rescliqs (A,S)$ and a solution $(X_B, (T_B, \bag_B))$ of $\ins \rescliqs (B,S)$ then it returns the solution $(X_A \cup X_B, (T_A, \bag_A) \cup_S (T_B, \bag_B))$ of $\ins$.

Then, on \Cref{alg:exsubtw:forcliqupairs,alg:exsubtw:mergrec,alg:exsubtw:retmc} the algorithm does terminal clique merging branching.
In particular, the algorithm branches on merging all pairs of terminal cliques $W_i,W_j$ with $|W_i \cup W_j| \le k+1$ and returns a solution if any of the branches returned a solution.
After this, by \Cref{lem:maxmergedinvalid} we can assume that $\ins$ is either maximally merged or invalid.
This is used on \Cref{alg:exsubtw:retqlarge} to justify that if $q > k+1$ we can return \invalid because any solution of a maximally merged instance with at least two terminal cliques must contain an internal separation.

For the main branching of the algorithm there are two cases.
Either there are less than 2 terminal cliques of size $\ge q$, or there are at least 2 terminal cliques of size $\ge q$.
We first describe the case when there are less than 2 terminal cliques of size $\ge q$.
In this case, on \Cref{alg:exsubtw:caseless,alg:exsubtw:itleafpushclique,alg:exsubtw:itleafpushw,alg:exsubtw:itleafpushimps,alg:exsubtw:itleafpushif,alg:exsubtw:leafpush1rec,alg:exsubtw:retadd1} the algorithm performs leaf pushing branching according to \Cref{lem:validleafpush}.
In particular, the algorithm guesses a potential forget-clique $W_i$ that is not a uniquely largest terminal clique, a vertex $w \in W_i$, and an important $(\{w\}, \allw \setminus W_i)$-separator $S$ in the graph $G \setminus (W_i \setminus \{w\})$ so that $q+1 \le |W_i \cup S| \le k+1$, and branches on adding $S$ to $W_i$.
For iterating over such important separators, we use the algorithm of \Cref{lem:impsep4k} to iterate over all important separators of size at most $k$ and check the conditions.
Note that the purpose of this branching is to increase the number of terminal cliques of size $\ge q$.

When there are at least 2 terminal cliques of size $\ge q$, the algorithm branches on \Cref{alg:exsubtw:casemore,alg:exsubtw:bbfor,alg:exsubtw:impsepfor,alg:exsubtw:asb,alg:exsubtw:degencheck,alg:exsubtw:leafpush2rec,alg:exsubtw:retadd2} on how the internal separation of size $q$ would partition the terminal cliques in the solution, in particular, according to \Cref{lem:impsepbreak}.
The algorithm guesses the $q$-biased bipartition $(t_L, t_R)$ of $[t]$ and an important $(\allw[t_L], \allw[t_R])$-separator $S$ of size $q$, then denotes the separation corresponding to it by $(A,S,B) = (\reach_G(\allw_{\ins}[t_L], S), S, V(G) \setminus (\reach_G(\allw_{\ins}[t_L], S) \cup S))$, and if this separation would not be a degenerate internal separation solves the instances $\ins \rescliqs (A,S)$ and $\ins \rescliqs (B,S)$ recursively and combines the solutions in the same manner as when recursing on safe separations.
Again, we use the algorithm of \Cref{lem:impsep4k} for iterating over such important separators.
Finally, on \Cref{alg:exsubtw:lastret}, if none of the branches returned a solution the algorithm does a recursive call with an increased value of $q$.

The algorithm can clearly be implemented in polynomial space, in particular, on \Cref{alg:exsubtw:itleafpushimps,alg:exsubtw:impsepfor} we use the polynomial space enumeration of important separators of \Cref{lem:impsep4k}.
We then prove the correctness of \Cref{alg:exsubtw}.
Its running time will be analyzed in \Cref{subsec:alg2tcana}.

First we show that the algorithm is correct when it returns a solution.

\begin{lemma}
If \Cref{alg:exsubtw} returns a solution, then it is a solution of $\ins$.
\end{lemma}
\begin{proof}
We prove the lemma by induction on the recursion tree.
When the algorithm returns on \Cref{alg:exsubtw:bf} from the case analysis, this follows from the correctness of the case analysis.

When breaking the instance by a safe separation $(A,S,B)$ and returning on \Cref{alg:exsubtw:safesepcheckcombine} a solution formed by combining solutions of $\ins \rescliqs (A,S)$ and $\ins \rescliqs (B,S)$, the correctness follows from induction and \Cref{lem:sepredgen}.
For terminal clique merging on \Cref{alg:exsubtw:retmc} and leaf pushing on \Cref{alg:exsubtw:retadd1} this follows from induction and the fact that any solution of $\ins \mergcliqs (W_i,W_j)$ or $\ins \addv (W_i,S)$ is also solution of $\ins$.
When combining solutions on \Cref{alg:exsubtw:retadd2} the correctness again follows \Cref{lem:sepredgen} and induction.
For the final line \Cref{alg:exsubtw:lastret} it follows from induction.
\end{proof}

Then we show that the algorithm is correct when it returns that $\ins$ is invalid.

\begin{lemma}
If \Cref{alg:exsubtw} returns \invalid, then $\ins$ is invalid.
\end{lemma}
\begin{proof}
We prove the lemma by induction on the recursion tree.
When the algorithm returns on \Cref{alg:exsubtw:bf} from the case analysis, this follows from the correctness of the case analysis.

When returning \invalid on \Cref{alg:exsubtw:safesepcheckinvalid} if a safe separation $(A,S,B)$ with $|S|<q$ and $\allw_{\ins}$ intersecting both $A$ and $B$ exists, the correctness is given in \Cref{lem:invalidsafesep}.
%breaking the instance by a safe separation $(A,S,B)$ and returning \invalid on \Cref{alg:exsubtw:safesepcheckinvalid}, the correctness follows from \Cref{lem:invalidsafesep}.
For returning \invalid from the safe separation recursion on \Cref{alg:exsubtw:safesepcheckcombine}, we have that if $\ins$ is valid, then by \Cref{lem:safesepvalid} both $\ins \rescliqs (A,S)$ and $\ins \rescliqs (B,S)$ are valid, and therefore the correctness follows from induction.

%on \Cref{alg:exsubtw:safesepcheckcombine} after finding a safe separation $(A,S,B)$, if $\ins$ is valid, then by \Cref{lem:safesepvalid} both $\ins \rescliqs (A,S)$ and $\ins \rescliqs (B,S)$ are valid, and therefore the lemma follows from induction.

After terminal clique merging on \Cref{alg:exsubtw:forcliqupairs,alg:exsubtw:mergrec,alg:exsubtw:retmc}, by induction and \Cref{lem:maxmergedinvalid} we may assume that $\ins$ is either invalid or maximally merged, and in particular in the rest of this proof we may assume that $\ins$ is maximally merged, as the conclusion trivially holds when $\ins$ is invalid.
Then, for returning \invalid on \Cref{alg:exsubtw:retqlarge} if $q > k+1$, we have that in this case if $t \ge 2$, and $\ins$ is maximally merged it must be invalid because then any solution must either contradict that $\ins$ is valid or have all of the terminal cliques in a single bag, which would contradict that $\ins$ is maximally merged.

For returning \invalid on the final \Cref{alg:exsubtw:lastret} there are two cases depending on the number of terminal cliques of size $\ge q$.

First, if there are less than 2 terminal cliques of size $\ge q$, we use \Cref{lem:validleafpush}.
Towards contradiction assume that $\ins$ is valid but we return \invalid from \Cref{alg:exsubtw:lastret}.
By \Cref{lem:twopotforcli}, $\ins$ has at least two potential forget-cliques, and some iteration of \Cref{alg:exsubtw:itleafpushclique} fixed such potential forget-clique $W_i$, and some iteration of \Cref{alg:exsubtw:itleafpushw,alg:exsubtw:itleafpushimps,alg:exsubtw:itleafpushif} fixed a vertex $w \in W_i$ and an important $(\{w\}, \allw_{\ins} \setminus W_i)$-separator $S$ in $G \setminus (W_i \setminus \{w\})$ satisfying the conditions of \Cref{lem:validleafpush}.
Now, by \Cref{lem:validleafpush}, $\ins \addv (W_i, S)$ is a valid instance, so by induction \Cref{alg:exsubtw} would return on \Cref{alg:exsubtw:retadd1}.

Then, if there are at least 2 terminal cliques of size $\ge q$, we use \Cref{lem:impsepbreak}.
Towards contradiction assume that $\ins$ is valid but we return \invalid from \Cref{alg:exsubtw:lastret}.
By induction, $(G, \{W_1, \ldots, W_t\}, k, \tc, q+1)$ is invalid, implying that either $\ins$ is invalid or there exists a solution of $\ins$ that contains an internal separation of order $q$.
Therefore, we can assume that $\ins$ satisfies the preconditions of \Cref{lem:impsepbreak}.
Now, let $(t_L, t_R)$ be the $q$-biased partition of $[t]$ and $S$ the important $(\allw_{\ins}[t_L], \allw_{\ins}[t_R])$-separator of size $|S| = q$ with $\reach_G(\allw_{\ins}[t_L], S) \neq \emptyset$ given by \Cref{lem:impsepbreak}, and let $(A,S,B) = (\reach_G(\allw_{\ins}[t_L], S), S, V(G) \setminus (\reach_G(\allw_{\ins}[t_L], S) \cup S))$.
By \Cref{lem:impsepbreak}, both $\ins \rescliqs (A,S)$ and $\ins \rescliqs (B,S)$ are valid.
Some iteration of \Cref{alg:exsubtw:bbfor,alg:exsubtw:impsepfor} will fix such $(t_L, t_R)$ and $S$, and therefore some iteration of \Cref{alg:exsubtw:asb} will fix such $(A,S,B)$.
If $\sum_{W_i \mid W_i \cap A \neq \emptyset} \ctc_{\ins}(W_i) > k+1-q$ holds, we would have returned from \Cref{alg:exsubtw:retadd2} and obtain a contradiction.

It remains to show that if $\ins$ is valid, then $\sum_{W_i \mid W_i \cap A \neq \emptyset} \ctc_{\ins}(W_i) > k+1-q$ indeed holds for such $(A,S,B)$.
Because $\reach_G(\allw_{\ins}[t_L], S) \neq \emptyset$, we have that $\allw_{\ins}[t_L]$ intersects $A$, and because there are at least two terminal cliques of size $\ge q$ and $\ins$ is maximally merged, we have that $\allw_{\ins}[t_R]$ intersects $B$.
Therefore, because both $\ins \rescliqs (A,S)$ and $\ins \rescliqs (B,S)$ are valid, we can construct from their solutions a solution of $\ins$ so that $(A,S,B)$ is its internal separation.
However, if $\sum_{W_i \mid W_i \cap A \neq \emptyset} \ctc_{\ins}(W_i) \le k+1-q$ would hold, then this would be a degenerate internal separation.
\end{proof}

\subsection{Running time analysis}
\label{subsec:alg2tcana}
We analyze the running time of \Cref{alg:exsubtw}.
Throughout, we let $\ql = k+2 - k/\log_2 k$, and consider $q \ge \ql$ to be large and $q < \ql$ to be small.
In order to simplify the analysis we will also assume that $k \ge \kb$.
If $k < \kb$ we use the algorithm of \Cref{sec:algopstw}, which works in $\OO(nm)$ time in this case.

Informally, the idea of choosing $\ql = k+2 - k/\log_2 k$ is that once $q \ge \ql$, a single leaf push takes a terminal clique to size at least $\ql$, and after that the terminal clique is only $k/\log_2 k$ vertices away from the maximum size, implying that we can ``pay'' $k^{\OO(1)}$ branching degree for increasing the size by one and still end up with $2^{\OO(k^2)} n^{\OO(1)}$ running time.
In the other case, when $q < \ql$, we use the absence of degenerate internal separations to argue that there should be at least $k/\log_2 k$ original terminal vertices behind any internal separation of size $q$ and amortize the cost of the branching on the original terminal vertices.
We note that any choice of $\ql$ between $k + 2 - k / \log_2 k$ and $k + 2 - \log_2 k$ would be sufficient for the analysis, but we fix the value $\ql = k+2 - k/\log_2 k$.

%all terminal cliques can be taken to be size $\ql$ by single leaf pushes, and then they are at most $k/\log_2 k$ from their maximum size, 

We now define the measure of a terminal clique based on cases depending on $q$ and $\ql$.
We note that even though we use similar notation to \Cref{sec:algopstw}, the measure of this section is unrelated to the measure of \Cref{sec:algopstw}.
The measure of a terminal clique $W_i$ is

\[\potl_{\ins}(W_i) = \begin{cases}
(k+2 - \min(q, |W_i|)) \cdot \log_2 k + 4k &\text{if } q \ge \ql \text{ and } |W_i| \ge \ql\\
6k &\text{if } q \ge \ql \text{ and } |W_i| < \ql\\
(k+2 - \min(q, |W_i|)) \cdot \ctc_{\ins}(W_i) + 6k &\text{if $q < \ql$.}
\end{cases}
\]

Note that if $q \ge \ql$ and $|W_i| \ge \ql$, then $4k \le \potl_{\ins}(W_i) \le 5k$, implying that when $q \ge \ql$, we have $4k \le \potl_{\ins}(W_i) \le 6k$, implying $\sum_{i=1}^t \potl_{\ins}(W_i) \le 6kt \le \OO(k^2)$.
If $q < \ql$, then we have a lower bound of $6k \le \potl_{\ins}(W_i)$ on the measure of a single terminal clique, and the sum is upper bounded by $\sum_{i=1}^t \potl_{\ins}(W_i) \le 6kt + |\ogw|(k+2) \le \OO(k^2)$.

We then let $\cq(\ins) = \min(2, \text{number of terminal cliques of size}\ge q)$ and define the measure of the instance to be

\[
\potl(\ins) = \begin{cases}
(k+2-q)\cdot 3k + (2-\cq(\ins)) \cdot k + \sum_{i=1}^t \potl_{\ins}(W_i) &\text{if } t \ge 2 \text{ and}\\
1 &\text{if } t=1.
\end{cases}
\]

Note that $\potl(\ins) \le \OO(k^2)$ and if $t \ge 2$ then $\potl(\ins) \ge 8k$.

We then give a general lemma for arguing that breaking the instance by separations of size at least $q$ does not increase the measure, and moreover decreases the measure by at least $k$ if this decreases the number of terminal cliques.

\begin{lemma}
\label{lem:al2potsep}
Let $\ins = (G, \{W_1, \ldots, W_t\}, k, \tc, q)$ be an instance with $t \ge 2$ terminal cliques and $(A,S,B)$ a separation with $|S| \ge q$.
If the number of terminal cliques of $\ins \rescliqs (A,S)$ is $t$, then $\potl(\ins \rescliqs (A,S)) \le \potl(\ins)$, and if the number of terminal cliques of $\ins \rescliqs (A,S)$ is less than $t$, then $\potl(\ins \rescliqs (A,S)) \le \potl(\ins)-k$.
\end{lemma}
\begin{proof}
First, consider the case when the number of terminal cliques of $\ins \rescliqs (A,S)$ is $t$.
In this case $\{W_1, \ldots, W_t\} \rescliqs (A,S) = \{W_1, \ldots, W_t\} \cup \{S'\} \setminus \{W_i\}$ for some $S' \supseteq S$ and $W_i \subseteq B \cup S$ or $W_i = S'$.
Note that now, $\ctc_{\ins \rescliqs (A,S)}(S') = \ctc_{\ins}(W_i)$ and $\min(q, |W_i|) \le \min(q, |S'|)$ because $|S'| \ge q$, so $\potl_{\ins \rescliqs (A,S)}(S') \le \potl_{\ins}(W_i)$, implying together with $\cq(\ins \rescliqs (A,S)) \ge \cq(\ins)$ that $\potl(\ins \rescliqs (A,S)) \le \potl(\ins)$.

Then, when the number of terminal cliques of $\ins \rescliqs (A,S)$ is less than $t$, first consider the case when $q \ge \ql$.
In this case, $\potl_{\ins}(W_i) \ge 4k$ for any $W_i$ and $\potl_{\ins \rescliqs (A,S)}(S') \le 5k$ when $S' \supseteq S$ because $|S| \ge q$, so we decrease the measure by at least $3k$ from the sum over terminal cliques, and increase by at most $k$ from the $\cq(\ins)$ measure as $\cq(\ins)-\cq(\ins \rescliqs (A,S)) \le 1$, so in total we decrease the measure by at least $2k$.

Then, when $q < \ql$, first consider the case when $\{W_1, \ldots, W_t\}$ contains a terminal clique $S' \supseteq S$ with $S' \subseteq A \cup S$, in which case $\{W_1, \ldots, W_t\} \rescliqs (A,S) \subset \{W_1, \ldots, W_t\}$.
In this case we have that 
\[\ctc_{\ins \rescliqs (A,S)}(S') = \ctc_{\ins}(S') + \sum_{W_i \in \{W_1, \ldots, W_t\} \setminus \{W_1, \ldots, W_t\} \rescliqs (A,S)} \ctc_{\ins}(W_i).\]
Therefore, as $\min(|W_i|, q) \le \min(|S'|, q)$, we have that
\[\potl_{\ins \rescliqs (A,S)}(S') \le \potl_{\ins}(S') + \sum_{W_i \in \{W_1, \ldots, W_t\} \setminus \{W_1, \ldots, W_t\} \rescliqs (A,S)} \potl_{\ins}(W_i) - 6k,\]
implying
\[\potl(\ins \rescliqs (A,S)) \le \potl(\ins) - 6k \cdot |\{W_1, \ldots, W_t\} \setminus \{W_1, \ldots, W_t\} \rescliqs (A,S)| \le \potl(\ins) - 6k.\]

The final case is that $\{W_1, \ldots, W_t\}$ does not contain a terminal clique $S' \supseteq S$ with $S' \subseteq A \cup S$, in which case $\{W_1, \ldots, W_t\} \rescliqs (A,S) = \{W_1, \ldots, W_t\} \cup \{S\} \setminus \{W_i \in \{W_1, \ldots, W_t\} \mid W_i \subseteq S \cup B\}$ and
\[\ctc_{\ins \rescliqs (A,S)}(S) = \sum_{W_i \in \{W_1, \ldots, W_t\} \mid W_i \subseteq S \cup B} \ctc_{\ins}(W_i).\]
By $\min(|W_i|, q) \le \min(|S|, q)$ this implies that
\[\potl_{\ins \rescliqs (A,S)}(S) \le 6k+\sum_{W_i \in \{W_1, \ldots, W_t\} \mid W_i \subseteq S \cup B} \potl_{\ins}(W_i) - 6k,\]
which by $|\{W_i \in \{W_1, \ldots, W_t\} \mid W_i \subseteq S \cup B\}| \ge 2$ and $\cq(\ins)-\cq(\ins \rescliqs (A,S)) \le 1$ implies that $\potl(\ins \rescliqs (A,S)) \le \potl(\ins)-5k$.
\end{proof}

We then show that breaking the instance by a safe separation on \Cref{alg:exsubtw:safesepcheckcombine}, i.e., when the safe separation $(A,S,B)$ has either order $\ge q$ or $\allw_{\ins}$ intersects only one of $A$ and $B$, does not increase the measure.

\begin{lemma}
\label{lem:al2potsafesep}
Let $\ins = (G, \{W_1, \ldots, W_t\}, k, \tc, q)$ be an instance with $t \ge 2$.
Let $(A,S,B)$ be a safe separation with $|S| \ge q$ or $\allw_{\ins} \subseteq A \cup S$ or $\allw_{\ins} \subseteq B \cup S$.
Then, $\potl(\ins \rescliqs (A,S)) \le \potl(\ins)$.
\end{lemma}
\begin{proof}
First, if $\allw_{\ins} \subseteq B \cup S$, then $\ins \rescliqs (A,S)$ has only one terminal clique and $\potl(\ins \rescliqs (A,S)) = 1$.
Then, if $\allw_{\ins} \subseteq A \cup S$, there is a terminal clique $S' \supseteq S$ with $S' \subseteq A \cup S$ because $(A,S,B)$ is a safe separator, and therefore we have that $\{W_1, \ldots, W_t\} \rescliqs (A,S) \subseteq \{W_1, \ldots, W_t\}$, and moreover all terminal cliques $W_i \in \{W_1, \ldots, W_t\} \setminus \{W_1, \ldots, W_t\} \rescliqs (A,S)$ are subsets of $S$ and thus have size $|W_i| \le |S| \le |S'|$ and the original terminal vertices mapped to them get mapped to $S'$ in $\ins \rescliqs (A,S)$, implying that $\potl(\ins \rescliqs (A,S)) \le \potl(\ins)$.
Then, if $|S| \ge q$, this follows from \Cref{lem:al2potsep}.
\end{proof}

We then show that terminal clique merging decreases the measure by at least $k$.

\begin{lemma}
\label{lem:al2potmerg}
Let $\ins = (G, \{W_1, \ldots, W_t\}, k, \tc, q)$ be an instance with $t \ge 2$.
It holds that $\potl(\ins \mergcliqs (W_i,W_j)) \le \potl(\ins)-k$.
\end{lemma}
\begin{proof}
Observe that $\potl(\ins \mergcliqs (W_i,W_j)) \le \potl(\ins) + \potl_{\ins \mergcliqs (W_i,W_j)}(W_i \cup W_j) - \potl_{\ins}(W_i) - \potl_{\ins}(W_j) + k$ (where the $+k$ comes from the fact that $\cq(\ins \mergcliqs (W_i,W_j)) = \cq(\ins)-1$ might hold).
Therefore, $\potl(\ins \mergcliqs (W_i,W_j)) \le \potl(\ins)-k$ holds when $q \ge \ql$ because in that case $\potl_{\ins \mergcliqs (W_i,W_j)}(W_i \cup W_j) \le 6k$ and $\potl_{\ins}(W_i) + \potl_{\ins}(W_j) \ge 8k$.

When, $q < \ql$, first if $W_i \cup W_j \in \{W_1, \ldots, W_t\}$, we just map more original terminal vertices into a larger terminal clique and decrease the measure by at least $12k$.
In the other case, we have that $\ctc_{\ins \mergcliqs (W_i,W_j)}(W_i \cup W_j) = \ctc_{\ins}(W_i) + \ctc_{\ins}(W_j)$, which by $|W_i \cup W_j| \ge \max(|W_i|, |W_j|)$ implies that $\potl_{\ins}(W_i) + \potl_{\ins}(W_j) - \potl_{\ins \mergcliqs (W_i,W_j)}(W_i \cup W_j) \ge 6k$, implying $\potl(\ins \mergcliqs (W_i,W_j)) \le \potl(\ins)-5k$.
\end{proof}

We then show that increasing the value of $q$ decreases the measure by at least $k$.

\begin{lemma}
\label{lem:al2potincq}
Let $\ins_1 = (G, \{W_1, \ldots, W_t\}, k, \tc, q)$ and $\ins_2 = (G, \{W_1, \ldots, W_t\}, k, \tc, q+1)$ and $t \ge 2$.
It holds that $\potl(\ins_2) \le \potl(\ins_1)-k$.
\end{lemma}
\begin{proof}
Observe that the measures of each terminal cliques do not decrease, in particular, if $q+1 \ge \ql$ and $q < \ql$, then the measure goes from at least $6k$ to at most $6k$.
Then, the measure of the instance decreases $3k$ from the term $(k+2-q)\cdot 3k$ and increases by at most $2k$ from the term $(2-\cq(\ins)) \cdot k$.
\end{proof}

We then show that if the instance has less than $2$ terminal cliques of size $\ge q$, then increasing the size of a terminal clique from less than $q$ to at least $q$ decreases the measure by at least $k$.

\begin{lemma}
\label{lem:al2potleafpush1}
Let $\ins = (G, \{W_1, \ldots, W_t\}, k, \tc, q)$ be an instance with $t \ge 2$, $W_i$ a terminal clique of $\ins$, and $S \subseteq V(G)$.
If $\cq(\ins)<2$, $|W_i| < q$, and $|W_i \cup S| \ge q$, then $\potl(\ins \addv (W_i, S)) \le \potl(\ins)-k$.
\end{lemma}
\begin{proof}
If $W_i \cup S \in \{W_1, \ldots, W_t\}$, then $\ins \addv (W_i, S) = \ins \mergcliqs (W_i, W_i \cup S)$ and this holds by \Cref{lem:al2potmerg}.
Otherwise, observe that increasing the size of a terminal clique (while keeping the mapping $\tc$ same) cannot decrease its measure, and therefore as $\cq(\ins \addv (W_i, S)) \ge \cq(\ins)+1$, it holds that $\potl(\ins \addv (W_i, S)) \le \potl(\ins)-k$.
\end{proof}

We then argue how the measure behaves when we break the instance by a separation $(A,S,B)$ of order $q$ and $\ins \rescliqs (B,S)$ has the same number of terminal cliques as $\ins$.
In particular, this corresponds to \Cref{alg:exsubtw:leafpush2rec} of \Cref{alg:exsubtw} when $|t_L| = 1$.
This lemma is the main motivation of the somewhat involved definition of the measure and the definition of $\ql$.

\begin{lemma}
\label{lem:al2potleafpush2}
Let $\ins = (G, \{W_1, \ldots, W_t\}, k, \tc, q)$ be an instance with $t \ge 2$.
Let $(A,S,B)$ be a separation so that $|S| \ge q$ and there is a terminal clique $W_i$ with $W_i \subseteq A \cup S$, $|W_i| < q$, and $\ctc_{\ins}(W_i) > k + 1 - q$.
It holds that $\potl(\ins \rescliqs (B, S)) \le \potl(\ins) - \min(k, (q - |W_i|) \log_2 k)$.
\end{lemma}
\begin{proof}
First, if $\ins \rescliqs (B,S)$ has less terminal cliques than $\ins$, then $\potl(\ins \rescliqs (B,S)) \le \potl(\ins) - k$ by \Cref{lem:al2potsep}.
%First, if there is some other terminal clique than $W_i$ that is a subset of $A \cup S$, we have that $\ins \rescliqs (B,S)$ has less terminal cliques than $\ins$, and by \Cref{lem:al2potsep}, $\potl(\ins \rescliqs (B,S)) \le \potl(\ins) - k$.
Then, we assume that $W_i$ is the only terminal clique that is a subset of $A \cup S$ and $\ins \rescliqs (B,S)$ has the same number of terminal cliques as $\ins$.
In this case, because $\cq(\ins \rescliqs (B,S)) \ge \cq(\ins)$, we have $\potl(\ins \rescliqs (B, S)) \le \potl(\ins) + \potl_{\ins \rescliqs (B, S)}(S) - \potl_{\ins}(W_i)$.
We also have that $\ctc_{\ins \rescliqs (B,S)}(S) = \ctc_{\ins}(W_i)$.
We consider the cases $q < \ql$ and $q \ge \ql$.

First, when $q < \ql$
\[\ctc_{\ins}(W_i) > k + 1 - \ql > k + 1 - (k+2 - k/\log_2 k) > k/\log_2 k - 1 \ge \log_2 k,\]
where the last inequality follows from $k \ge \kb$.
This implies that
\begin{align}
\potl(\ins \rescliqs (B,S)) \le& \potl(\ins) + (k+2-q) \cdot \ctc_{\ins \rescliqs (B,S)}(S) - (k+2-|W_i|) \cdot \ctc_{\ins}(W_i)\nonumber\\
\le& \potl(\ins) - (q-|W_i|) \cdot \ctc_{\ins}(W_i) \le \potl(\ins) - (q-|W_i|) \cdot \log_2 k.\nonumber
\end{align}

Then, consider the case when $q \ge \ql$.
If $|W_i| < \ql$, then $\potl_{\ins \rescliqs (B, S)}(S) \le 5k$ and $\potl_{\ins}(W_i) = 6k$, implying that $\potl(\ins \rescliqs (B, S)) \le \potl(\ins) - k$.
Then, if $|W_i| \ge \ql$,
\[
\potl(\ins \rescliqs (B,S)) \le \potl(\ins) + (k+2-q) \cdot \log_2 k - (k+2-|W_i|) \cdot \log_2 k \le \potl(\ins) - (q-|W_i|) \cdot \log_2 k.\]
\end{proof}

We then put the lemmas together to prove the running time of \Cref{alg:exsubtw}.

\begin{lemma}
\label{lem:al2tcmain}
\Cref{alg:exsubtw} works in time $2^{\OO(k^2)} nm$.
\end{lemma}
\begin{proof}
First we observe that all of the operations in a single call of the recursive procedure can be performed in $2^{\OO(k)} m'$ time, where $m'$ is the number of edges in the instance given to the recursive call.
In particular, the case analysis of \Cref{alg:exsubtw:bf} can be implemented in $k^{\OO(1)} m'$ time by \Cref{lem:smallcases}, safe separations can be found in $k^{\OO(1)} m'$ time by \Cref{lem:safeseptime}, the terminal clique merging of \Cref{alg:exsubtw:forcliqupairs,alg:exsubtw:mergrec,alg:exsubtw:retmc} can be implemented in $k^{\OO(1)}$ time, the branching on \Cref{alg:exsubtw:caseless,alg:exsubtw:itleafpushclique,alg:exsubtw:itleafpushw,alg:exsubtw:itleafpushimps,alg:exsubtw:itleafpushif,alg:exsubtw:leafpush1rec,alg:exsubtw:retadd1} when there are less than $2$ terminal cliques of size $\ge q$ can be implemented in $k^{\OO(1)} 4^k m' = 2^{\OO(k)} m'$ time by \Cref{lem:impsep4k}, and also the branching on \Cref{alg:exsubtw:casemore,alg:exsubtw:bbfor,alg:exsubtw:impsepfor,alg:exsubtw:asb,alg:exsubtw:degencheck,alg:exsubtw:leafpush2rec,alg:exsubtw:retadd2} when there are at least $2$ terminal cliques of size $\ge q$ can be implemented in $k^{\OO(1)} 2^t 4^k m' = 2^{\OO(k)} m'$ time.

By the definition of $\ins \rescliqs (A,S)$, observe that at each recursive call the current graph can be obtained from an induced subgraph of the original graph by adding all edges inside the terminal cliques, and therefore we can bound $m' \le k^{\OO(1)} m$, where $m$ is the number of original edges.
Therefore, the running time of the algorithm can be bounded by $2^{\OO(k)} m \cdot R(\ins)$, where $R(\ins)$ is the total number of recursive calls.

We show by induction that the number of recursive calls is bounded by
\[R(\ins) \le \vertsz(\ins) \cdot \tcbasea^{\potl(\ins)} \le 2^{\OO(k^2)}n,\]
which then gives the conclusion since $2^{\OO(k)} m \cdot 2^{\OO(k^2)}n = 2^{\OO(k^2)} nm$.

First, when the algorithm returns from the case analysis of \Cref{alg:exsubtw:bf}, this holds because $\vertsz(\ins) \ge 1$ and $\potl(\ins) \ge 1$ always.
Then we can assume that $t \ge 2$ and $|V(G)| \ge k+3$.
If there exists a safe separation $(A,S,B)$, then the number of recursive calls is
\begin{align}
R(\ins) =& 1+R(\ins \rescliqs (A,S)) + R(\ins \rescliqs (B,S)) \nonumber\\
\le& 1 + (\vertsz(\ins \rescliqs (A,S)) + \vertsz(\ins \rescliqs (B,S))) \cdot \tcbasea^{\potl(\ins)} && \text{by \Cref{lem:al2potsafesep} and induction} \nonumber\\
\le& \vertsz(\ins) \cdot \tcbasea^{\potl(\ins)}. && \text{by \Cref{lem:gs_sep}} \nonumber
\end{align}

Now, let $R_1(\ins)$ denote the total number of calls in the recursion trees from terminal clique merging on \Cref{alg:exsubtw:mergrec}.
By \Cref{lem:al2potmerg}, induction, and the fact that $k \ge \kb$, we have that 
\[R_1(\ins) \le (k+2)^2 \cdot \vertsz(\ins) \cdot \tcbasea^{\potl(\ins)-k} \le \vertsz(\ins) \cdot \tcbasea^{\potl(\ins)}/5.\]

Then, let $R_2(\ins)$ denote the total number of calls in the recursion tree from the final \Cref{alg:exsubtw:lastret} where $q$ is incremented.
By induction and \Cref{lem:al2potincq}, we have that 
\[R_2(\ins) \le \vertsz(\ins) \cdot \tcbasea^{\potl(\ins)-k} \le \vertsz(\ins) \cdot \tcbasea^{\potl(\ins)}/5.\]

Now, consider the case when there are less than two terminal cliques of size $\ge q$, and let $R_3(\ins)$ denote the total number of calls in the recursion tree from \Cref{alg:exsubtw:leafpush1rec} where the leaf pushing branching is done.
By induction, \Cref{lem:impsep4k}, \Cref{lem:al2potleafpush1}, and $k \ge \kb$, we have that
\[R_3(\ins) \le (k+2)^2 \cdot 4^k \cdot \vertsz(\ins) \cdot \tcbasea^{\potl(\ins)-k} \le \vertsz(\ins) \cdot \tcbasea^{\potl(\ins)}/5.\]
This finishes the running time analysis in the case when there are less than $2$ terminal cliques, as in this case we have that 
\[R(\ins) \le 1 + R_1(\ins) + R_2(\ins) + R_3(\ins) \le \vertsz(\ins) \cdot \tcbasea^{\potl(\ins)}.\]

It remains to consider the case when there are at least two terminal cliques of size $\ge q$.
Let $R_4(\ins)$ denote the total number of calls in the recursion tree from \Cref{alg:exsubtw:leafpush2rec} when $|t_L| \ge 2$ and $R_5(\ins)$ the total number of calls when $|t_L| = 1$.
Recall that in all cases $|t_R| \ge 2$.

First, let $|t_L| \ge 2$ and consider a single call from \Cref{alg:exsubtw:leafpush2rec}.
As $\allw_{\ins}[t_L] \subseteq A \cup S$ and $\allw_{\ins}[t_R] \subseteq B \cup S$, we have that both $\ins \rescliqs (A,S)$ and $\ins \rescliqs (B,S)$ have less terminal cliques than $\ins$, and therefore by induction, \Cref{lem:al2potsep}, and \Cref{lem:gs_sep} the number of calls for fixed $(A,S,B)$ is at most
\[\vertsz(\ins) \cdot (\tcbasea^{\potl(\ins \rescliqs (A,S))} + \tcbasea^{\potl(\ins \rescliqs (B,S))}) \le 2 \cdot \vertsz(\ins) \cdot \tcbasea^{\potl(\ins)-k}.\]
Then, the total number of calls from \Cref{alg:exsubtw:leafpush2rec}
in this case is
\[R_4(\ins) \le 2^t \cdot 4^k \cdot 2 \cdot \vertsz(\ins) \cdot \tcbasea^{\potl(\ins)-k} \le \vertsz(\ins) \cdot \tcbasea^{\potl(\ins)}/5.\]

Now, let $|t_L| = 1$ and consider a single call from \Cref{alg:exsubtw:leafpush2rec}.
As $|t_R| \ge 2$, we again have by \Cref{lem:al2potsep} that $\potl(\ins \rescliqs (A,S)) \le \potl(\ins)-k$.
Let $W_i$ be the single terminal clique with $i \in t_L$.
We have by \Cref{lem:al2potleafpush2} that $\potl(\ins \rescliqs (B,S)) \le \potl(\ins) - \min(k, (q-|W_i|) \log_2 k)$.
%Observe that $\min(k, (q-|W_i|) \log_2 k) \ge \log_2 k$.

Because $\ins$ has no safe separations, $|W_i| < q$, and there is a terminal clique of size $\ge q$, we have that $\flow_G(W_i, \ow_{\ins}(W_i)) = |W_i|$, which implies by \Cref{lem:impsepfk} that the number of important $(W_i, \ow_{\ins}(W_i))$-separators of size $q$ is at most $k^{q-|W_i|}$.
Combining with \Cref{lem:impsep4k} and the fact that $q \le k$ here, we actually obtain an upper bound of 
\[\min(k^{q-|W_i|}, 4^k) \le 4^{\min(k, (q-|W_i|) \log_2 k)}.\]
Now, for fixed $W_i$, by induction and \Cref{lem:gs_sep} the number of recursive calls is
\begin{align}
R_5(\ins, i) &\le 4^{\min(k, (q-|W_i|) \log_2 k)} \cdot \vertsz(\ins) \cdot (\tcbasea^{\potl(\ins)-\min(k, (q-|W_i|) \log_2 k)} + \tcbasea^{\potl(\ins)-k}) \nonumber\\
&\le \vertsz(\ins) \cdot 4^{\min(k, (q-|W_i|) \log_2 k)} \cdot 2 \cdot \tcbasea^{\potl(\ins)} \cdot \tcbasea^{-\min(k, (q-|W_i|) \log_2 k)} \nonumber\\
&\le \vertsz(\ins) \cdot 2 \cdot \tcbasea^{\potl(\ins)} \cdot (1/4)^{\min(k, (q-|W_i|) \log_2 k)}\nonumber\\
&\le \vertsz(\ins) \cdot 2 \cdot \tcbasea^{\potl(\ins)} / 20k \le \vertsz(\ins) \cdot \tcbasea^{\potl(\ins)} / 10k. \nonumber
\end{align}
Then, over all terminal cliques this is 
\[R_5(\ins) = \sum_{i=1}^t R_5(\ins, i) \le (k+2) \cdot \vertsz(\ins) \cdot \tcbasea^{\potl(\ins)}/10k \le \vertsz(\ins) \cdot \tcbasea^{\potl(\ins)}/5.\]
This finishes the running time analysis of the case when there are at least 2 terminal cliques of size $\ge q$, as in this case we have that 
\[R(\ins) \le 1 + R_1(\ins) + R_2(\ins) + R_4(\ins) + R_5(\ins) \le \vertsz(\ins) \cdot \tcbasea^{\potl(\ins)}.\]
\end{proof}

Putting together with the pre-branching of \Cref{lem:prebranching}, this finishes the proof of \Cref{the:stwexalg}, and together with \Cref{the:stweximpl} they imply \Cref{the:mainexact}.

\section{Conclusion}
\label{sec:conclusion}
We have given a $2^{\OO(k^2)} n^4$ time algorithm for deciding if a given graph has treewidth at most $k$ and computing the corresponding tree decomposition, and also a $k^{\OO(k/\varepsilon)} n^4$ time $(1+\varepsilon)$-approximation algorithm for the same problem.
In this section we conclude by explicitly stating some results that were implicitly proven in \Cref{sec:redu} and by asking some open questions.

Recall that \Cref{the:stweximpl} gave a connection between \stw and treewidth, showing that algorithms for \stw imply algorithms for treewidth.
The statement of \Cref{the:stweximpl} is a bit technical because we paid attention to the factors polynomial in $n$, so let us note here that \Cref{lem:main_impr_lem} used together with iterative compression implies the following elegant connection between treewidth and \stw.

\begin{proposition}
For any function $f : \mathbb{N} \rightarrow \mathbb{N}$, there is an $f(k) \cdot n^{\OO(1)}$ time algorithm for treewidth if and only if there is an $f(k) \cdot n^{\OO(1)}$ time algorithm for \stw.
\end{proposition}

Here, the if-direction comes from applying \Cref{lem:main_impr_lem} iteratively together with iterative compression, and the only-if direction comes from using that our statement of \stw allows to conclude that the treewidth of $G$ is more than $k$, which implies that any algorithm for treewidth is trivially also an algorithm for \stw.

The techniques used in \Cref{sec:redu} generalize techniques that were used to prove the existence of lean tree decompositions of width at most the treewidth of the graph~\cite{bellenbaum2002two,DBLP:journals/jct/Thomas90}.
Therefore, it is natural to ask if our proofs give some generalization of the notion of lean tree decompositions, and it turns out that they indeed do.
A tree decomposition $(T, \bag)$ is called \emph{lean} if for any two nodes $t_1, t_2 \in V(T)$ and sets $Z_1 \subseteq \bag(t_1)$ and $Z_2 \subseteq \bag(t_2)$ with $|Z_1| = |Z_2|$, either $Z_1$ is linked into $Z_2$, or on the unique $t_1$--$t_2$-path in $T$ there is an edge $xy$ with $|\bag(x) \cap \bag(y)| < |Z_1|$.
By setting $t_1 = t_2$ we get that any two subsets of a same bag are linked into each other, which implies that any lean tree decomposition has width at most $3k+2$, where $k$ is the treewidth of the graph.

For our generalization of lean tree decompositions, let us define that a node $t \in V(T)$ is \emph{strongly lean} if there does not exist a torso tree decomposition of width $|\bag(t)|-2$ that covers $\bag(t)$.
Note that a node $t$ that is strongly lean must satisfy the leanness condition for $t_1 = t_2 = t$.
Then, we say that a tree decomposition is strongly lean if it is lean, and its every node is strongly lean.
In \Cref{sec:appendixlean} we observe that techniques in \Cref{sec:redu} imply the existence of strongly lean tree decompositions.

\begin{restatable}{proposition}{strongleanproposition}
\label{pro:stronglean}
Every graph admits a strongly lean tree decomposition.
\end{restatable}

Note that every strongly lean tree decomposition must have width at most the treewidth of the graph, so \Cref{pro:stronglean} is a generalization of the result of Thomas~\cite{DBLP:journals/jct/Thomas90}.

We then turn to open questions.
The main open question that we made progress towards in this work is the question of what is the asymptotically smallest function $f(k)$ so that treewidth can be computed in time $f(k) \cdot n^{\OO(1)}$.
Even after our improvement from $2^{\OO(k^3)}$ to $2^{\OO(k^2)}$, this question remains wide open.
On the lower bound side, we are not aware of any lower bounds under the Exponential Time Hypothesis (ETH)~\cite{ImpagliazzoP01} for treewidth, apart from the lower bound that follows from the NP-hardness proof~\cite{ArnborgCP87}.
By tracing the chain of reductions in the proof, we observe in \Cref{sec:appendixhard} that they imply the following lower bound under ETH.

\begin{restatable}{proposition}{ethboundproposition}
\label{pro:ethbound}
Assuming ETH, there is no $2^{o(\sqrt{n})}$ time algorithm for computing the treewidth on graphs with $n$ vertices.
\end{restatable}

This of course implies also a $2^{o(\sqrt{k})}$ lower bound on the dependence on $k$, but nothing better since the graph produced by the reduction is co-bipartite.
It would be very surprising if there would be a subexponential time algorithm for treewidth, so on the lower bound side we ask if the lower bound of \Cref{pro:ethbound} can be improved to $2^{o(n)}$, which would be tight.
Proving a $2^{o(k)}$ lower bound for the parameterized running time could be easier, but we do not conjecture where between $2^{\OO(k)}$ and $2^{\OO(k^2)}$ the right answer for the dependence on $k$ should be.

On the approximation side, we ask if treewidth can be $(1+\varepsilon)$-approximated in time $2^{\OO(k)} n^{\OO(1)}$ for every fixed $\varepsilon > 0$.
Even obtaining a $1.9$-approximation in time $2^{\OO(k)} n^{\OO(1)}$ would seem to require new techniques, as well as the related problem of solving \pstw in time $2^{\OO(k)} n^{\OO(1)}$ when $t = 2$.

Finally, let us discuss the dependence on $n$, which is $n^4$ for both of our algorithms.
It can be observed that the bottleneck is in the polynomial-time process of iteratively improving the torso tree decomposition $(X, (T_X, \bag_X))$ in \Cref{subsec:itimprtd}.
In particular, the running times in \Cref{the:mainexact,the:mainapx} could be stated in a tighter way as $2^{\OO(k^2)} n^3 + k^{\OO(1)} n^4$ and $k^{\OO(k/\varepsilon)} n^3 + k^{\OO(1)} n^4$, respectively.
We believe that improving the dependence on $n$ significantly is an interesting problem.
Improving it below $n^3$ should require new techniques, and below $n^2$ could need a completely new approach.

\bibliographystyle{alpha}
\bibliography{book_kernels_fvf}

\newcommand{\etalchar}[1]{$^{#1}$}
\begin{thebibliography}{WAPL14}

\bibitem[ACP87]{ArnborgCP87}
Stefan Arnborg, Derek~G. Corneil, and Andrzej Proskurowski.
\newblock Complexity of finding embeddings in a $k$-tree.
\newblock {\em SIAM J. Alg. Disc. Meth.}, 8:277--284, 1987.

\bibitem[Ami01]{Amir01}
E.~Amir.
\newblock Efficient approximation for triangulation of minimum treewidth.
\newblock In {\em Uncertainty in Artificial Intelligence: Proceedings of the
  Seventeenth Conference (UAI-2001), San Francisco, CA}, pages 7--15. Morgan
  Kaufmann Publishers, 2001.

\bibitem[Ami10]{DBLP:journals/algorithmica/Amir10}
Eyal Amir.
\newblock Approximation algorithms for treewidth.
\newblock {\em Algorithmica}, 56(4):448--479, 2010.

\bibitem[BB72]{BerteleF72}
Umberto Bertel{\`e} and Francesco Brioschi.
\newblock {\em Nonserial dynamic programming}.
\newblock Academic Press, New York, 1972.
\newblock Mathematics in Science and Engineering, Vol. 91.

\bibitem[BCC{\etalchar{+}}06]{bodlaender2006open}
Hans~L Bodlaender, Leizhen Cai, Jianer Chen, Michael~R. Fellows, Jan~Arne
  Telle, and D{\'a}niel Marx.
\newblock Open problems in parameterized and exact computation -- {IWPEC} 2006.
\newblock Technical Report UU-CS-2006-052, Department of Information and
  Computing Sciences, Utrecht University, 2006.

\bibitem[BD02]{bellenbaum2002two}
Patrick Bellenbaum and Reinhard Diestel.
\newblock Two short proofs concerning tree-decompositions.
\newblock {\em Combinatorics, Probability and Computing}, 11(6):541--547, 2002.

\bibitem[BDD{\etalchar{+}}16]{BodlanderDDFLP13}
Hans~L. Bodlaender, P{\aa}l~Gr{\o}n{\aa}s Drange, Markus~S. Dregi, Fedor~V.
  Fomin, Daniel Lokshtanov, and Micha\l{} Pilipczuk.
\newblock A $c^k n$ $5$-approximation algorithm for treewidth.
\newblock {\em SIAM J. Computing}, 45(2):317--378, 2016.

\bibitem[BF21]{BelbasiF21}
Mahdi Belbasi and Martin F{\"{u}}rer.
\newblock Finding all leftmost separators of size $\leq k$.
\newblock In {\em Proceedings of 15th International Conference on Combinatorial
  Optimization and Applications (COCOA)}, volume 13135 of {\em Lecture Notes in
  Comput. Sci.}, pages 273--287. Springer, 2021.

\bibitem[BF22]{BelbasiF22}
Mahdi Belbasi and Martin F{\"{u}}rer.
\newblock An improvement of {Reed's} treewidth approximation.
\newblock {\em J. Graph Algorithms Appl.}, 26(2):257--282, 2022.

\bibitem[BGHK95]{BodlaenderGHK95}
Hans~L. Bodlaender, John~R. Gilbert, Hj{\'{a}}lmtyr Hafsteinsson, and Ton
  Kloks.
\newblock Approximating treewidth, pathwidth, frontsize, and shortest
  elimination tree.
\newblock {\em J. Algorithms}, 18(2):238--255, 1995.

\bibitem[BJT21]{bodlaenderrevisited}
Hans~L. Bodlaender, Lars Jaffke, and Jan~Arne Telle.
\newblock Typical sequences revisited - computing width parameters of graphs.
\newblock {\em Theory Comput. Syst.}, 2021.

\bibitem[BK91]{DBLP:conf/icalp/BodlaenderK91}
Hans~L. Bodlaender and Ton Kloks.
\newblock Better algorithms for the pathwidth and treewidth of graphs.
\newblock In {\em Proceedings of the 18th International Colloquium of Automata,
  Languages and Programming (ICALP)}, volume 510 of {\em Lecture Notes in
  Comput. Sci.}, pages 544--555. Springer, 1991.

\bibitem[BK96]{BodlaenderK96}
Hans~L. Bodlaender and Ton Kloks.
\newblock Efficient and constructive algorithms for the pathwidth and treewidth
  of graphs.
\newblock {\em J. Algorithms}, 21(2):358--402, 1996.

\bibitem[BK06]{BodlaenderK06}
Hans~L. Bodlaender and Arie M. C.~A. Koster.
\newblock Safe separators for treewidth.
\newblock {\em Discret. Math.}, 306(3):337--350, 2006.

\bibitem[Bod93a]{Bodlaender93}
H.~L. Bodlaender.
\newblock A tourist guide through treewidth.
\newblock {\em Acta Cybernet.}, 11(1-2):1--21, 1993.

\bibitem[Bod93b]{DBLP:conf/stoc/Bodlaender93}
Hans~L. Bodlaender.
\newblock A linear time algorithm for finding tree-decompositions of small
  treewidth.
\newblock In {\em Proceedings of the 25th Annual ACM Symposium on Theory of
  Computing (STOC)}, pages 226--234. {ACM}, 1993.

\bibitem[Bod94]{DBLP:journals/dam/Bodlaender94}
Hans~L. Bodlaender.
\newblock Improved self-reduction algorithms for graphs with bounded treewidth.
\newblock {\em Discret. Appl. Math.}, 54(2-3):101--115, 1994.

\bibitem[Bod96]{Bodlaender96}
Hans~L. Bodlaender.
\newblock A linear-time algorithm for finding tree-decompositions of small
  treewidth.
\newblock {\em SIAM J. Computing}, 25(6):1305--1317, 1996.

\bibitem[BPT92]{BoriePT92}
Richard~B. Borie, R.~Gary Parker, and Craig~A. Tovey.
\newblock Automatic generation of linear-time algorithms from predicate
  calculus descriptions of problems on recursively constructed graph families.
\newblock {\em Algorithmica}, 7(5{\&}6):555--581, 1992.

\bibitem[CFK{\etalchar{+}}15]{cygan2015parameterized}
Marek Cygan, Fedor~V. Fomin, Lukasz Kowalik, Daniel Lokshtanov, D{\'a}niel
  Marx, Marcin Pilipczuk, Micha{\l} Pilipczuk, and Saket Saurabh.
\newblock {\em Parameterized Algorithms}.
\newblock Springer, 2015.

\bibitem[CJ03]{CaiJ03}
Liming Cai and David~W. Juedes.
\newblock On the existence of subexponential parameterized algorithms.
\newblock {\em J. Computer and System Sciences}, 67(4):789--807, 2003.

\bibitem[CLL09]{DBLP:journals/algorithmica/ChenLL09}
Jianer Chen, Yang Liu, and Songjian Lu.
\newblock An improved parameterized algorithm for the minimum node multiway cut
  problem.
\newblock {\em Algorithmica}, 55(1):1--13, 2009.

\bibitem[Cou90]{Courcelle90}
Bruno Courcelle.
\newblock The monadic second-order logic of graphs {I}: {R}ecognizable sets of
  finite graphs.
\newblock {\em Information and Computation}, 85:12--75, 1990.

\bibitem[DF99]{DowneyF99}
Rodney~G. Downey and Michael~R. Fellows.
\newblock {\em Parameterized complexity}.
\newblock Springer-Verlag, New York, 1999.

\bibitem[DF13]{DowneyFbook13}
Rodney~G. Downey and Michael~R. Fellows.
\newblock {\em Fundamentals of Parameterized Complexity}.
\newblock Texts in Computer Science. Springer, 2013.

\bibitem[Die05]{Diestel}
Reinhard Diestel.
\newblock {\em Graph theory}, volume 173 of {\em Graduate Texts in
  Mathematics}.
\newblock Springer-Verlag, Berlin, 3rd edition, 2005.

\bibitem[EJT10]{DBLP:conf/focs/ElberfeldJT10}
Michael Elberfeld, Andreas Jakoby, and Till Tantau.
\newblock Logspace versions of the theorems of bodlaender and courcelle.
\newblock In {\em Proceedings of the 51th Annual {IEEE} Symposium on
  Foundations of Computer Science, {FOCS} 2010}, pages 143--152. {IEEE}
  Computer Society, 2010.

\bibitem[FG06]{FlumGrohebook}
J{\"o}rg Flum and Martin Grohe.
\newblock {\em Parameterized Complexity Theory}.
\newblock Texts in Theoretical Computer Science. An EATCS Series.
  Springer-Verlag, Berlin, 2006.

\bibitem[FHL08]{FeigeHL08}
Uriel Feige, MohammadTaghi Hajiaghayi, and James~R. Lee.
\newblock Improved approximation algorithms for minimum weight vertex
  separators.
\newblock {\em SIAM J. Computing}, 38(2):629--657, 2008.

\bibitem[FL89]{FellowsL89}
Michael~R. Fellows and Michael~A. Langston.
\newblock On search, decision and the efficiency of polynomial-time algorithms
  (extended abstract).
\newblock In {\em Proceedings of the 21st Annual ACM Symposium on Theory of
  Computing (STOC)}, pages 501--512. ACM, 1989.

\bibitem[FLS{\etalchar{+}}18]{DBLP:journals/talg/FominLSPW18}
Fedor~V. Fomin, Daniel Lokshtanov, Saket Saurabh, Michal Pilipczuk, and Marcin
  Wrochna.
\newblock Fully polynomial-time parameterized computations for graphs and
  matrices of low treewidth.
\newblock {\em {ACM} Transactions on Algorithms}, 14(3):34:1--34:45, 2018.

\bibitem[FLSZ19]{fomin2019kernelization}
Fedor~V Fomin, Daniel Lokshtanov, Saket Saurabh, and Meirav Zehavi.
\newblock {\em Kernelization: {Theory} of parameterized preprocessing}.
\newblock Cambridge University Press, 2019.

\bibitem[FTV15]{DBLP:journals/siamcomp/FominTV15}
Fedor~V. Fomin, Ioan Todinca, and Yngve Villanger.
\newblock Large induced subgraphs via triangulations and {CMSO}.
\newblock {\em {SIAM} Journal on Computing}, 44(1):54--87, 2015.

\bibitem[Gav77]{gavril1977some}
Fanica Gavril.
\newblock Some {{NP}}-complete problems on graphs.
\newblock In {\em Proceedings of the 1977 Conference on Information Sciences
  and Systems, The Johns Hopkins University, Baltimore, Maryland}, pages
  91--95, 1977.

\bibitem[Hal76]{Halin:1976it}
Rudolf Halin.
\newblock {$S$}-functions for graphs.
\newblock {\em J. Geometry}, 8(1-2):171--186, 1976.

\bibitem[IP01]{ImpagliazzoP01}
Russell Impagliazzo and Ramamohan Paturi.
\newblock On the complexity of $k$-{SAT}.
\newblock {\em J. Computer and System Sciences}, 62(2):367--375, 2001.

\bibitem[Kor21]{Korhonen21}
Tuukka Korhonen.
\newblock A single-exponential time 2-approximation algorithm for treewidth.
\newblock In {\em Proceedings of the 62nd Annual Symposium on Foundations of
  Computer Science (FOCS)}, pages 184--192. {IEEE}, 2021.

\bibitem[LA91]{DBLP:conf/icalp/LagergrenA91}
Jens Lagergren and Stefan Arnborg.
\newblock Finding minimal forbidden minors using a finite congruence.
\newblock In {\em Proceedings of the 18th International Colloquium of Automata,
  Languages and Programming (ICALP)}, volume 510 of {\em Lecture Notes in
  Comput. Sci.}, pages 532--543. Springer, 1991.

\bibitem[Lag96]{DBLP:journals/jal/Lagergren96}
Jens Lagergren.
\newblock Efficient parallel algorithms for graphs of bounded tree-width.
\newblock {\em Journal of Algorithms}, 20(1):20--44, 1996.

\bibitem[Mar06]{DBLP:journals/tcs/Marx06}
D{\'a}niel Marx.
\newblock Parameterized graph separation problems.
\newblock {\em Theoretical Computer Science}, 351(3):394--406, 2006.

\bibitem[MR14]{marx-razgon-stoc2011-multicut}
D{\'{a}}niel Marx and Igor Razgon.
\newblock Fixed-parameter tractability of multicut parameterized by the size of
  the cutset.
\newblock {\em {SIAM} J. Comput.}, 43(2):355--388, 2014.

\bibitem[MT91]{DBLP:journals/jal/MatousekT91}
Jir{\'{\i}} Matousek and Robin Thomas.
\newblock Algorithms finding tree-decompositions of graphs.
\newblock {\em J. Algorithms}, 12(1):1--22, 1991.

\bibitem[Nie06]{Niedermeierbook06}
Rolf Niedermeier.
\newblock {\em Invitation to fixed-parameter algorithms}, volume~31 of {\em
  Oxford Lecture Series in Mathematics and its Applications}.
\newblock Oxford University Press, Oxford, 2006.

\bibitem[Ree92]{DBLP:conf/stoc/Reed92}
Bruce~A. Reed.
\newblock Finding approximate separators and computing tree width quickly.
\newblock In S.~Rao Kosaraju, Mike Fellows, Avi Wigderson, and John~A. Ellis,
  editors, {\em Proceedings of the 24th Annual {ACM} Symposium on Theory of
  Computing, STOC 1992}, pages 221--228. {ACM}, 1992.

\bibitem[RS84]{RobertsonS3}
Neil Robertson and Paul~D. Seymour.
\newblock Graph minors. {III}. {P}lanar tree-width.
\newblock {\em J. Combinatorial Theory Ser. B}, 36:49--64, 1984.

\bibitem[RS95]{RobertsonS-GMXIII}
Neil Robertson and Paul~D. Seymour.
\newblock Graph minors. {XIII}. {T}he disjoint paths problem.
\newblock {\em J. Combinatorial Theory Ser. B}, 63(1):65--110, 1995.

\bibitem[RS04]{RobertsonS04}
Neil Robertson and Paul~D. Seymour.
\newblock Graph minors. {XX}. {Wagner's} conjecture.
\newblock {\em J. Combinatorial Theory Ser. B}, 92(2):325--357, 2004.

\bibitem[Tho90]{DBLP:journals/jct/Thomas90}
Robin Thomas.
\newblock A menger-like property of tree-width: The finite case.
\newblock {\em J. Combinatorial Theory Ser. B}, 48(1):67--76, 1990.

\bibitem[WAPL14]{WuAPL14}
Yu~Wu, Per Austrin, Toniann Pitassi, and David Liu.
\newblock Inapproximability of treewidth and related problems.
\newblock {\em J. Artif. Intell. Res.}, 49:569--600, 2014.

\end{thebibliography}

\appendix
\newpage
\section{Strongly lean tree decompositions}
\label{sec:appendixlean}
We gave a definition of \emph{strongly lean} tree decompositions in \Cref{sec:conclusion}.
Here we show using the tools developed in \Cref{sec:redu} that strongly lean tree decompositions exist for all graphs.

%\begin{lemma}
%Every graph admits a strongly lean tree decomposition.
%\end{lemma}
\strongleanproposition*
\begin{proof}
Let $G$ be a graph and $(T, \bag)$ a tree decomposition of $G$ with the lexicographically minimal sequence of bag sizes, i.e., one that first minimizes the number of bags of size $n$, then the number of bags of size $n-1$, and so on.
The proof of Bellenbaum and Diestel~\cite{bellenbaum2002two} shows that such a tree decomposition is lean.
Let us prove that each node of such a tree decomposition is strongly lean.

Assume otherwise, and let $r \in V(T)$ be a node that is not strongly lean.
Denote $\bag(r) = W$, and let $(X, (T_X, \bag_X))$ be a torso tree decomposition that covers $W$ and has width at most $|W|-2$.
Now we proceed as in the proof of \Cref{lem:imprmainv2}.
First, by iterating \Cref{lem:make_d_linkedv2} we can assume that $(X, (T_X, \bag_X))$ is $d_{(T,\bag,r)}$-linked into $W$.
Then, we implement the construction of an improved tree decomposition $(T',\bag')$ as in \Cref{lem:imprmainv2}, in particular, we construct the tree decompositions $(T_C, \bag_C)$ for each connected component $C$ of $G \setminus X$ as described in the proof, and attach them to $(T_X, \bag_X)$ to construct $(T', \bag')$.

By the fact that $(X, (T_X, \bag_X))$ is $d_{(T,\bag,r)}$-linked into $W$, we know that the \Cref{lem:width_constrv2:good} of \Cref{lem:width_constrv2} always holds, and we will use it to argue that $(T',\bag')$ has lexicographically smaller sequence of bag sizes than $(T,\bag)$, contradicting the choice of $(T,\bag)$.
The idea will be that the lexicographical improvement will happen already in bags of size at least $|W|$, so we will not care about the bags of $(T_X, \bag_X)$ in the construction.
Call a node $t \in V(T)$ a 1-component node if $\bag(t)$ intersects exactly one component $C$ of $G \setminus X$.
Let $k$ be the minimum integer so that every node $t$ with $|\bag(t)| > k$ is a 1-component node.
Note that $k \ge |W|$ because $W \subseteq X$.

First, we claim that for each $s > k$, the number of bags of size at least $s$ in $(T',\bag')$ is at most the number of bags of size at least $s$ in $(T,\bag)$.
To prove this, recall that for every bag $\bag(t)$ of $(T, \bag)$ we constructed a corresponding bag $\bag_C(t)$ in $(T_C, \bag_C)$ for a component $C$ only if $\bag(t)$ intersected $C$, so for each 1-component node there exists only one corresponding node in $(T',\bag')$, and by \Cref{lem:width_constrv2}, such node has bag of size at most $|\bag_C(t)| \le |\bag(t)|$.
For nodes that are not 1-component nodes, their bags have size at most $k$, so by \Cref{lem:width_constrv2} no bag of size $s > k$ in $(T',\bag')$ can correspond to such node.

Then, we know that either $(T',\bag')$ has lexicographically smaller sequence of bag sizes than $(T,\bag)$ already on bags of size $> k$, in which case we are ready, or that each bag of size $s > k$ in $(T',\bag')$ corresponds to exactly one bag of size $s$ in $(T,\bag)$.
In the latter case, again each 1-component node of $(T,\bag)$ with a bag of size $k$ corresponds to exactly one bag of size at most $k$ of $(T',\bag')$.
However, there is at least one node of $(T,\bag)$ with a bag of size $k$ that is not a 1-component node.
Let $t$ be such a node.
First, if $\bag(t) \subseteq X$, then $\bag(t)$ does not intersect any components $C$ of $G \setminus X$, and therefore there are no bags in $(T',\bag')$ corresponding to $t$.
Second, if $\bag(t)$ intersects multiple components $C_1, \ldots, C_c$ of $G \setminus X$, it cannot hold that $\bag_{C_i}(t) = \bag(t)$ for any component $C_i$ because $\bag_{C_i}(t)$ is a subset of $N[C_i]$, so therefore by \Cref{lem:width_constrv2} it holds that $|\bag_{C_i}(t)| < |\bag(t)|$.
In particular, in both of the cases no bags of size $k$ in $(T',\bag')$ correspond to $t$, and therefore the number of bags of size $k$ in $(T',\bag')$ must be smaller than in $(T,\bag)$.
\end{proof}

\section{An ETH lower bound for treewidth}
\label{sec:appendixhard}
We detail how the NP-hardness proof of treewidth of Arnborg, Corneil, and Proskurowski~\cite{ArnborgCP87} implies that there is no $2^{o(\sqrt{n})}$ time algorithm for treewidth assuming ETH.
We start with the hardness of max-cut from~\cite{CaiJ03}.

\begin{lemma}[\cite{CaiJ03}]
Assuming ETH, there is no $2^{o(n)}$ time algorithm for max-cut on graphs with $n$ vertices.
\end{lemma}

Next we recall the definition of cutwidth (also known as minimum cut linear arrangement).
Let $G$ be an $n$-vertex graph and $\tau : V(G) \rightarrow [n]$ be an ordering of the vertices, i.e., a bijection from $V(G)$ to $[n]$.
The cutwidth of $\tau$ is $\max_{1 \le i < n} |\{uv \in E(G) \mid \tau(u) \le i < \tau(v)\}|$.
The cutwidth of $G$ is the minimum cutwidth of an ordering of the vertices of $G$.

Next we reduce max-cut to cutwidth, following the idea of the reduction of Gavril~\cite{gavril1977some}, but being careful to increase the number of vertices only by a constant factor. 

\begin{lemma}
Assuming ETH, there is no $2^{o(n)}$ time algorithm for computing the cutwidth on graphs with $n$ vertices.
\end{lemma}
\begin{proof}
We give a reduction from max-cut on graphs with $n$ vertices to cutwidth on graphs with $\OO(n)$ vertices.
Let $G$ be the input graph for max-cut with $n$ vertices.
We create a graph $G'$ by first taking the complement of $G$ and then adding $3n$ universal vertices, resulting in a graph with a total of $4n$ vertices and $\binom{4n}{2} - |E(G)|$ edges.

We claim that for all $k \le n^2$, the cutwidth of $G'$ is at most $4n^2-k$ if and only if the max-cut of $G$ is at least $k$.
First, for the only if direction, consider the point in the ordering with $2n$ vertices on the left and $2n$ vertices on the right.
All non-edges of $G'$ correspond to edges of $G$, so if there are at most $4n^2-k$ edges crossing this point, this gives a cut with at least $k$ edges of $G$.

For the if direction, let $V_1, V_2$ be a partition of $V(G)$ with at least $k$ edges between $V_1$ and $V_2$, and let $U$ denote the $3n$ added universal vertices of $G'$.
We create an ordering of the vertices of $G'$ by first letting the vertices in $V_1$ appear in an arbitrary order, then the vertices in $U$, and then the vertices in $V_2$.
The number of edges in the first $|V_1|$ cuts and the last $|V_2|$ cuts of the ordering is at most $n \cdot 3n \le 4n^2-k$.
Because the number of edges between $V_1$ and $V_2$ in $G$ is at least $k$, the number of edges in any other cut of the ordering is at most $4n^2-k$.
\end{proof}

Then, a direct application of~\cite{ArnborgCP87} finishes the chain of reductions.

%\begin{lemma}
%\label{lem:twethhardness}
%Assuming ETH, there is no $2^{o(\sqrt{n})}$ time algorithm for computing the treewidth on graphs with $n$ vertices.
%\end{lemma}
\ethboundproposition*
\begin{proof}
We observe that the reduction of Arnborg, Corneil, and Proskurowski~\cite{ArnborgCP87} reduces the cutwidth on graphs with $n$ vertices to treewidth on graphs with $\OO(n^2)$ vertices.
\end{proof}

We remark that the reduction of~\cite{ArnborgCP87} in fact reduces the cutwidth on graphs with $n$ vertices and maximum degree $\Delta$ to treewidth on graphs with $\OO(n \cdot \Delta)$ vertices.
Therefore, one possible way to improve \Cref{pro:ethbound} would be to give a $2^{o(n)}$ lower bound for computing the cutwidth of bounded-degree graphs.
\end{document}